\documentclass{birkau}

\usepackage{amsmath,amssymb,url}
\usepackage{enumitem} 
\usepackage{kantlipsum}
\usepackage{microtype}%if unwanted, comment out or use option "draft"
\usepackage{bm}
\usepackage{tikz}
\usepackage{mathrsfs}
\usepackage{mathabx}
\usepackage{comment}
\usepackage{color}
\usepackage[all]{xy}

\numberwithin{equation}{section}

\theoremstyle{plain}
\newtheorem{theorem}{Theorem}[section]
\newtheorem{lemma}[theorem]{Lemma}
\newtheorem{proposition}[theorem]{Proposition}
\newtheorem{corollary}[theorem]{Corollary}

\theoremstyle{definition}
\newtheorem{example}[theorem]{Example}
\newtheorem{definition}[theorem]{Definition}
\newtheorem{notation}[theorem]{Notation}
\newtheorem{remark}[theorem]{Remark}
 
\newcommand{\e}{\mathsf e}

\begin{document}

\title[New Topologies for clones]{\LARGE Exploring New Topologies for the Theory of Clones} 
\author[A. Bucciarelli and A. Salibra]{Antonio Bucciarelli and Antonino Salibra}
\address{Institut de Recherche en Informatique Fondamentale\\
Universit\'e Paris Cit\'e\\ 8 Place Aur\'elie Nemours, 75205 Paris Cedex 13, France}
%and\\
%Department of Environmental Sciences, Informatics and Statistics\\
%Universit\`a Ca'Foscari Venezia\\ Via Torino 155, 30173 Venezia, Italia}
%\email{salibra@unive.it}
%\urladdr{http://www.dsi.unive.it/~salibra}


% Key words and phrases

\begin{abstract}

Clones of operations of arity $\omega$  (referred to as $\omega$-operations) have been employed by Neumann \cite{neu70} to represent varieties of infinitary algebras defined by operations of at most arity $\omega$. More recently, clone algebras \cite{BS22} have been introduced to study clones  of functions, including $\omega$-operations, within the framework of one-sorted universal  algebra. Additionally, polymorphisms of arity $\omega$, which are $\omega$-operations preserving the relations of a given first-order structure, have recently been used to establish model theory results \cite{BHM12} with  applications in the field of complexity of CSP problems.

In this paper, we undertake  a topological and  algebraic  study of polymorphisms of arity $\omega$ and their corresponding invariant relations. Given a Boolean ideal $X$ on the set $A^\omega$, we endow the set of $\omega$-operations on $A$ with a topology, which we refer to as $X$-topology. Notably, the topology of pointwise convergence can be retrieved as a special case of this approach. Polymorphisms and invariant relations are then defined parametrically, with respect to the $X$-topology. We characterise the $X$-closed clones of $\omega$-operations in terms of $Pol^\omega$-$Inv^\omega$ and present a method to relate $Inv^\omega$-$Pol^\omega$ to the classical (finitary) $Inv$-$Pol$.

\end{abstract}
\maketitle
\vspace{-1.00cm}
\section{Introduction}

%Clones are sets of finitary operations that include all projections and are closed under composition, making them crucial in universal algebra. Every clone consists of the complete set of term operations within an algebra, and conversely, every clone can be represented in this form. Thus, comparing clones of algebras is more suitable for classifying algebras based on different behaviours than comparing their basic operations.
Clones are sets of finitary operations that include all projections and are closed under composition (see \cite{L06,SZ86,T93}). They play a significant role in universal algebra, as the set of all term operations of an algebra always constitutes a clone, and, in fact, every clone is of this form.
Therefore, comparing clones of algebras is much more appropriate than comparing their basic operations for the purpose of classifying algebras based on different behaviours. 
In addition to their significance in universal algebra, clones also play an important role in the study of first-order structures. The polymorphism clone of a first-order structure, containing all finitary operations that preserve the structure, holds valuable information and serves as a powerful analytical tool.
Clones are also essential in theoretical computer science, especially in the context of constraint satisfaction problems (CSPs) (e.g. see \cite{B15,BKW17,B21,B22}). In a CSP, a specific structure  is fixed (the template), and the problem involves deciding whether a given conjunction of atomic formulas over the signature of the template is satisfiable in the structure. Jeavons' groundbreaking discovery \cite{J98} states that the complexity of a CSP for a finite structure is entirely determined by the polymorphism clone of that structure.

For countable structures classified as $\omega$-categorical, the polymorphism clone carries a substantial amount of information. In \cite{BP15} primitive positive bi-interpretability of two $\omega$-categorical structures $A$ and $B$ is linked to the isomorphism of their polymorphism clones as topological clones. %This concept of primitive positive interpretations has practical implications in theoretical computer science, where the polynomial-time reduction of the CSP of a relational structure $B$ to a relational structure $A$ with a finite signature is possible when $B$ has a primitive positive interpretation in $A$.
Besides the theoretical interest they might have, primitive positive interpretations are additionally motivated by an application in theoretical computer science: when a relational structure $B$ has a primitive positive interpretation in a relational structure $A$ with a finite signature, then $\mathrm{CSP}(B)$ has a polynomial-time reduction to $\mathrm{CSP}(A)$. %It follows from the above theorem  that the computational complexity of the constraint satisfaction problem for a relational structure in a finite language only depends on its topological polymorphism clone.

The complexity of a CSP problem can be well-characterised in finite or $\omega$-categorical relational structures, 
as a relation is primitive positive definable if and only if it is preserved by all the polymorphisms.
A generalisation to infinite structures that are not necessarily $\omega$-categorical is presented in \cite{BHM12}, showing that every CSP can be formulated with a template where a relation is primitive positive definable if and only if it is first-order definable and preserved by all the polymorphisms of arity $\omega$ (hereafter called $\omega$-polymorphisms).

Efforts have been made to encode clones into one-sorted algebras, leading to concepts like abstract $\aleph_0$-clones ($\aleph_0$-AC) \cite{neu70} and clone algebras (CA) \cite{BS22,S22}. These provide algebraic theories for understanding clones and their related operations.
A crucial feature of these two approaches is connected to the role played by projections in clones. In clone algebras and abstract $\aleph_0$-clones these are abstracted out, and take the form of a countable infinite system of fundamental elements (nullary operations) $\e_0, \e_1,\ldots , \e_n, \ldots$ of the algebra.
An important consequence of the abstraction of  projections is the abstraction of functional composition, obtained  in CAs (resp. $\aleph_0$-ACs) by introducing an operator $q_n$ of arity $n+1$ for every $n\geq 0$ (resp. $q$ of arity $\omega$).
%Roughly speaking, $q_n(a,b_0,\ldots,b_{n-1})$ represents the composition of $a$ with $b_0,\ldots,b_{n-1}$ into the first $n$ coordinates. Similarly, in $\aleph_0$-ACs it is introduced an operator $q$ of arity $\omega$ such that $q(a,b_0,\ldots,b_{n-1},b_n,\dots)$ represents the composition of $a$ with $b_0,\ldots,b_{n-1},b_n,\dots$. 

The axioms of CAs and $\aleph_0$-ACs characterise up to isomorphism algebras of functions, called respectively functional clone algebras (FCAs) and concrete $\aleph_0$-clones ($\aleph_0$-CC). The elements of a FCA and of a $\aleph_0$-CC with value domain $A$ are  operations $\varphi: A^\omega \to A$ of arity $\omega$ (referred to as {\em $\omega$-operations}).
In this framework the nullary operators $\e_i$ are the projections, and the $q$'s operators are the compositions of $\omega$-operations.  The universe of a FCA (resp. $\aleph_0$-CC)  is referred to as an $\omega$-clone (resp. infinitary $\omega$-clone).
%It turns out that every infinitary $\omega$-clone is also an $\omega$-clone and that the set of $\omega$-polymorphisms of a relational structure is an infinitary $\omega$-clone. 
It is notable that every infinitary $\omega$-clone is also an $\omega$-clone, and the set of $\omega$-polymorphisms of a relational structure forms an infinitary $\omega$-clone as well.
Furthermore, infinitary $\omega$-clones naturally extend the concept of clones, as every clone can be encoded into an appropriate infinitary $\omega$-clone.
%Moreover,  infinitary $\omega$-clones are a natural generalisation of the notion of clone, as every clone can be encoded into a suitable  infinitary $\omega$-clone.

In the algebraic approach to the complexity of CSPs, a key tool is the Galois connection $Pol$-$Inv$ between polymorphisms and invariant relations. An important problem is to characterise clones of operations $F$ for which $F=Pol(Inv(F))$. 
To identify the ``good clones'' satisfying $F=Pol(Inv(F))$, two approaches can be followed.
 The topological one  involves equipping the set of operations with the topology of pointwise convergence, so that  the good clones are the ones that are topologically closed. 
 This approach is widely followed in the literature (e.g. see \cite{B21}).
   The topology-free approach takes into account the cardinality of the basic set $A$ and the arity of operations. %For a clone $F$ on a finite set, it is possible to define a finitary relation $R_F\in Inv(F)$ such that $f\in Pol(R_F)$ if and only if $f\in F$. In this finite case, all clones are good. However, the drawback is that the arity of the relation $R_F$ grows rapidly with the cardinality of $A$ and/or the arity of the operations. 
   Bruno Poizat fully developed the topology-free approach in \cite{P81}. In \cite{Po80,Po84} R. P\"oschel characterised the Galois closed sets of relations with the help of infinitary operations.

%Clones are characterised by the Galois connection $Pol$-$Inv$, relating polymorphisms and invariant relations. Identifying "good clones" satisfying $F=Pol(Inv,F)$ can be approached through the topological or topology-free methods. The topological approach relies on the topology of pointwise convergence (see for instance \cite{B21}), while the topology-free approach considers the cardinality and arity of operations and relations (see \cite{P81}).

In this paper, we adopt the topological approach as we primarily focus on $\omega$-operations and relations of arity $\omega$, referred to as $\omega$-relations.
We present  a method for defining topologies on sets of functions. The key idea is to choose a Boolean ideal $X$ of subsets of $B$ to endow $A^B$ with a topology. A subset $F$ of $A^B$ is $X$-closed if it satisfies the following condition:
 $$\forall g\in A^B \ [(\forall d\in X \ \exists f\in F\ g_{|d}=f_{|d}) \Rightarrow g\in F].
 $$
When  $X$ is the family of finite subsets of $B$, we recover the topology of pointwise  convergence on $A^B$, referred to as local topology in this paper.

An ideal $X$ on $A^\omega$ determines the so-called $X$-topology on the set $O_A^{(\omega)}$ of all the $\omega$-operations on $A$.  We provide four significant examples of $X$-topologies on $O_A^{(\omega)}$: the local, global, trace and uniform topologies, respectively. We prove that if $C$ is an $\omega$-clone, then the local, global, trace and uniform closures of $C$ are all $\omega$-clones. The same result for the infinitary $\omega$-clones only holds for the local and global topologies.

%In the classical approach to polymorphisms and invariant relations, topology, especially the local topology, plays a crucial role only in the sets of  finitary operations.
%However, it seems to have no significant impact but not in the  finitary relations.
%The reason for this difference lies in the nature of the exponents involved. While the exponent $A^k$ in $A^{A^k}$ is generally infinite, the exponent $n$ in $A^n$ is always finite. As a result, 
%In fact, the local topology on $A^{A^k}$ is generally non-trivial, whereas the local topology on $A^n$ is discrete.
In the infinitary approach to polymorphisms and invariant relations, topology plays an important role %not only in the sets of $\omega$-operations but
also in the $\omega$-relations. The framework we develop in this paper enables us to define, given an ideal $X$ on $A^\omega$, an ideal $\mathcal X(m)$ on  $\omega$, which depends on a matrix $m:\omega\times\omega\to A$.
Therefore, we can equip the set $A^\omega$ with the $\mathcal{X}_m$-topology. These parametric topologies are used to define the concepts of $X$-polymorphism and $X$-invariant relation, both of which are parametric with respect to the ideal $X$ on $A^\omega$.
We prove that: (i) if $R\subseteq A^\omega$ is an $\omega$-relation, then the  set $Pol^\omega_X(R)$ of all the $X$-polymorphisms of $R$ is an $X$-closed infinitary $\omega$-clone; (ii) an infinitary $\omega$-clone $C$ is $X$-closed if and only if $C=Pol^\omega_X(Inv^\omega_X\, C)$. %We also provide further characterisations of  the locally closed and the globally closed infinitary $\omega$-clones.

%For the local topology, the locally closed infinitary $\omega$-clones can be characterised by using the locally closed $\omega$-relations only. The global topology is used to characterise the polymorphisms of the invariant $\omega$-relations of clones.

For the opposite side $Inv^\omega$-$Pol^\omega$ of the Galois connection, we define the notion of $\omega$-relation clone and show that the operators acting in such $\omega$-relation clones preserve the local closedness of $\omega$-relations.
%To compare relation clones %, which  are sets of finitary relations definable by infinitary primitive positive formulas,
%and $\omega$-relation clones we introduce the notion of cut-closedness and  show that significant relation clones are cut closed. 
%the following relation clones are cut-closed:
%(a) $\mathcal R_{\mathrm{fin}}=\{S\in Rel(A) :S^\top\in \mathcal R\}$, for every $\omega$-relation clone $\mathcal R$, where $S^\top$ is the $\omega$-relation obtained by expanding in all the possible ways the tuples of $S$; (b) 
%Relation clones on a finite set.
In order to characterise the $\omega$-relation clones  of locally closed $\omega$-relations ($c\omega$-relation clones) 
we define the notion of decreasing sequence of finitary relations. Each of these sequences has a locally closed $\omega$-relation as a limit  and we show that 
$\mathcal R$ is a $c\omega$-relation clone iff $\mathcal R =\mathrm{Lim}\,\mathcal S$, for some relation clone $\mathcal S$.
%\begin{enumerate}
%\item $\mathcal R$ is a $c\omega$-relation clone iff $\mathcal R =\mathrm{Lim}\,\mathcal S$, for some relation clone $\mathcal S$.
%\item A relation clone $\mathcal S$ is cut-closed iff $\mathcal S=(\mathrm{Lim}\,\mathcal S)_{\mathrm{fin}}$.
%\end{enumerate}
In the last result of the paper, Theorem \ref{thm:convergent}, we  present a method to relate $Inv^\omega$-$Pol^\omega$ to $Inv$-$Pol$.

\vspace{-0.5cm}
\section{Preliminaries}\label{sec:prel}

The notations in this paper are pretty standard. For
concepts, notations and results not covered hereafter, the reader is
referred to \cite{BS81,mac87} for universal algebra, to \cite{FMMT22a,L06,SZ86} for the theory of clones and to \cite{BS22,S22} for clone algebras. 

%In this paper, the symbol $\omega$ represents the first infinite ordinal. If $n<\omega$ is a finite ordinal, it can be understood as the set of elements $\{0, \ldots, n-1\}$.

In this paper, the symbol $\omega$ represents the first infinite ordinal. A finite ordinal $n<\omega$ can be understood as the set of elements $\{0, \ldots, n-1\}$.

In the rest of this section, the symbol $A$ will be used to denote an arbitrary set.

\subsection{$\omega$-Sequences}
%A \emph{thread} is an element of $A^\omega$.  

\begin{enumerate}
\item If $s\in A^\omega$, then $set(s)=\{s_i: i\in\omega\}$.
\item If $s\in A^\omega$ and $a_0,\dots,a_{n-1}\in A$ then $s[a_0,\dots,a_{n-1}]\in A^\omega$ is defined as follows:
  $$s[a_0,\dots,a_{n-1}]_i=\begin{cases}a_i&\text{if $0\leq i\leq n-1$}\\ s_i&\text{if $i \geq n$}\end{cases}$$
If $\boldsymbol a = a_0,\dots,a_{n-1}$, we write $s[\boldsymbol a]$ for $s[a_0,\dots,a_{n-1}]$.
%\item  If $s\in A^\omega$, $a\in A$ and $k\in \omega$, then $s[a/k]$ is a thread such that  $s[a/k]_i=a$ if $i=k$; $s[a/k]_i=s_i$ if $i\neq k$. 
    
 \item If $a\in A$, then $a^\omega\in A^\omega$ is such that $(a^\omega)_i = a$ for all $i$. Similarly, we define $a^n$ to be the finite sequence $(a,\dots,a)$ with $n$ occurrences of $a$.

 \item Sometimes, we employ word notation to represent $\omega$-sequences. For example, $0^n1^\omega$ signifies the sequence $(0, ..., 0, 1, 1, 1, ...)$ with $n$ occurrences of $0$ followed by an infinite repetition of $1$.
\item If $s\in A^\omega$ and $d\subseteq_{\mathrm{fin}} \omega$, then $s_{|d}$ is a map from $d$ to $A$ such that
$(s_{|d})_i=s_i$ for every $i\in d$. In particular, if $n$ is a natural number, then $s_{|n}=(s_0,\dots,s_{n-1})\in A^n$.

\end{enumerate}

\subsection{Traces}\label{sec:tra}
%\begin{enumerate}
%\item Let $E$ be an equivalence relation on $A^\omega$. An \emph{$E$-trace} is a subset $\mathsf a$  of $A^\omega$ closed under $E$:
%$$s\in\mathsf a\ \text{and}\ rE s \Rightarrow r\in\mathsf a.$$
%The equivalence class of $s\in A^\omega$ will be denoted by $[s]_E$.
%A \emph{compact $E$-trace} is a $E$-trace, which is union of a finite number of equivalence classes.
%A \emph{basic $E$-trace} refers to a single equivalence class of the equivalence relation $E$.

The  equivalence relation $\equiv$ on $A^\omega$, defined by:
$$s\equiv r\ \Leftrightarrow\ |\{i: s_i\neq r_i\}|< \omega,$$
will be of particular importance in this article.
A \emph{trace on $A$} is a subset $\mathsf a$  of $A^\omega$ closed under the relation $\equiv$:
$$s\in\mathsf a\ \text{and}\ r\equiv s \Rightarrow r\in\mathsf a.$$
The equivalence class of $s\in A^\omega$ will be denoted by $[s]_\equiv$.
A \emph{compact trace} is a trace, which is union of a finite number of equivalence classes.
A \emph{basic trace} refers to a single equivalence class of the equivalence relation $\equiv$.

Observe that a compact trace is a compact (or finite) element in the complete lattice of traces, much like how  in a topological space the compact open sets are compact elements in the lattice of open sets.
%We denote $\mathsf c\subseteq_{\mathrm{com}}A^\omega$ to indicate that $\mathsf c$ is a compact trace on $A$.
%We write $\mathsf c\subseteq_{\mathrm{com}}A^\omega$ when $\mathsf c$ is a compact trace on $A$.

%\end{enumerate}

\subsection{Operations and $\omega$-operations}
\begin{enumerate}
\item A \emph{finitary operation on $A$} is a function $f:A^n\to A$ for some $n\in\omega$. We denote by $O_A$ the set of all finitary operations on $A$. If $F\subseteq O_A$ then $F^{(n)}=\{f\in F\ |\ f:A^n\to A \}$.

 \item An  \emph{$\omega$-operation on $A$} is a function $\varphi:A^\omega\to A$. We denote by $O_A^{(\omega)}$ the set of all $\omega$-operations on $A$.

%\item $O^{\leq\omega}_A=O_A\cup O^{(\omega)}_A$.

\item If $f: A^n\to A$ is a finitary operation, then the \emph{top extension $f^\top$ of $f$} is the $\omega$-operation defined as follows: $f^\top(s)=f(s_0,\dots,s_{n-1})$ for every $s\in A^\omega$. If $F$ is a set of finitary operations, then we define
$F^\top=\{f^\top : f\in F\}$.

\item If $C$ is a set of $\omega$-operations, then we define $C_{\mathrm{fin}}=\{f\in O_A: f^\top \in C\}$.

\item An $\omega$-operation $\varphi:A^\omega\to A$ is \emph{semiconstant} if for every basic trace $[s]_\equiv$, we have:
$r,u\in [s]_\equiv \Rightarrow \varphi(r)=\varphi(u)$.

\item The $\omega$-operation $\e^A_n:A^\omega\to A$ is the projection in the $n$-th coordinate: $\e^A_n(s)=s_n$ for every $s\in A^\omega$.

\item If $\varphi_i\in O_A^{(\omega)}$ for $i\in\omega$, then $\boldsymbol \varphi$ denotes the sequence 
$(\varphi_0,\varphi_1,\dots,\varphi_n,\dots)$.
\end{enumerate}

\subsection{Relations and $\omega$-relations}\label{sec:rel}
\begin{enumerate}
\item A \emph{finitary relation on $A$} is a subset $S$ of $A^n$ for some natural number $n\neq 0$.
We denote by $Rel_A$ the set of all finitary relations on $A$.

\item An \emph{$\omega$-relation on $A$} is a subset $R$ of $A^\omega$. We denote by $Rel_A^{(\omega)}$ the set of all $\omega$-relations on $A$.

\item $Rel^{\leq\omega}_A=Rel_A\cup Rel^{(\omega)}_A$.

\item If $S\subseteq A^n$ is a finitary relation, then the \emph{top extension $S^\top$ of $S$} is the $\omega$-relation defined as follows: $S^\top=\{s\in A^\omega: (s_0,\dots,s_{n-1})\in S\}$. If $\mathcal S$ is a set of finitary relations, then
$\mathcal S^\top=\{S^\top : S\in \mathcal S\}$ is a set of $\omega$-relations.

\item If $\mathcal R$ is a set of $\omega$-relations on $A$, then $\mathcal R_{\mathrm{fin}}= \{S\in Rel_A:  S^\top\in \mathcal R\}$.

\item $\Delta_A=\{ a^\omega: a\in A\}$ is called the \emph{diagonal $\omega$-relation}.
\end{enumerate}

\subsection{Matrices}\label{sec:matr} Let $\alpha,\beta\leq \omega$. Then $\mathsf M_\alpha^\beta(A)$ denotes the set of matrices of order $\alpha\times \beta$ on $A$, i.e., $\mathsf M_\alpha^\beta(A)$ is the set of all maps $m:\alpha\times\beta\to A$. 
When there is no danger of confusion, we write $\mathsf M_\alpha^\beta$ for $\mathsf M_\alpha^\beta(A)$.

If $m\in \mathsf M_\alpha^\beta$, then $m_i$ is the $i$-th row of $m$ and $m^j$ is the $j$-th column of $m$. 
%We  describe a matrix by rows $m=(m_0,m_1,\dots,m_n,\dots)$ or by columns $m=(m^0,m^1,\dots,m^n,\dots)$. 
The element belonging to row $i$ and  column $j$ is denoted by $m^j_i$. 

We adopt the following notations:
\begin{enumerate}
\item If $B\subseteq A^\alpha$, then $\mathsf M_\alpha^{\beta,B}$  denotes the set of all matrices $m\in \mathsf M_\alpha^\beta$ such that each column $m^j$  belongs to $B$. 
%If $B\subseteq A^\alpha$ and $C\subseteq A^\beta$, then $\mathsf M_\alpha^{\beta,B}$ (resp. $\mathsf M_{\alpha,C}^\beta$) denotes the set of all matrices $m\in \mathsf M_\alpha^\beta$ such that each column $m^j$ (resp. row $m_j$)  belongs to $B$ (resp. $C$). Accordingly, $\mathsf M_{\alpha,C}^{\beta,B}= \mathsf M_\alpha^{\beta,B}\cap \mathsf M_{\alpha,C}^\beta$.  

\item $\mathsf M = \bigcup_{\alpha,\beta\leq \omega} \mathsf M_\alpha^\beta$ and $\mathsf M^\omega = \bigcup_{\alpha\leq \omega} \mathsf M_\alpha^\omega$. 

\item If $m\in \mathsf M_\alpha^\beta$, then $\mathrm{row}(m)=\{ m_i: i\in \alpha\}$.

\item If $m\in \mathsf M_\alpha^\beta$ and $\varphi:A^\beta \to A$ is a function, then we define $\varphi[m]\in A^\alpha$ by
$$\varphi[m]= (\varphi(m_i): i< \alpha).$$
%If $\psi:A^\beta \to A$ is another function, then we write $\varphi[m]=\psi[m]$ iff $\varphi(m_i)=\psi(m_i)$ for every row $m_i$.  Then $\varphi[m]\neq\psi[m]$ means that there exists a row $m_i$ such that $\varphi(m_i)\neq\psi(m_i)$.
\item If $m\in \mathsf M_\alpha^\beta$ and  $\mathbf s=(s^0,\dots,s^{k-1})$ with $s^i\in A^\alpha$, then $m[\mathbf s]$ is the matrix defined as follows:
$$m[\mathbf s]^j=\begin{cases}s^j&\text{if $0\leq j\leq k-1$}\\ m^j&\text{if $k \leq j <\beta$}\end{cases}$$
%\item If $m\in \mathsf M_\alpha^\beta$ and $\varphi,\psi:A^\beta \to A$ are functions, we write $\varphi[m]=\psi[m]$ if $\varphi(m_i)=\psi(m_i)$ for every $i<\alpha$; $\varphi[m]\neq\psi[m]$ if there exists $i<\alpha$ such that $\varphi(m_i)\neq\psi(m_i)$.
\end{enumerate}

\subsection{The local topology on spaces of functions}\label{sec:loctop}
Let $A$ and $B$ sets. In the context of this work, the topology of pointwise convergence on $A^B$ is referred to as the \emph{local topology}. 
When $A$ is equipped with the discrete topology, the local topology  is the product topology on $A^B$.

The local topology on $A^B$ can be defined using a subbase consisting of sets of the form $O_{d,f} = \{ g \in A^B : g_{|d} = f_{|d} \}$, where $d$ ranges over all finite subsets of $B$ and $f$ belongs to $A^B$. It is important to note that if $B$ is finite, then the local topology coincides with the discrete topology.

For a subset $R \subseteq A^B$, the local closure of $R$ in $A^B$, denoted by $\mathrm{Cl}_L(R)$, is defined as follows: $f \in \mathrm{Cl}_L(R)$ if for every finite subset $d \subseteq B$, there exists $g \in R$ such that $f_{|d} = g_{|d}$.

In this paper we will consider the local topology on the set $A^\alpha$ for $\alpha \leq \omega$ and on the set $O_A^{(\omega)}$ of all $\omega$-operations. As specified before, the local topology on $A^\alpha$ for $\alpha < \omega$ is the discrete topology.

%For typographical reasons, if $R$ is an $\omega$-relation, we denote by $\overline R$ the local closure of $R$. In fact,  for denoting the set of matrices, whose columns belong to
%$\mathrm{Cl}_L(R)$, the notation 
%$\mathsf M_\alpha^{\beta,{\overline R}}$ is more readable than
%$\mathsf M_\alpha^{\beta,\mathrm{Cl}_L(R)}$.

When $R$ is an $\omega$-relation, we employ $\overline R$ to represent the local closure of $R$. Specifically, for indicating the set of matrices, whose columns belong to the local closure of $R$, the notation $\mathsf M_\alpha^{\beta,{\overline R}}$ is preferred over $\mathsf M_\alpha^{\beta,\mathrm{Cl}_L(R)}$ for its enhanced readability.

\begin{lemma}\label{lem:clofin} 
   For every  relation $S\subseteq A^n$, the $\omega$-relation $S^\top$ is  locally closed.
\end{lemma}

\begin{proof}
 Let $s\in \overline{S^\top}$. There exists $u\in S^\top$ such that $u_i=s_i$ for $0\leq i\leq n-1$. This implies that $(s_0,\ldots,s_{n-1})\in S$ and therefore $s\in S^\top$. 
 \end{proof}

\subsection{Clones and polymorphisms}\label{sec:cp}
In this section we recall notations and terminology on clones and polymorphisms we will use in the following. 

\begin{enumerate}
\item  The \emph{composition}  of $f\in  O_A^{(n)}$ with  $g_1,\dots,g_n\in   O_A^{(k)}$ is the operation $f(g_1,\dots,g_n)\in  O_A^{(k)}$ defined as follows, for all $\mathbf a\in A^k$: 
$$f(g_1,\dots,g_n)(\mathbf a)= f(g_1(\mathbf a),\dots, g_n(\mathbf a)).$$

%\noindent If $f\in  \mathcal{O}_A^{(0)}$ then $f()_k\in \mathcal{O}_A^{(k)}$ and $f()_k(\mathbf a)= f$ for all $\mathbf a\in A^k$. When there is no danger of confusion, we write $f(g_1,\dots,g_n)$ for $f(g_1,\dots,g_n)_k$.

\item   A {\em  clone on a set $A$} is a subset $F$ of $O_A$ containing all projections
$p^{(n)}_i:A^n\to A$ ($n\geq i$) and closed under composition. 

%\item If $F$ is an arbitrary set of finitary operations on $A$, then $[F]$ denotes the clone generated by $F$.
\end{enumerate}

%A \emph{clone on algebra $\mathbf A$} is a clone on $A$  containing the finitary operations $\sigma^\mathbf A$ of $\mathbf A$.

% The classical approach to clones, as evidenced by the standard monograph \cite{SZ86}, considers clones only containing operations that are at least unary. In this paper we require clones allowing nullary operators.

\begin{definition}\label{def:2.2}
  Given a finitary operation $f:A^k \to A$ and a finitary relation $S\subseteq A^n$, we say that
$f$ is a \emph{polymorphism of $S$} and that $S$ is an \emph{invariant relation} of $f$ if, for all matrices $m\in\mathsf M_n^{k,S}$,
 $f[m]\in S$.
\end{definition}

An operation $f$ is a polymorphism of a set of relations if it is a polymorphism of each of them. Similarly, a finitary relation $S$ is an invariant relation of a set $F$ of operations if it is an invariant relation of each element of $F$.

Let $\mathcal S$ be a set of finitary relations on $A$ and $F$ be a set of finitary operations on $A$.
We define $Pol\,\mathcal S= \{ f\in O_A: \text{$f$ is a polymorphism of each $S\in \mathcal S$}\}$
and $Inv\,F=\{S\in Rel_A: \text{$S$ is an invariant relation of each $f\in F$} \}$. 

\begin{definition}\label{def:src}
  For a set $A$, a \emph{relation clone on $A$} is a set $\mathcal S$ of finitary relations satisfying the following conditions:
\begin{itemize}
\item[(i)]  $\Delta_A^{(n)}= \{(a_1,\dots,a_n) : a_1,\dots,a_n\in A, a_1=a_2=\dots=a_n \}\in \mathcal S$ for every $n\geq 1$.
\item[(ii)] If $S,U\in \mathcal S$, $S$ is $n$-ary and $U$ is $m$-ary, then 
$$S\times U=\{(a_0,\dots,a_{n-1},b_0,\dots,b_{m-1}): $$
$$\qquad\qquad\qquad\qquad\qquad\qquad(a_0,\dots,a_{n-1})\in S, (b_0,\dots,b_{m-1})\in U \}\in\mathcal S.$$
\item[(iii)] If $S\in\mathcal S$ is $n$-ary and $f:m\to n$ is a map, then
$\pi_f(S)=\{s\circ f: s\in S\}\in\mathcal S$.
%\item[(iv)] If $U,S\in\mathcal S$ are $n$-ary relations, then $U\cap S\in\mathcal S$.
%If $R\subseteq A^m$ is in $\mathcal R$, $n\leq m$ and  $1\leq i_1\dots,i_n\leq m$ are distinct, then $$R_{i_1,\dots,i_n} =\{(a_{i_1},\dots,a_{i_n}): (a_1,\dots,a_m)\in R\} \in\mathcal R.$$
 \item[(iv)] If $S_i\in\mathcal S$ ($i\in I$) is a family of relations in $\mathcal S$ having the same arity, then $\bigcap_{i\in I}S_i\in \mathcal S$.
\end{itemize}
 A relation clone $\mathcal S$ is called \emph{strong} if it satisfies the following further property: 
\begin{itemize}
\item[(v)] If $\{S_i\}_{i\in I}$ is a directed family of relations $S_i\in \mathcal S$  having the same arity, then $\bigcup_{i\in I}S_i\in \mathcal S$.
\end{itemize}
\end{definition}

It should be noted that when $\mathcal S$ is a relation clone on a finite set, property (v) holds vacuously and in   property (iv)
the finite intersections suffice.

It is well known that, for every set $\mathcal S$ of finitary relations and set $F$ of finitary operations,  $Pol\,\mathcal S$ is a clone and $Inv\,F$ is a strong relation clone.

Recall that the local topology  on $O_A$ is obtained by viewing this space as the sum space of the spaces $A^{A^k}$, and each $A^{A^k}$ as a power of $A$, which itself is taken to be discrete.

The following result, due to Romov \cite{rom77}, can be found in \cite[Corollary 1.9]{SZ86} and \cite[Proposition 6.1.5]{B21}.

\begin{proposition}\label{prop:b21} For a clone $F$ of finitary operations on $A$, $F$ is closed under the local  topology on $O_A$ if and only if  $Pol(Inv\,F) = F$. 
\end{proposition}

%Moreover, $F$ is a closed  clone of finitary operations w.r.t. the pointwise convergent topology on $O_A$ (see the comment before Lemma \ref{lem:toplocclosed}) if and only if $Pol(Inv\,F)=F$ (cf. \cite[Corollary 1.9]{SZ86} and \cite[Proposition 6.1.5]{B21}).

%The set of  relation clones on $A$ is closed under arbitrary intersection. 
%If $\mathcal S$ is a set of finitary relations, then $\langle \mathcal S\rangle_{\infty pp}$ denotes the least relation clone including $\mathcal S$.

%\begin{theorem} \cite{BKKR69}
% Let $A$ be a set and $F$ be a set of finitary operations on $A$. Then, $F$ is a clone iff  $F = Pol(Inv(F)))$.
%\end{theorem}

%If   $A$ is a finite set, then  item (v) in the above definition is trivial, because there exists a finite number of $n$-ary relations. Similarly in (iv), it is sufficient to require the closure under finite intersection.

The following theorem is due to  Geiger \cite{gei68} and, independently, to Bondarc\v{u}k, Kalu\v{z}nin, Kotov and Romov \cite{BKKR69}.

\begin{theorem}\label{thm:invpol} For a set $\mathcal S$ of finitary relations on a finite set $A$, $\mathcal S$ forms a relation clone if and only if $\mathcal S = Inv(Pol\,\mathcal S)$.
\end{theorem}

The following generalisation of Theorem \ref{thm:invpol} was shown by Romov \cite[Theorem 3.5]{rom08}.

\begin{theorem}\label{thm:invpolromov} For a set $\mathcal S$ of finitary relations on a countable set $A$, $\mathcal S$ forms a strong relation clone if and only if $\mathcal S = Inv(Pol\,\mathcal S)$.
\end{theorem}

\section{$\omega$-Clones}

In this section we introduce the notions of an $\omega$-clone and an infinitary $\omega$-clone. The reader may consult \cite{neu70,BS22,S22} for an algebraic theory of (infinitary) $\omega$-clones.

\begin{definition}
  An \emph{$\omega$-clone on a set $A$} is a subset $C$ of $O_A^{(\omega)}$ containing all projections $\e_n^A:A^\omega\to A$ and closed under the following family of operations $q_n^A$ of arity $n+1$ ($n\geq 0$):
  $$q^A_n(\varphi,\psi_0,\dots,\psi_{n-1})(s)=\varphi(s[\psi_0(s),\dots, \psi_{n-1}(s)]),$$
   for every $s\in A^\omega$ and $\varphi,\psi_0,\dots,\psi_{n-1}\in  C$.
\end{definition}

As a matter of notation, we sometimes write $q_n(x,\bar y)$ for  $q_n(x,y_0,\dots,y_{n-1})$.

If $C$ is an $\omega$-clone, then the algebra $(C, q_n^A, \e^A_n)_{n\geq 0}$ is called a \emph{functional clone algebra $(\mathsf{FCA})$ with value domain $A$} (see \cite{BS22}).

The following was one of the main results in \cite{BS22}.

\begin{theorem}\label{thm:bs} 
The class $\mathbb I \mathsf{FCA}$ is a variety of algebras, axiomatised by the following identities:
 \begin{itemize}
\item[(C1)] $q_n(\e_i,y_0,\dots,y_{n-1})=y_i$ $(0\leq i\leq n-1)$;
\item[(C2)] $q_n(\e_i,y_0,\dots,y_{n-1})=\e_i$ $(i\geq n)$;
\item[(C3)] $q_n(x,\e_0,\e_1,\dots,\e_{n-1})=x$;
\item[(C4)]  $q_n(x,y_0,\dots,y_{n-1})= q_k(x,y_0,\dots,y_{n-1},\e_n,\dots,\e_{k-1})$ ($k\geq n$);
 \item[(C5)] $q_n(q_n(x,\bar y),\bar z)=q_n(x,q_n(y_0,\bar z), \dots,q_n(y_{n-1},\bar z))$.
\end{itemize}

\end{theorem}

\begin{definition} An $\omega$-clone $C$ is \emph{infinitary} if it is closed under the following operation $q^A$ of arity $\omega$:
 $$q^A(\varphi,\psi_0,\dots,\psi_{n-1},\psi_n,\dots)(s)=\varphi(\psi_0(s),\dots, \psi_{n-1}(s), \psi_n(s),\dots).$$
   for every $s\in A^\omega$ and $\varphi,\psi_i\in  C$.
\end{definition}

 The operation $q^A$ has been introduced by Neumann in \cite{neu70}. 
Note that $q^A_n(\varphi,\psi_0,\dots,\psi_{n-1})= q^A(\varphi,\psi_0,\dots,\psi_{n-1},\e_n^A,\e_{n+1}^A,\dots)$.

If $C$ is an infinitary $\omega$-clone, then the infinitary algebra $(C, q^A, \e^A_n)_{n\geq 0}$ is called in \cite{neu70} a \emph{concrete $\aleph_0$-clone ($\aleph_0$-$\mathsf{CC}$) with value domain $A$}.

As a matter of notation, we write $q(x,\bar y)$ for  $q(x,y_0,y_1,\dots,y_k,\dots)$. 

 In \cite{neu70} it was shown the following theorem.

\begin{theorem}\label{thm:neu} The class $\mathbb I\,\aleph_0$-$\mathsf{CC}$  is a variety of infinitary algebras axiomatised by the following identities:
 \begin{itemize}
\item[(N1)] $q(\e_n,\bar y)=y_n$;
 \item[(N2)] $q(x,\e_0,\e_1,\dots,\e_k,\dots)=x$;
 \item[(N3)] $q(q(x,\bar y),\bar z)=q(x,q(y_0,\bar z), \dots,q(y_k,\bar z),\dots)$.
\end{itemize}
\end{theorem}

% The proof of Theorem \ref{thm:bs} is much more complex than the proof of Theorem \ref{thm:neu}.
 
 \subsection{Clones as $\omega$-clones}\label{sec:clomega}
The following  lemmas establish that the theory of clones can be encoded within the theory of $\omega$-clones.

We define an equivalence relation $\approx$ on the set $O_A$ of finitary operations as follows:
$$f\approx g\ \text{if}\ f^\top=g^\top.$$
For every $k\in\omega$ and $f\in O_A$, the set $[f]_\approx\cap A^{A^k}$ is either a singleton or empty. If it is a singleton, then $[f]_\approx\cap A^{A^n}$ is a singleton for every $n\geq k$. 
%Roughly speaking, $f\approx g$ if and only if $f$ and $g$ are ``equal'' up to dummy arguments. 
It was shown in \cite{BS22} that every clone $F$ is closed under the equivalence $\approx$, i.e. $f\in F$ and $g\approx f$ implies $g\in F$.

\begin{lemma}\label{lem:clomegacl} Let $F\subseteq O_A$ and $C\subseteq O_A^{(\omega)}$. We have:

\begin{enumerate}
\item If $F$ is a clone, then $F^\top$ is an infinitary $\omega$-clone and $(F^\top)_{\mathrm{fin}}=F$.
\item If $C$ is an $\omega$-clone, then $C_{\mathrm{fin}}$ is a clone.
\end{enumerate}
\end{lemma}

\begin{proof} (1) We have $\e^A_i = (p^{(n)}_i)^\top$ ($n\geq i$) and, for every $f:A^n\to A$, 
$q^A(f^\top,g_0^\top,\dots,g_k^\top,\dots)=(f(h_0,\dots,h_{n-1}))^\top$, where the finitary operations $h_i$ have the same arity $m$ and $g_i\approx h_i$ for every $0\leq i\leq n-1$. We have $(F^\top)_{\mathrm{fin}}=F$, because the clone $F$ is closed under the equivalence $\approx$.
 
 (2) Similar to (1).
\end{proof}

The following lemma will be used in the proof of Proposition \ref{prop:4.9bis} below.

\begin{lemma}\label{lem:4.9boh} Let $\varphi$ be an $\omega$-operation on $A$ and $S\subseteq A^n$ be a finitary relation. Then the following conditions are equivalent:
 \begin{enumerate}
\item  For all matrices $m\in\mathsf M_n^{\omega,S}$,
 $\varphi[m]\in S$.
\item For all matrices $m\in\mathsf M_\omega^{\omega,S^\top}$, $\varphi[m]\in S^\top$.
\end{enumerate}
\end{lemma}

\begin{proof}
  Let $m\in  \mathsf M_n^\omega$. An extension of $m$ is any matrix $p \in \mathsf M_\omega^\omega$ such that $p_i=m_i$ for every $0\leq i\leq n-1$. If $m\in  \mathsf M_n^{\omega,S}$, we have  $p \in \mathsf M_\omega^{\omega,S^\top}$ for every extension $p$ of $m$. The equivalence of (1) and (2) follows because $\varphi[p]\in S^\top$ if and only if $\varphi[m]=\varphi[p]_{|n}\in S$.
\end{proof}
 
Let $R\subseteq A^\alpha$ be a relation for $\alpha\leq\omega$ and $\varphi:A^\omega\to A$ be an $\omega$-operation. We define 
$$Pol^\top(R)= \{f^\top : \forall m\in\mathsf M_\alpha^{\omega,R},\ f^\top[m]\in R\}$$ 
$$Inv^\top(\varphi)=\{S^\top: \forall m\in\mathsf M_\omega^{\omega,S^\top},\ \varphi[m]\in S^\top \}.$$
Similarly, we can define $Pol^\top(\mathcal R)$ for a set $\mathcal R$ of relations and $Inv^\top(C)$ for a set $C$ of $\omega$-operations.

\begin{proposition}\label{prop:4.9bis} Let $f:A^k\to A$ be a finitary operation and $S\subseteq A^n$ be a finitary relation. Then the following conditions are equivalent:
 \begin{enumerate}
\item $f\in Pol(S)$  (see Definition \ref{def:2.2}).
\item  $f^\top\in Pol^\top(S)$.
\item $f^\top\in Pol^\top(S^\top)$.
\end{enumerate}
Therefore, $(Pol\, S)^\top= Pol^\top(S^\top)$ and  $(Inv\,f)^\top=Inv^\top(f^\top)$.
\end{proposition}

\begin{proof}
  (2) $\Leftrightarrow$ (3) follows from Lemma \ref{lem:4.9boh}.

  (1) $\Leftrightarrow$ (2) For all 
  $m\in \mathsf M_n^{k,S}$, $f[m]\in S$  if and only if for all extensions
  $p\in \mathsf M_n^{\omega,S}$ of $m$, $f^\top[p]\in S$.
  \end{proof}

We conclude this section by rephrasing Proposition \ref{prop:b21} using the language of $\omega$-clones.
We begin with the following lemma, where we demonstrate that the operator $\cdot^\top$ is continuous with respect to the local topology.

%which compares the local topology on $O_A$ and on $O_A^{(\omega)}$.

%As a matter of notation, if $[f]_\approx$ is an equivalence class of the relation $\approx$, we write $[f]_\approx^\top$ for the $\omega$-operation $g^\top$, where  $g$ is an arbitrary element of $[f]_\approx$.

\begin{lemma}\label{lem:toplocclosed} Let $F\subseteq O_A$ be a set of finitary operations closed under the equivalence $\approx$.
  %\begin{itemize}
  %  \item[(1)]
  If $F^\top$ is locally closed in the subspace  $O_A^\top$ of $O_A^{(\omega)}$, then $F$ is locally closed in $O_A$.
  %\item[(2)] Let $F$ be a clone.  $F$ is locally closed in $O_A$ if and only if  $F^\top$ is locally closed in the subspace  $(O_A)^\top$ of $O_A^{(\omega)}$.
     % \end{itemize}
\end{lemma}

\begin{proof} %(1)
  Let $g$ be an $n$-ary operation in the local closure of $F$ in $O_A$. 
%  $\mathrm{Cl}_L(F^\top)$ be the local closure of $F^\top$ in $O_A^{(\omega)}$, and $g\in\mathrm{Cl}_L(F)$ be an $n$-ary operation. 
  We will prove that $g^\top$ belongs to the local closure of $F^\top$ in $O_A^{(\omega)}$.
 Let $d\subseteq_{\mathrm{fin}} A^\omega$ be arbitrary and $d'=\{ s_{|n} : s\in d\}$. There exists an $n$-ary operation $f\in F$ such that $f_{|d'}=g_{|d'}$. Since $(f^\top)_{|d}=(g^\top)_{|d}$,   by the arbitrariness of $d$ we obtain that $g^\top$ belongs to the local closure of $F^\top$ in $O_A^{(\omega)}$. Using the hypothesis that $F^\top$ is locally closed in the subspace  $O_A^\top$ of $O_A^{(\omega)}$,
 %$\mathrm{Cl}_L(F^\top)\cap O_A^\top= F^\top$, 
 we deduce that $g^\top \in F^\top$, and consequently, $g \in F$ by considering that $F$ is closed under the equivalence $\approx$.
 % (2 $\Leftarrow$) By (1). 
 %(2 $\Rightarrow$)  Let $ [g]_\approx^\top\in \overline{F^\top}$. We will prove that $g\in F$. Let $d\subseteq_{\mathrm{fin}} A^n$ be arbitrary and $c\subseteq_{\mathrm{fin}} A^\omega$ such that $|c|=|d|$ and $d=\{ s_{|n} : s\in c\}$. There exists an operation $f$ of arity $k\geq n$ such that $f^\top\in F^\top$ and $(g^\top)_{|c}=(f^\top)_{|c}$. 
  %Since  $F$ is closed under $\approx$,  then $f\in F$.  Let $h$ be an operation of arity $k$ such that $h\approx g$. 
% by the arbitrariness of $d$ we obtain $g^\top \in \overline{F^\top}$. Using the hypothesis $\overline{F^\top}= F^\top$, we deduce that $g^\top \in F^\top$, and consequently, $g \in F$ by considering that $F$ is closed under the equivalence $\approx$.
\end{proof}

%he converse of Lemma \ref{lem:toplocclosed}    is false as shown in Proposition  \ref{prop:ftop}(i).

%We conclude this section with the following characterisation of Pol-Inv (see Section \ref{sec:cp}).

\begin{proposition}\label{prop:4.9tris} Let $F \subseteq O_A$ be a clone. Then the following conditions are equivalent:
 \begin{enumerate}
\item $Pol(Inv\,F)=F$.
\item $F^\top$ is locally closed in the subspace $O_A^\top$ of $O_A^{(\omega)}$.
\item   $Pol^\top(Inv^\top\,F^\top)=F^\top$.
\end{enumerate}
\end{proposition}

\begin{proof} (2) $\Rightarrow$ (1) By Lemma \ref{lem:toplocclosed} and Proposition \ref{prop:b21}.

  (3) $\Rightarrow$ (2) Let $g^\top$ be an element of the local closure of $F^\top$ in $O_A^{(\omega)}$. We show that $g^\top\in Pol^\top(Inv^\top\,F^\top)$. Let $S^\top\in Inv^\top(F^\top)$ be such that $S\subseteq A^k$. By  Proposition \ref{prop:4.9bis}  $(Inv\,F)^\top=Inv^\top(F^\top)$, so that $S^\top\in (Inv\, F)^\top$ and therefore
  $S\in Inv(F)$. Given  $m\in \mathsf M_\omega^{\omega,S^\top}$, there exists $f^\top\in F^\top$ such that
  $g^\top(m_i)=f^\top(m_i)$ for all $0\leq i\leq k-1$.  As $f^\top[m]\in S^\top$, then
  $(f^\top(m_0),\ldots,f^\top(m_{k-1}))=(g^\top(m_0),\ldots,g^\top(m_{k-1}))\in S$.
  It follows that $g^\top[m]\in S^\top$. We get the conclusion $g^\top\in Pol^\top(Inv^\top\,F^\top)=F^\top$.

  (1) $\Rightarrow$ (3) For every finitary relation $S$, by  Proposition \ref{prop:4.9bis} we have that  $f\in Pol\, S$ if and only if 
$f^\top\in Pol^\top S^\top$. The conclusion follows from  $(Inv\,F)^\top=Inv^\top(F^\top)$.
\end{proof}

\vspace{-0.5cm}
 \section{Topologies on function spaces}\label{sec:topo}
 In this section we introduce a method for defining topologies over sets of functions, and we
 provide some examples of such topologies, that will be studied in the rest of this work. The basic idea is that, for endowing $A^B$ with a topology, it is enough to  choose a suitable family $X$ of subsets of $B$.
 %Then,  the topological closure of  $C\subseteq A^B$ is the set of functions that, when restricted to any $d\in X$, are equal to some element of $C$ restricted to $d$.
 The only condition needed for this to give raise to a topology on $A^B$ is that $X$ must be closed under finite union. The topologies introduced in this section are a particular case of those defined in \cite[Section 1, p. 353]{BK14}, when $A$ and $B$ are discrete topological spaces. 
 \begin{definition}
 Let $B$ be a set. A \emph{union semilattice} $X$ on $B$ is a nonempty family of subsets of $B$ closed under finite union. 
  \end{definition}
 A union semilattice $X$ generates the (Boolean) ideal $X\!\!\downarrow =\{c\subseteq B : (\exists d\in X)\ c\subseteq d \}$.

 \begin{definition}
   Let $X$ be a union semilattice on  $B$ and  $C\subseteq A^B$ for a given set $A$.
   The {\em $X$-closure of $C$},  $\mathrm{Cl}_X(C)$,   is defined as follows:
   $$\mathrm{Cl}_X(C)=\{\varphi\in A^B\ |\  \forall d\in X\exists \psi\in C: \varphi_{|d}=\psi_{|d}\}.
   $$
 \end{definition}

We say that $C$ is {\em $X$-closed} if $C=\mathrm{Cl}_X(C)$.
 Given two sets $A$ and $B$ and a union semilattice $X$ on $B$, let
 $$\tau_X^A=\{O\subseteq A^B: A^B\setminus O \mbox{ is $X$-closed}\}.$$

\begin{lemma}\label{lem:X} Let $A$ be a set and $X, Y$ be union semilattices on a set  $B$. Then we have:
\begin{enumerate}
\item $\tau_X^A$ is a topology on $A^B$.
\item If $X \subseteq Y\!\!\downarrow$ then $\tau_X^A\subseteq \tau_{Y}^A$.
\item $X$ and $X\!\!\downarrow$ define the same topology.
\end{enumerate}

\end{lemma}

\begin{proof}
  (1)  Let $C_1,C_2$ be $X$-closed sets. We want to prove that their union $C_1\cup C_2$ is $X$-closed. 
 Suppose $\varphi\notin C_1\cup C_2$ but $\varphi\in \mathrm{Cl}_X(C_1\cup C_2)$. 
  Since $C_1$ and $C_2$ are $X$-closed and $\varphi\notin C_1\cup C_2$, there exist $d_1, d_2\in X$ such that $\varphi_{|d_1}\neq\psi_{|d_1}$ for every $\psi\in C_1$, and $\varphi_{|d_2}\neq\psi_{|d_2}$ for every $\psi\in C_2$.
   This contradicts the hypothesis that there exists  $\gamma\in C_1\cup C_2$ with $\varphi_{|d_1\cup d_2}=\gamma_{|d_1\cup d_2}$. 
   Furthermore, the intersection of an arbitrary family of $X$-closed sets is trivially $X$-closed.

(2) For every $C\subseteq A^B$, we have that $C\subseteq \mathrm{Cl}_Y(C) \subseteq \mathrm{Cl}_X(C)$. If $C$ is $X$-closed, then $C$ is $Y$-closed.

(3) It immediately follows from (2).
\end{proof}

According to Lemma \ref{lem:X}(3), any union-closed semilattice on $B$ defines an equivalent topology as its corresponding ideal. Therefore, from now on we will focus on ideals. 

This approach based on ideals allows us to recover well-known topologies on $A^B$, such as the  topology of pointwise convergence (referred to as the local topology in this work) and the uniform topology (refer to \cite{Mun}, Section 19). 
Furthermore, it provides a convenient framework for defining new topologies, particularly for studying (infinitary) $\omega$-clones.

\vspace{-0.5cm}
\section{Topologies on $O_A^{(\omega)}$}
An ideal $X$ on $A^\omega$ determines the so-called $X$-topology on $O_A^{(\omega)}$, as described in Section \ref{sec:topo}.  
In this section we provide four examples of $X$-topologies on $O_A^{(\omega)}$: the local, global, trace and uniform topologies, respectively.  These topologies  are defined by the following ideals: 
\begin{itemize}
\item The \emph{local ideal} $L=\{d\subseteq A^\omega : |d|<\omega\}$ defines the local topology; 
\item The \emph{global ideal} $G=\{d\subseteq A^\omega : |d|\leq\omega\}$ defines the global topology; 
\item %Let $E$ be an equivalence relation on $A^\omega$. 
  The \emph{trace ideal}  (see Section \ref{sec:tra} for the definition of trace)
  $$ T= \{ d\subseteq A^\omega:  |d|\leq\omega\ \text{and $\exists$ a compact trace $\mathsf c$ s.t.}\ d\subseteq \mathsf c\}$$ gives rise to the trace topology.
\item The \emph{uniform ideal}
$$ U=\{d\subseteq A^\omega : |d|\leq\omega\ \text{and $\forall$ basic trace $\mathsf c$},\ |\mathsf c\cap d|<\omega\}$$
determines the uniform topology.

%\item  $T_A = \{ d\subseteq A^\omega: |d|\leq\omega\ \text{$\exists$ a compact trace $\mathsf c$ s.t.}\ d\subseteq \mathsf c\}$ for the trace topology; 
%\item $U_A=\{d\subseteq A^\omega : |d|\leq\omega, \text{$\forall$ compact trace $\mathsf c$}\ |\mathsf c\cap d|<\omega\}$ for the uniform topology.
%\item If $\tau$ is a topology on $A^\omega$, then $C_{A,\tau}=\{  d\subseteq A^\omega :  \text{$\exists$ a compact subset $K$}$ $\text{of the space $(A^\omega,\tau)$ s.t. $d\subseteq K$}\}$ for the $\tau$-compact topology
\end{itemize}

%Given $C\subseteq O_A^{(\omega)}$, we denote by $\overline C$ the local closure of $C$, by  $\widetilde C$ its global closure, by $\widehat C$ its trace closure and by $\vec C$ its uniform closure.
%
%\begin{lemma}\label{lem:clos}
%  For every set $C$ of $\omega$-operations on $A$, we have:
%  \begin{itemize}
%  \item[(i)]  $\overline C$ is  globally, trace and uniform closed;
%
%  \item[(ii)] $\vec C$ is  globally closed;
%   
%    \item [(iii)] $\widehat C$ is globally closed. 
%  \end{itemize}
%Then we have: $ \widetilde C\subseteq \vec C \subseteq \overline C$.
%and   $\widetilde C \subseteq \widehat C \subseteq \overline C.$
%  
%\end{lemma}
%\begin{proof} We have: $L_A\subseteq G_A\cap T_A\cap U_A$ , $U_A\subseteq G_A$ and $T_A\subseteq G_A$. The conclusion follows from Lemma \ref{lem:X}(2). 
%\end{proof}

We recall from Section \ref{sec:topo} that the local closure of $C$ is denoted by $\mathrm{Cl}_L(C)$, the global closure by $\mathrm{Cl}_G(C)$, the trace closure by $\mathrm{Cl}_T(C)$, and the uniform closure by $\mathrm{Cl}_U(C)$.

Let $C$ be a set of $\omega$-operations on $A$.
We recall from Lemma \ref{lem:X}(2) that if $X\subseteq Y$ are ideals, then $\mathrm{Cl}_Y(C)\subseteq \mathrm{Cl}_X(C)$ and $\mathrm{Cl}_X(C)$ is  $Y$-closed.

\begin{lemma}\label{lem:clos}
  Let $C$ be a set of $\omega$-operations on $A$. Then
  $\mathrm{Cl}_L(C)$ is  globally, trace, uniform closed, and $\mathrm{Cl}_U( C), \mathrm{Cl}_T( C)$ are  globally closed:  $$ \mathrm{Cl}_G(C)\subseteq \mathrm{Cl}_U( C) \cap \mathrm{Cl}_T( C)\subseteq \mathrm{Cl}_U( C) \cup \mathrm{Cl}_T( C) \subseteq \mathrm{Cl}_L(C).$$ 
\end{lemma}

\begin{proof} We have: $L\subseteq    T\cap  U \subseteq T\cup  U\subseteq  G$. 
% The conclusion follows from Lemma \ref{lem:X}(2). 
\end{proof}

Notice that if $C$ is locally closed, then $C= \mathrm{Cl}_G(C)= \mathrm{Cl}_U( C) = \mathrm{Cl}_T( C) = \mathrm{Cl}_L(C)$.

The operation of  $X$-closure preserves the $\omega$-clones, if the ideal $X$ that determines the topology enjoys the property expressed in the following definition.

\begin{definition}\label{def:subst}
An ideal $X$ on $A^\omega$ is 
\begin{enumerate}
\item \emph{substitutive}  if for every $d\in X$, $n\in\omega$ and $\boldsymbol\varphi\in (O_A^{(\omega)})^n$,
we have $d_{\boldsymbol\varphi}=\{s[\varphi_0(s),\dots, \varphi_{n-1}(s)]\ |\ s\in d\}\in X$.
\item  \emph{infinitely substitutive} if for every $d\in X$ and $\boldsymbol\varphi\in (O_A^{(\omega)})^\omega$,
we have $d_{\boldsymbol\varphi}=\{(\varphi_i(s): i\in\omega)\ |\ s\in d\}\in X$.
\end{enumerate}
\end{definition}

Notice that an infinitely substitutive ideal is also substitutive, because one can choose $\varphi_k=\e_k$ for all $k\geq n$.

\begin{proposition}\label{prop:sub} Let $X$ be an ideal on $A^\omega$ and $C$ be a set of $\omega$-operations. Then we have:
\begin{enumerate}
\item If $X$ is substitutive and $C$ is an $\omega$-clone, then $\mathrm{Cl}_X(C)$ is an $\omega$-clone.
\item If $X$ is infinitely substitutive and $C$ is an infinitary $\omega$-clone, then $\mathrm{Cl}_X(C)$ is an infinitary $\omega$-clone.
\end{enumerate}
\end{proposition}

\begin{proof} We will prove the first statement. The proof of the second one is similar.
Let $\psi,\psi_0,\ldots,\psi_{n-1}\in \mathrm{Cl}_X(C)$. We want to show that $q_n^A(\psi,\boldsymbol\psi)\in \mathrm{Cl}_X(C)$.
Let $d \in X$. By the hypothesis, there exist $\varphi_i\in C$ for $0\leq i\leq n-1$ such that $(\psi_i)_{|d}= (\varphi_i)_{|d}$. Additionally, we have $d_{\boldsymbol\psi}\in X$ since $X$ is substitutive. This implies that there exists $\varphi\in C$ such that $\psi_{|d_{\boldsymbol\psi}}= \varphi_{|d_{\boldsymbol\psi}}$.
Now, for every $s\in d$, we have:
\[
\begin{array}{lll}
q_n^{A}(\psi,\boldsymbol\psi)(s)  &  = &  \psi(s[\boldsymbol \psi(s)])\\
& =&  \varphi(s[\boldsymbol \psi(s)])\ \text{by $s[\boldsymbol \psi(s)]\in d_{\boldsymbol\psi}$}\\
& =  &   \varphi (s[\boldsymbol \varphi(s)])\ \text{by $s\in d$}\\
& =  &   q_n^{A}(\varphi,\boldsymbol\varphi)(s).  \\
\end{array} 
\]
Therefore, $q_n^{A}(\psi,\boldsymbol\psi)_{|d}= q_n^{A}(\varphi,\boldsymbol\varphi)_{|d}$. Moreover, we know that $q_n^{A}(\varphi,\boldsymbol\varphi)\in C$ since $C$ is an $\omega$-clone.
Since $d$ was chosen arbitrarily from $X$, we can conclude that $q_n^A(\psi,\boldsymbol{\psi}) \in \mathrm{Cl}_X(C)$.
\end{proof}

Both the local ideal and the global ideal are examples of infinitely substitutive ideals, while the trace ideal and the uniform ideal are substitutive. Therefore, we can derive the following corollary.

\begin{corollary} \label{prop:44}
(i) If $C$ is an $\omega$-clone, then $\mathrm{Cl}_L(C)$, $\mathrm{Cl}_G(C)$, $\mathrm{Cl}_U( C)$ and $\mathrm{Cl}_T( C)$   are  $\omega$-clones.

(ii)   If $C$ is an infinitary $\omega$-clone, then $ \mathrm{Cl}_L(C)$ and $\mathrm{Cl}_G(C)$  are infinitary $\omega$-clones.

\end{corollary}

\begin{example}\label{exa:tilde3} The $\omega$-clone $P=\{\e_n^A: n\geq 0\}$ of all projections on $A$ is locally, globally, trace, and uniform closed.
\end{example}

\begin{proposition}\label{prop:ftop} Let $O_A$ be the set of all finitary operations on $A$ and $O_A^{(\omega)}$ be the set of all $\omega$-operations on $A$. Then we have:
  \begin{itemize}
  \item[(i)] $\mathrm{Cl}_L(O_A^\top)=O_A^{(\omega)}$.
  \item[(ii)] %$O_A^\top\subsetneq\overrightarrow{(O_A)^\top}\subsetneq O_A^{(\omega)}$.
    $\mathrm{Cl}_U(O_A^\top)\subsetneq O_A^{(\omega)}$.

    % \item[(iii)]  $\widetilde{(O_A)^\top}\subsetneq O_A^{(\omega)}$.
% $\overrightarrow{(O_A)^\top}\subsetneq O_A^{(\omega)}$.
    
    \item[(iii)]  For a countable set $A$, we have $\mathrm{Cl}_T(F^\top) = F^\top = \mathrm{Cl}_G(F^\top)$ for every clone $F$ on $A$. Consequently, $\mathrm{Cl}_T(O_A^\top) = O_A^\top = \mathrm{Cl}_G(O_A^\top)$.
    %If $A$ is countable, then $\widehat{F^\top}=F^\top = \widetilde{F^\top}$ for every clone $F$. Hence, $\widehat{(O_A)^\top}=(O_A)^{\top} = \widetilde{(O_A)^\top}$.
     \end{itemize}
\end{proposition}

\begin{proof}
\begin{itemize}
  \item[(i)]  
Let $\varphi: A^\omega\to A$ be an $\omega$-operation.   To prove that $\varphi \in \mathrm{Cl}_L(O_A^\top)$, we proceed as follows. Let $d=\{d_0,\dots,d_{k-1}\}\subseteq_{\mathrm{fin}}A^\omega$.
We define $n$ as the least natural number such that $(d_i^0,\dots,d_i^{n-1}) \neq (d_j^0,\dots,d_j^{n-1})$ for every $0 \leq i\neq j \leq k-1$. Next, we define a finitary operation $f: A^n \to A$ as follows: Fix $a \in A$.
$$f(a_0,\dots,a_{n-1})=\begin{cases}\varphi(d_i)&\text{if $\exists i\in\{0,\dots,k-1\}$ s.t. $d_i=d_i[a_0,\dots,a_{n-1}]$}\\a&\text{otherwise.} \end{cases}$$ 
Thus $f^\top[d] = \varphi[d]$. By considering the arbitrary choice of $d$, we conclude that $\varphi \in \mathrm{Cl}_L(O_A^\top)$.

%  Let $\varphi: A^\omega\to A$ be an $\omega$-operation. We prove that $\varphi\in \overline{(O_A)^\top}$.  
% Let $k\in\omega$ and  $d \in\mathsf M^\omega_k$.  
% Let $n$ be the least natural number such that $(d_i^0,\dots,d_i^{n-1})\neq (d_j^0,\dots,d_j^{n-1})$ for every $0\leq i,j\leq k-1$ such that $d_i\neq d_j$.
%Let $f:A^n\to A$ be a finitary operation defined as follows. Fix $a\in A$.
%$$f(a_0,\dots,a_{n-1})=\begin{cases}\varphi(d_i)&\text{if $\exists 0\leq i\leq k-1$ s.t. $d_i=d_i[a_0,\dots,a_{n-1}]$}\\a&\text{otherwise} \end{cases}$$
%Since $f^\top[d]= \varphi[d]$, by the arbitrariness of $d$ we have that  $\varphi\in\overline{(O_A)^\top}$.
\item[(ii)]  
Let $a\neq b\in A$, and let $\psi\in O_A^{(\omega)}$ be the function defined as follows, for every $r\in A^\omega$:
\begin{equation}\psi(r) = \begin{cases} a&\text{if $|\{i : r_i=b\}|=0$}\\
b&\text{otherwise.}\end{cases}\end{equation}
Consider the matrix $m\in \mathsf M_\omega^\omega$ defined as follows: $m_0=a^\omega$ and, for $i>0$,
\begin{equation} m^j_i= \begin{cases} b&\text{if $j$ is a power of the $i$-th prime number}\\
a& \text{otherwise.}\end{cases}\end{equation}
The set of rows of $m$ belongs to the uniform ideal $U$, because no pair of rows is included in the same basic trace. Therefore, any compact trace can contain only a finite number of rows of $m$.

\noindent We observe that $\psi[m]= ab^\omega$. However, there is no finitary operation $f\in O_A$ such that $f^\top[m]=ab^\omega$. Consequently, we conclude that $\psi\not\in \mathrm{Cl}_U(O_A^\top)$.

\item[(iii)] 
By the given hypothesis that every compact trace is countable, we will now prove that $\mathrm{Cl}_T(F^\top)= F^\top$.

Let $\varphi \in \mathrm{Cl}_T(F^\top)$, and let $\mathsf c\subseteq \mathsf d$ be nonempty compact  traces. Since $\mathsf c$ and  $\mathsf d$ are countable, they belong to the trace ideal $T$. Then there exist $f\in F$ of minimal arity $k$ and $g\in F$ of minimal arity $n$ such that $\varphi_{|\mathsf c}= (f^\top)_{|\mathsf c}$ and $\varphi_{|\mathsf d}= (g^\top)_{|\mathsf d}$. Since $\mathsf c\subseteq \mathsf d$, it follows that $k\leq n$.

We will now show that $f^\top=g^\top$. Let $s\in A^\omega$ and $u\in\mathsf c$. Then, $u[s_0,\ldots,s_{n-1}]\in \mathsf c$, and
%\begin{align*}
$g^\top(s)=g(s_0,\ldots,s_{n-1})
=g^\top(u[s_0,\ldots,s_{n-1}])
=\varphi(u[s_0,\ldots,s_{n-1}])
=f^\top(u[s_0,\ldots,s_{n-1}])
=f(s_0,\ldots,s_{k-1})
=f^\top(s)$.
%\end{align*}
From the arbitrariness of $\mathsf d$, it follows that $\varphi = f^\top$. Therefore, $\varphi \in F^\top$, and we conclude that $\mathrm{Cl}_T(F^\top)= F^\top$.

Next, we will prove that $\mathrm{Cl}_G(F^\top)= F^\top$. The conclusion follows from Lemma \ref{lem:clos} (which states, in particular, that for all $C\subseteq O_A^{(\omega)}$, $\mathrm{Cl}_G(C)\subseteq\mathrm{Cl}_T(C)$) and from $\mathrm{Cl}_T(F^\top)= F^\top$.
  \qedhere
\end{itemize}
\end{proof}

\section{Polymorphisms and invariant relations}
In the classical approach to polymorphisms and invariant relations, topology, especially local topology, plays an important role in the set $O_A$ of finitary operations. However, it seems to have no role in the set $Rel_A$ of finitary relations. The reason is that while the exponent $A^k$ in $A^{A^k}$ is generally infinite, the exponent $n$ in $A^n$ is always finite. Therefore, the local topology on $A^{A^k}$ is generally non-trivial, while the local topology on $A^n$ is discrete. 

In the infinitary case, topology plays an important role not only in the set $O_A^{(\omega)}$ of $\omega$-operations but also in the set $A^\omega$.
%For instance, if $F$ is a clone,  we will prove that $F=Pol(Inv^\omega(F))$ if and only if $F ^\top$ is globally closed, and  $F=Pol(Inv^\omega_c(F))$ if and only if $F^\top$
%is locally closed, where $Inv^\omega(F)$ (resp. $Inv_c^\omega(F)$) is the set of invariant
%(resp. locally-closed invariant)  $\omega$-relations for $F$.
The framework we will develop in this section enables us to define, given an ideal $X$ on $A^\omega$, an ideal $\mathcal X(m)$ on  $\alpha\leq\omega$, which depends on a matrix $m\in\mathsf M^\omega_\alpha$. As described in Section \ref{sec:topo}, the ideal $\mathcal X(m)$ defines a topology on $A^\alpha$. These parametric topologies will be used to define the concepts of polymorphism and invariant relation, both of which are parametric with respect to $X$.
 
 In the following we assume that every ideal $X$ on $A^\omega$ satisfies the following condition:
$$L \subseteq X\subseteq G,$$
where $L$ and $G$ are the local ideal and the global ideal respectively.

\subsection{Topologies on $A^\omega$}\label{sec:topaalpha}

%Hereafter, we write $[\mathcal P_{\mathrm{fin}}(\alpha),X_m]$ the interval of ideals on $\alpha$ between $\mathcal P_{\mathrm{fin}}(\alpha)$ and $X_m$.

In the following we assume without mention that every ideal $I$ on $\alpha$ ($\alpha\leq\omega$) satisfies the following condition:
\begin{equation}\label{eq:ide} I \supseteq \mathcal P_{\mathrm{fin}}(\alpha).\end{equation}
We denote by $\mathrm{Ide}(\alpha)$ the set of ideals $I$ on $\alpha$ satisfying (\ref{eq:ide}). As a matter of notation, we define $\mathrm{Ide}=\bigcup_{\alpha\leq\omega}\mathrm{Ide}(\alpha)$.

Let us introduce a general framework to describe topologies on $A^\alpha$ for $\alpha\leq\omega$.

\begin{definition} Let $\mathsf M^\omega$ denote the set of all matrices on $A$ with $\omega$ columns. We say that a
   function  $\mathcal I: \mathsf M^\omega\to \mathrm{Ide}$ is an {\em ideal map} if for each  $m\in \mathsf M_\alpha^{\omega}$, we have $\mathcal{I}(m)\in \mathrm{Ide}(\alpha)$.
\end{definition}

The set of all ideal maps is denoted by $\mathsf{I}_A$. We define a partial ordering on $\mathsf{I}_A$: $\mathcal{I} \leq \mathcal{I}'$ if $\mathcal{I}(m) \subseteq \mathcal{I}'(m)$ for all $m$. The set $\mathsf{I}_A$ with this ordering forms a complete lattice, with the bottom element being the ideal map $\bot$ defined by $\bot(m) = \mathcal{P}_{\mathrm{fin}}(\alpha)$ for every $m \in \mathsf{M}_\alpha^\omega$.

As described in Section \ref{sec:topo}, the ideal $\mathcal{I}(m)$ defines a topology on $A^\alpha$, which we henceforth refer to as the $\mathcal{I}_m$-topology. When $\alpha$ is finite, this topology is discrete because $\mathrm{Ide}(\alpha) = \{\mathcal{P}(\alpha)\}$. However, if $\alpha = \omega$, the $\mathcal{I}_m$-topology is not necessarily discrete. The closure of $R \subseteq A^\alpha$ under the $\mathcal{I}_m$-topology is denoted by $\mathrm{Cl}_{\mathcal{I}_m}(R)$.

\bigskip

Every ideal $X$ on $A^\omega$ determines an ideal map $\mathcal X$ defined as follows for every  $m\in \mathsf M_\alpha^{\omega}$: $$\mathcal X(m) = \{ c\subseteq\alpha: \{m_i: i\in c\} \in X\}.$$

The ideal maps of the form $\mathcal X$ for an   ideal $X$ on $A^\omega$ are  called {\em canonical}. The set of canonical ideal maps is denoted by $\mathsf C_A$.

\begin{remark}
 (1)  Let $G$ be the global ideal on $A^\omega$ and  $\mathcal G$ be its  canonical ideal map. For every $m\in \mathsf M_\alpha^\omega$,  $\mathcal G(m) = \mathcal P(\alpha)$, so that  the $\mathcal G_m$-topology is the discrete topology on $A^\alpha$ for every $\alpha \leq \omega$.

 (2)   Let $L$ be the local  ideal on $A^\omega$ and $m\in \mathsf M^\omega_\omega$. The ideal $\mathcal L(m)$  may contain infinite subsets. For example, if $c\subseteq \omega$, $|c|=\omega$ and $|\{m_i : i\in c\}|<\omega$, then $c\in \mathcal L(m)$. In this case, the $\mathcal L_m$-topology on $A^\omega$ is different from the local topology on $A^\omega$ defined in Section \ref{sec:loctop}.
 
\end{remark}

The following lemma highlights the lattice structure and closure properties of $\mathsf{C}_A$ and its relationship with $\mathsf I_A$.

\begin{lemma}\label{lem:canon}
  \begin{itemize}
  
  \item[(i)] The set $\mathsf{C}_A$ forms a complete lattice, and for any $\Gamma \subseteq \mathsf{C}_A$, the infimum (meet) of $\Gamma$ in $\mathsf{C}_A$ is the same as its infimum in $\mathsf{I}_A$. The ideal map $\mathcal{L}$ is the bottom element of $\mathsf{C}_A$.

%    \item[(ii)]  The map from the interval $[L,G]$ of ideals on $A^\omega$ to $\mathsf C_A$ defined by $X\mapsto\mathcal X$ is a homomorphism of complete lattices.

  \item[(ii)]   The operator $(-)^+:\mathsf I_A \to \mathsf C_A$, defined by $$\mathcal I^+ =\bigwedge \{\mathcal X:\ \text{$X$ is an ideal on $A^\omega$ and $\mathcal I\leq \mathcal X$}\},$$
    is a closure operator.
  \end{itemize}
  \end{lemma}
Notice that $\bigvee_{\mathsf I_A} \Gamma\leq\bigvee_{\mathsf C_A} \Gamma $  and in general $\bigvee_{\mathsf C_A} \Gamma\neq\bigvee_{\mathsf I_A} \Gamma $.

\begin{remark}
 % The minimal ideal map $\bot\in\mathsf I_A$ is invariant regardless of the specific matrix $m\in \mathsf M_\alpha^\omega$ ($\alpha\leq\omega$) chosen. The $\bot_m$-topology on $A^\omega$ corresponds to the local topology  defined in Section \ref{sec:loctop}.

  The minimal ideal map $\bot$ of $\mathsf I_A$ is not canonical. We have that
  $\bot^+=\mathcal L$, where $\mathcal L$ is the canonical ideal map determined by the local ideal $L$ on $A^\omega$.
 \end{remark}

\subsection{Interlude: the categorical approach}\label{sec:int}
This section  shows that the
notion of canonical ideal maps is natural in the sense that these maps  appear as
the object part of some functor in suitably defined categories. 
%In reality, this section is not essential to continue reading the work, and in particular, the last part of the section contains ideas that could be developed at a later time."
%Ideal maps are functions from the set of all matrices over  $A$ to the set
%of all ideals over $\alpha\leq\omega$.
We define the categories
$\mathbb M$ of matrices over $A$ and $\mathbb{I}de$ of ideals over $\alpha\leq\omega$ in such a way that the canonical ideal maps turn out to be exactly the object parts of presheaves on $\mathbb M$, i.e. of functors
$\mathbb M^{\mathsf{op}}\to \mathbb{I}de\hookrightarrow {\mathbf{Sets}}$.

In the last part of the section, rather exploratory,  we address the question of  whether these functors are sheaves for a suitably defined Grothendieck topology on $\mathbb M$.
%Interestingly, 

%The canonical ideal maps may be characterised as the functorial ones, for suitable categorsations of the 

 If $m\in \mathsf M^\omega_\alpha$ is a matrix, then we write $dim(m)=\alpha$.

\begin{definition}
The category $\mathbb M$ has the set $\mathsf M^\omega$ of matrices on $A$ as set of objects, and for $m,p\in\mathsf M^\omega$, $R\in \mathbb M(m,p)$ if  $R\subseteq dim(m)\times dim(p)$ is a binary relation  such that $m_i=p_j$ for every $(i,j)\in R$.
\end{definition}
The composition and identities  in $\mathbb M$ are the relational ones.

Given a binary relation $R\subseteq A\times B$, let $f_R:\mathcal P(A)\to \mathcal P(B)$ be the function defined by
$f_R(Y)=\{b\in B:\exists a \in Y \ (a,b)\in R\}$, for every $Y\subseteq A$.

\begin{definition}
  The category $\mathbb{I}de\subseteq\mathrm{\bf Sets}$ has the ideals in
  $\bigcup_{\alpha\leq \omega}\mathrm{Ide}(\alpha)$ as objects, and
  for $ I\in \mathrm{Ide}(\alpha)$, $ J\in \mathrm{Ide}(\beta)$,
  $\mathbb{I}de( I,  J)=\{(f_R)_{| I}: R\subseteq \alpha\times\beta \text{ and } \forall c\in  I\ f_R(c)\in J\}$.
\end{definition}
The identities and the composition in $\mathbb{I}de$ are those of $\mathrm{\bf Sets}$.

\begin{proposition}\label{prop:can}
 An ideal map is canonical if and only if it is the object part of a functor from $\mathbb M^{op}$ into $\mathbb{I}de$.
\end{proposition}
\begin{proof}

  ($\Rightarrow$)
  Let $\mathcal X$ be the canonical ideal map associated to the ideal $X$ on $A^\omega$. We extend this ideal map to a functor  $\mathcal X : \mathbb M^{op}\to \mathbb{I}de$ as follows: for $R\in \mathbb M (m,p)$, $\mathcal X(R)\in \mathbb{I}de(\mathcal X(p), \mathcal X(m))$ is defined by $\mathcal X(R)= (f_{R^{-1}})_{|\mathcal X(p)}$. Indeed, if $d\in \mathcal X(p)$, then
  $\{m_i: i\in f_{R^{-1}}(d)\}\subseteq \{p_j:j\in d\}\in X$, since $\forall (i,j)\in R\ m_i=p_j$,  hence, $f_{R^{-1}}(d)\in \mathcal X(m)$.

 ($\Leftarrow$)  Let $\mathcal F$ be a functor from $\mathbb M^{op}$ into $\mathbb{I}de$.
We define $$J=\{t\subseteq A^\omega\ |\ \exists m\in\mathsf M^\omega\ \exists c\in \mathcal F(m):  t=\{ m_i: i\in c\}\}$$
and we demonstrate that $J$ is an ideal on $A^\omega$ such that $\mathcal J(m)= \mathcal F(m)$ for every $m\in\mathsf M^\omega$.
 Since the ideal $\mathcal F(m)$ is down-closed for every matrix $m$, then $J$ is down-closed.

 Let $t_i\in J$ for $i=1,2$. We will prove that $t_1\cup t_2\in J$. There exist a matrix $mi\in\mathsf M_{\alpha_i}^\omega$ and a set $c_i\in \mathcal F(mi)$ such that $t_i=\{ (mi)_j: j\in c_i\}$. Let $w\in\mathsf M_\omega^\omega$ be any matrix such that $\mathrm{row}(w)=t_1\cup t_2$.
 %and, for every $s\in t_1\cup t_2$, $|\{k: w_k=s\}| \geq |\{k: m1_k=s\}|+  |\{k: m2_k=s\}|$. 
 Let $R_i: w\to mi$ ($i=1,2$) be an arrow in $\mathbb M$ such that $R_i=\{(k,j): w_k= (mi)_j\ \text{and}\ j\in c_i\}$. Therefore, $\{w_k:\exists j \ (k,j)\in R_i\}=t_i$.
Hence, $\mathcal F(R_i): \mathcal F(mi)\to\mathcal F(w)$ is an arrow in the category $\mathbb{I}de$ such that $\mathcal F(R_i)(c_i)= f_{R_i^{-1}}(c_i)=\{n\in\omega: w_n\in t_i\}\in \mathcal F(w)$.
Since $\mathcal F(w)$ is an ideal,  we have that $f_{R_1^{-1}}(c_1)\cup f_{R_2^{-1}}(c_2)\in \mathcal F(w)$ and
$t_1\cup t_2 =\{w_k: k \in f_{R_1^{-1}}(c_1)\cup f_{R_2^{-1}}(c_2)\}\in J$.
We have concluded the proof that $J$ is an ideal on $A^\omega$.

We now show that the ideal map $\mathcal J$ determined by the ideal $J$ on $A^\omega$ satisfies the following condition: $\mathcal J(m)=\mathcal F(m)$ for every matrix $m\in \mathsf M^\omega$. First, $\mathcal J(m)\supseteq \mathcal F(m)$ because  given $c\in\mathcal F(m)$, the set $t=\{ m_i:i\in c\}\in J$, and hence
$c\subseteq \{j: m_j\in t \}\in \mathcal J(m)$.
Concerning the other inclusion, 
let $c\in \mathcal J(m)$. Therefore, $t=\{m_j: j\in c\}\in J$ and, by definition of $J$,
$\exists p\in\mathsf M^\omega$
$\exists d\in \mathcal F(p)$ such that $\{ p_i: i\in d\}= \{m_j: j\in c\}\in J$.
Let $S: m\to p$ be an arrow such that $S=\{(i,j): m_i=p_j, i\in c, j\in d\}$. Then, we have $\mathcal F(S): \mathcal F(p)\to\mathcal F(m)$ with $\mathcal F(S)(d)=c\in \mathcal F(m)$.
  \end{proof}

 The functor $\mathcal X:\mathbb M^{\mathrm {op}}\to \mathrm{\bf Sets}$
defined in the previous proposition is a presheaf for every  ideal $X$ on $A^\omega$.

In the remaining part of this section we investigate the possibility that these presheaves are actually sheaves.
 
Consider the functor $\mathbf y:\mathbb M\to \mathrm{\bf Sets}^{\mathbb M^{\mathrm{op}}}$ from the category $\mathbb M$ to the category of contravariant functors on $\mathbb M$. It is defined as follows:
$$\mathbf y(p)= \mathbb M(- ,p): \mathbb M^{\mathrm{op}} \to \mathrm{\bf Sets},$$
where $\mathbb M(m ,p)$ is the set of all arrows from $m$ to $p$.
 The functor $\mathbf y$ is well known and called the Yoneda embedding (see p.25 in \cite[Chapter I(1)]{MM92}). 
 
We recall that a \emph{sieve} $S$ on $p\in\mathsf M^\omega$ is a subfunctor of the Hom-functor $\mathbf y(p): \mathbb M^{\mathrm{op}} \to \mathrm{\bf Sets}$ (see p.38 in \cite[Chapter II(4)]{MM92}). 

If $X$ is an ideal on $A^\omega$, then we define a functor $$\mathsf X:\mathbb M\to \mathrm{\bf Sets}^{\mathbb M^{\mathrm{op}}}.$$ 
For every matrix $p$, the functor $\mathsf X_p:\mathbb M^{\mathrm{op}} \to \mathrm{\bf Sets}$ is a sieve on $p$ defined as follows:
 \begin{itemize}
   \item Let  $m$ be a matrix. Then
  $$\mathsf X_p(m)= \{ S\in \mathbb M(m,p)\ |\ (\exists c\in\mathcal X(m), d\in\mathcal X(p))\ S=\{(i,j) \in c\times d : m_i=p_j\}\}. $$
\item Let $R\in \mathbb M(m,q)$. Then for every relation $S\in \mathsf X_p(q)$:
  $$\mathsf (X_p(R))(S)=S\circ R \in \mathsf X_p(m).$$
 \end{itemize}

\begin{lemma}
 The definition of $\mathsf X_p$ is well-done.
\end{lemma}
\begin{proof}
  The definition is well-done if there exist $c\in \mathcal X(m)$ and $d\in \mathcal X(p)$ such that $S\circ R=\{(i,j) \in c\times d : m_i=p_j\}$. Since $S\in \mathsf X_p(q)$, there exist $ e\in\mathcal X(q), d\in\mathcal X(p)$ such that $S=\{(i,j) \in e\times d : q_i=p_j\}$. Let $(i,j)\in S\circ R$. Then there exists $u$ such that $(i,u)\in R$ and $(u,j)\in S$. Therefore, $u\in e$, $j\in d$ and $m_i=q_u=p_j$. Since $j\in d$, then $d'=\{j: \exists i\ (i,j)\in S\circ R\}\subseteq d \in \mathcal X(p)$ and $\{p_j : j\in d'\}\in X$.
  It follows that $\{m_i: \exists j\in d'\ (i,j)\in S\circ R\}= \{p_j : j\in d'\}\in X$. We conclude that $c'=\{i: \exists j\ (i,j)\in S\circ R\}\in X$.
  Then $S\circ R= \{(i,j) \in c'\times d' : m_i=p_j\}$.
\end{proof}
 
For every $p\in\mathsf M^\omega$ we define $J(p)= \{ \mathsf X_p :\ \text{$X$ is an ideal on $A^\omega$}\}$ to be a family of sieves on $p$.
This family generates a least Grothendieck topology $J^\star$ on the category $\mathbb M$ (see \cite[Chapter III(2)]{MM92}  and p.157 in \cite[Chapter III, Exercise 6]{MM92}). The open question is whether the functors $\mathcal X$ are sheaves for this
Grothendieck topology $J^\star$.

\subsection{$\mathcal I$-polymorphisms}

We now introduce the concept of an $\mathcal I$-polymorphism for an ideal map $\mathcal I$. We recall from Section \ref{sec:topaalpha} that, for every $m\in \mathsf M_\alpha^{\omega}$, the ideal $\mathcal I(m)$ defines the $\mathcal I_m$-topology on $A^\alpha$.

\begin{definition}\label{def:pol} Let  $R\subseteq A^\alpha$ and   $\mathcal I$ be an ideal map. We say that an $\omega$-operation $\varphi$ is an \emph{$\mathcal I$-polymorphism of $R$} if 
   for every matrix $m\in \mathsf M_\alpha^{\omega,R}$, $\varphi[m]\in \mathrm{Cl}_{\mathcal I_m}(R)$.
\end{definition}

  We denote by $Pol^\omega_{\mathcal I}(R)$ the set of all  $\mathcal I$-polymorphisms of $R$.
If $\mathcal R$ is a set of relations, we define  $Pol^\omega_{\mathcal I}(\mathcal R) = \bigcap_{R\in\mathcal R} Pol^\omega_{\mathcal I}(R)$.

%\blue{\begin{remark}
%If (a given primive positive formula on $A$) (certain atomic formulas)  hold(s) for
%each column of $m\in  \mathsf M_\alpha^{\omega,R}$, then (it) (they) will also hold in $\varphi[m]$.
%All  symmetries existing in the matrix $m$ affect $\varphi[m]$.
%For example, if $m_i=m_j$ for all even numbers $i,j$, then $\varphi[m]_i=\varphi[m]_j$ for all  even numbers $i,j$.
%Hence, the set of  $\mathcal I$-polymorphisms of $R$...
%To require $\varphi[m]\in  \mathrm{Cl}_{\mathcal I_m}(R)$ means looking for the same symmetries in $R$, but each symmetry possibly in a different element of $R$.
%\end{remark}}

Notice that the map $\mathcal I\mapsto Pol^\omega_\mathcal I(R)$ is order reversing.

\begin{remark}
  Among the possible choices of $\mathcal I$ in the definition of polymorphism, we have:

\begin{itemize}
\item The minimal one: $\bot(m)=\mathcal P_{\mathrm{fin}}(\alpha)$ for all $m\in \mathsf M_\alpha^{\omega}$. The corresponding notion is:
   $\varphi$ is a $\bot$-polymorphism of $R$ if 
   for every matrix $m\in \mathsf M_\alpha^{\omega,R}$, $\varphi[m]\in\overline R$.
 \item The maximal one is $\mathcal G(m)=\mathcal P(\alpha)$ for all $m\in \mathsf M_\alpha^{\omega}$.
The corresponding notion is:
   $\varphi$ is a $\mathcal G$-polymorphism of $R$ if 
   for every matrix $m\in \mathsf M_\alpha^{\omega,R}$, $\varphi[m]\in R$.
  \end{itemize}
\end{remark}

We now introduce a map that will be useful in the following.
Any matrix $m\in \mathsf M_\alpha^\omega$ induces a mapping $\bar m$ from $O_A^{(\omega)}$ into $A^\alpha$ in the following manner:
$$\bar m(\varphi)= \varphi[m],\ \text{for every $\varphi\in O_A^{(\omega)}$}.$$

\begin{lemma}\label{lem:cont} Let $X$ be an ideal on $A^\omega$. %and $\mathcal I$ be an ideal map.
%For every ideal $I\in[\mathcal P_{\mathrm{fin}}(\alpha), \mathcal X_m]$, 
The map $\bar m$ is continuous from the $X$-space $O_A^{(\omega)}$ into the $\mathcal X_m$-space $A^\alpha$.
\end{lemma}

\begin{proof}
Let $R\subseteq A^\alpha$ be $\mathcal X_m$-closed.  We aim to demonstrate that $H=\bar{m}^{-1}(R)$ is $X$-closed.
Let $\varphi$ be an element of the $X$-closure of $H$. To establish the desired result, it suffices to show that $\varphi[m]\in R$. By the definition of $\mathcal X(m)$, for every $c\in \mathcal X(m)$, we have $\bar c=\{m_i : i\in c\}\in X$.
As $\varphi\in\mathrm{Cl}_X(H)$, there exists $\psi\in H$ such that $\varphi_{|\bar c}=\psi_{|\bar c}$. This implies that $\varphi[m]_{|c}= \psi[m]_{|c}$. Since $\psi[m]\in R$ and $R$ is $\mathcal X_m$-closed, we can conclude $\varphi[m]\in R$ by considering the arbitrary choice of $c$. 
\end{proof}

Let $\mathcal I$ be an ideal map. By definition of $\mathcal I$-polymorphism we have:
\begin{equation}\label{polpol} Pol^\omega_{\mathcal I}(R)= \bigcap \{ \bar m^{-1}(\mathrm{Cl}_{\mathcal I_m}(R)):m\in\mathsf M_\alpha^{\omega,R}\}.
\end{equation}

\begin{proposition}\label{prop:ii}
Let $X$ be an  ideal on $A^\omega$ and $R\subseteq A^\alpha$. %$\mathcal I$ be a $X$-ideal map.
\begin{itemize}
\item[(i)]  The set $Pol^\omega_{\mathcal I}(R)$ is $X$-closed for every ideal map ${\mathcal I}\leq{\mathcal X}$.
\item[(ii)] The set $Pol^\omega_{\mathcal I}(R)$ is an infinitary $\omega$-clone for every ideal map ${\mathcal I}\leq{\mathcal X}$.
 \item[(iii)] If $R$ is locally closed, then $R$ is $I$-closed for every $I\in\mathrm{Ide}(\alpha)$, and $Pol^\omega_{\mathcal I}(R)=Pol_\mathcal G^\omega(R)$ for every ideal map $\mathcal I$.
\end{itemize}
\end{proposition}

\begin{proof}

  (i): By Lemma \ref{lem:X}(2) we have that every $\mathcal I_m$-closed set is also $\mathcal X_m$-closed, because $\mathcal I(m)\subseteq \mathcal X(m)$.
According to (\ref{polpol}) and Lemma \ref{lem:cont}, $Pol^\omega_{\mathcal I}(R)$ is the intersection of a family of $X$-closed sets, which implies that it is itself $X$-closed.

  (ii):   Let $\psi,\varphi_i \in Pol^\omega_{\mathcal I}(R)$. We want to show that $q^A (\psi,\boldsymbol\varphi)[m]\in \mathrm{Cl}_{\mathcal I_m}(R)$ for every $m\in\mathsf M_\alpha^{\omega,R}$.
For every $i\in\omega$, we have
$$q^A(\psi,\boldsymbol\varphi)(m_i)   =   \psi(\varphi_0(m_i),\dots,\varphi_n(m_i),\dots).$$
We define a new matrix $p\in \mathsf M_\alpha^{\omega}$ column-wise as follows: 
$$p^n=\varphi_n[m]\in \mathrm{Cl}_{\mathcal I_m}(R).$$
Thus we have $q^A (\psi,\boldsymbol\varphi)[m] = \psi[p]$.
Since $p^n\in \mathrm{Cl}_{\mathcal I_m}(R)$ for every $n\in\omega$, then, for every $c\in \mathcal I(m)$, we can define a matrix $r\in\mathsf M_\alpha^{\omega,R}$ (depending on $c$) such that  $(r^n)_{|c} = (p^n)_{|c}$.

%For every $c\in \mathcal I_m$, we can define a matrix $r\in\mathsf M_\alpha^{\omega,R}$ (depending on $c$) such that $(r^n){|c} = (p^n){|c}$, since $p^n\in \mathrm{Cl}_{\mathcal I_m}(R)$ for every $n\in\omega$.

As $\psi$ is an $\mathcal I$-polymorphism of $R$ and $r\in\mathsf M_\alpha^{\omega,R}$, we have $\psi[r]\in \mathrm{Cl}_{\mathcal I_m}(R)$. From $\psi[r]_{|c}=\psi[p]_{|c} = q^A (\psi,\boldsymbol\varphi)[m]_{|c}$, and the fact that $c\in \mathcal I(m)$ was arbitrary, we conclude that $q^A (\psi,\boldsymbol\varphi)[m]\in \mathrm{Cl}_{\mathcal I_m}(\mathrm{Cl}_{\mathcal I_m}(R))=\mathrm{Cl}_{\mathcal I_m}(R)$.

(iii) By the inclusion $L\subseteq I$, we have that  
  $R\subseteq \mathrm{Cl}_{I}(R)\subseteq \overline R$.
  \end{proof}

\begin{notation}\label{not:0} If $\mathcal R$ is a set of locally closed (in particular finitary) relations then by Proposition \ref{prop:ii} $Pol^\omega_{\mathcal I}(\mathcal R)$ is independent of $\mathcal I$. In this case we will write $Pol^\omega(\mathcal R)$ for 
$Pol^\omega_{\mathcal I}(\mathcal R)$.
\end{notation}

We now introduce the concept of an invariant relation in various forms, depending on the arity and the closedness of the relations themselves.

\begin{definition}
  Let $C\subseteq O_A^{(\omega)}$ and $\mathcal I$ be an ideal map. We define: 
\begin{itemize}
\item $Inv^{\leq\omega}_\mathcal I(C)=\{R\subseteq A^\alpha\ |\ \alpha\leq\omega\ \text{and}\ C\subseteq Pol^\omega_\mathcal I(R)\}$; 
\item $Inv^{\omega}_\mathcal I(C)=\{R\subseteq A^\omega\ |\ C\subseteq Pol^\omega_\mathcal I(R)\}$; 
 %\item $Inv^{<\omega}(C)=\{R\subseteq A^\alpha\ |\ \alpha <\omega\ \text{and}\  C\subseteq Pol^\omega(R)\}$.
\item $Inv^{\omega}_{c,\mathcal I}(C)=\{R\subseteq A^\omega\ |\ C\subseteq Pol^\omega_\mathcal I(R),\ \text{$R$ is $\mathcal I_m$-closed for all $m\in  \mathsf M_\omega^{\omega,R}$}\}$;
\item $Inv^{<\omega}(C)=\{R\subseteq A^\alpha\ |\ \alpha <\omega\ \text{and}\  C\subseteq Pol^\omega(R)\}$;
  \item $Inv_c^{\omega}(C)=\{R\subseteq A^\omega\ |\ \text{$R=\overline R$  and   $C\subseteq Pol^\omega(R)$}\}$.
\end{itemize}
\end{definition}

\begin{example}
The set $Inv^{\omega}_{c,\bot}(C) = \{R\subseteq A^\omega\ |\ \text{$R=\overline R$  and   $C\subseteq Pol^\omega_\bot(R)$} \}=Inv_c^\omega(C)$ is the set of locally closed $\omega$-relations that are $C$-invariant. Observe that by Notation \ref{not:0} we have $Pol^\omega_\bot(\overline R)=Pol^\omega(\overline R)$ for $\omega$-relation $R$.
\end{example}

%\subsection{The finitary case: $\alpha < \omega$}

%  Let $m\in  \mathsf M_\alpha^\omega$ for $\alpha < \omega$. An extension of $m$ is any matrix $p \in \mathsf M_\omega^\omega$ such that $p_i=m_i$ for every $i\in\alpha$.
%  
%
%\begin{lemma}\label{lem:2.4} Let $S\subseteq A^\alpha$ be a finitary relation, $\varphi$ be an $\omega$-operation on $A$  and $m\in  \mathsf M_\alpha^\omega$. 
% \begin{enumerate}
%\item We have $\varphi[m]\in S$ iff $\varphi[p]\in S^\top$ %$\bar m^{-1}(S)=\bar p^{-1}(S^\top)$ 
%for some (and then every) extension $p \in \mathsf M_\omega^\omega$ of $m$.
%\item $Pol^\omega_\mathcal I(S)=Pol^\omega_{\mathcal I}(S^\top)= Pol^\omega_\mathcal G(S^\top)$, for every ideal map $\mathcal I$.
%\end{enumerate}
%\end{lemma}
%
%\begin{proof}
% %(1) Consider an $n$-ary relation $S$ and $s\in \mathrm{Cl}_X^m(S^\top)$. Since $X$ contains all finite subsets, we have $n=\{0,\dots,n-1\}\in X_m$. Therefore, there exists $u\in S^\top$ such that $s_i=u_i$ for every $0\leq i\leq n-1$. Hence, $s\in S^\top$.
% 
% (1) Trivial.
%  
% (2) A matrix $m\in \mathsf M_\alpha^{\omega,S}$ if and only if for every extension $p$ of $m$, $p\in \mathsf M_\omega^{\omega,S^\top}$. Hence, given $m\in \mathsf M_\alpha^{\omega,S}$, we have $\varphi[m]\in S$ if and only if for every extension $p \in \mathsf M_\omega^\omega$ of $m$, $\varphi[p]\in S^\top$. The conclusion follows because, by Proposition \ref{prop:ii}(iii), $S^\top$ is $\mathcal I_m$-closed for every matrix $m\in \mathsf M_\omega^\omega$.
%\end{proof}

The following proposition shows that a finitary relation can be freely substituted by its top extension.

\begin{proposition}\label{lem:2.5}
  Let $C\subseteq O_A^{(\omega)}$ and $\mathcal I$ be an ideal map. Then
  $Pol^\omega_{\mathcal I}(Inv^{\leq\omega}_{\mathcal I}\,C)=Pol^\omega_{\mathcal I}(Inv^{\omega}_{\mathcal I}\,C)$.
\end{proposition}

\begin{proof} 
From $Inv^{\omega}_{\mathcal I}\,C\subseteq Inv^{\leq\omega}_{\mathcal I}\,C$ and the contravariance of $Pol^\omega_{\mathcal I}$ it follows that $Pol^\omega_{\mathcal I}(Inv^{\leq\omega}_{\mathcal I}\,C)\subseteq Pol^\omega_{\mathcal I}(Inv^{\omega}_{\mathcal I}\,C)$.

We will prove the converse inequality. According to Lemma \ref{lem:clofin}, the top extension $S^\top$ of a finitary relation $S$ is locally closed. Therefore, Proposition \ref{prop:ii}(iii) implies that $\varphi\in Pol^\omega_{\mathcal I}(S^\top)$ if and only if for every matrix $m\in\mathsf M_\omega^{\omega,S^\top}$, $\varphi[m]\in S^\top$.

Using Lemma \ref{lem:4.9boh}, we can conclude that $\varphi\in Pol^\omega_\mathcal I(S^\top)$ if and only if $\varphi\in Pol^\omega_\mathcal I(S)$. In other words, $S$ and $S^\top$ have the same $\mathcal I$-polymorphisms. Therefore, we deduce that a finitary relation $S \in Inv^{\leq\omega}_{\mathcal I}\,C$ if and only if $S^\top \in Inv^{\omega}_{\mathcal I}\,C$. This easily leads to the desired conclusion.
\end{proof}

%\begin{lemma}\label{lem:4.9} Let $f\in O_A$ be a finitary operation and $S\subseteq A^\alpha$ be a finitary relation. Then the following conditions are equivalent:
% \begin{enumerate}
%\item $f$ is a  polymorphism of $S$ (according to Definition \ref{def:2.2}).
%\item $f^\top$ is a $\mathcal I$-polymorphism of $S$ for every ideal map $\mathcal I$.
%\item $f^\top$ is a $\mathcal I$-polymorphism of $S^\top$ for every ideal map $\mathcal I$.
%\end{enumerate}
%\end{lemma}
%
%\begin{proof}
% Trivial, because $S^\top$ is $\mathcal I_m$-closed for every matrix $m$ (see Lemma \ref{lem:closure}).
%\end{proof}
%
%It follows that $(Pol(S))^\top= Pol^\omega_\mathcal I(S^\top) \cap (O_A)^\top$.
%We define $Pol_X^\top(R^\top) = Pol^\omega_X(R^\top) \cap (O_A)^\top$. By Lemma \ref{lem:4.9}  $Pol_X^\top(R^\top) =(Pol(R))^\top$.
  
%We define now the notion of invariant relation, in few variants, depending on the arity and on the closedness of the relations themselves.

%The subscript $X$ in $Inv_X$ will be omitted when $X$ is such that $\mathrm{Cl}_X^m(R)=R$ for every relation $R$.

\section{A characterisation of the $X$-closed infinitary $\omega$-clones}

In this section, we consider an   ideal $X$ on $A^\omega$ satisfying $L\subseteq X\subseteq G$. %For any set $C$ of $\omega$-operations, we have $\widetilde C\subseteq\mathrm{Cl}_X(C)\subseteq \overline C$.

%Given an $\omega$-operation $\varphi: A^\omega\rightarrow A$ and $m\in \mathsf M_\omega^\omega$,  let $\mathrm{graph}_m(\varphi) = \varphi_\omega(m)$ be the $m$-graph of $\varphi$.

We term {\em row-injective} a matrix whose rows are pairwise distinct.

\begin{definition}\label{def:rmc}
 Let  $X$ be an ideal on $A^\omega$. Given    a set $C$ of $\omega$-operations and $m \in \mathsf M_\alpha^\omega$ ($\alpha\leq\omega$), 
we define the $\alpha$-relation $R_{m,C}=\{ \varphi[m] \ | \ \varphi\in  C\}$ and the following sets of relations:
\begin{itemize}
\item $\mathcal R_C=\{R_{m,C} \ |\ m \in \mathsf M_\alpha^\omega \text{ is row-injective } \text{and}\ \alpha\leq\omega\}$.
%\item $\mathcal J_C=\{R_{m,C} \ |\ m \in \mathsf M_\alpha^\omega\text{ is row-injective}  ,\ \alpha\leq\omega\ \text{and $R_{m,C}$ is $X_m$-closed}\}$.
\item $\mathcal W_{C}=\{R_{m,C} \ |\ m \in \mathsf M_\alpha^\omega\text{ is row-injective} ,\ \alpha\leq\omega\ \text{and $\mathrm{row}(m)\in X$}\}$.\end{itemize}
\end{definition}
 Notice that $R_{m,C}$ is  $\mathcal X_m$-closed if   $\mathrm{row}(m)\in X$, because in this case the $\mathcal X_m$-topology is discrete. The  condition  $\mathrm{row}(m)\in X$ is always satisfied if $|\mathrm{row}(m)|<\omega$.
%Thus, we have $\mathcal W_C\subseteq\mathcal R_C$. %If $X= G$ is the global ideal, then $\mathcal J_C = \mathcal R_C$.

% Observe that if $C$ contains the projections, then $m^j\in R_{m,C}$ for every $j\in\omega$.  

 \begin{definition}
   Let  $X$ be an ideal on $A^\omega$. An ideal map $\mathcal I$ is \emph{$X$-adequate} if $\mathcal I\leq \mathcal X$ and  $\mathcal I(m)=\mathcal P(\alpha)$ for every row-injective $m \in \mathsf M_\alpha^\omega$ such that $\mathrm{row}(m)\in X$.
 \end{definition}

 \begin{theorem}\label{thm:duedue2} Let  $C$ be an infinitary $\omega$-clone on $A$,
   $X$ be an ideal on $A^\omega$ satisfying $L\subseteq X\subseteq G$, and  $\mathcal I$ be an $X$-adequate ideal map.
The following conditions are equivalent, for all $\omega$-operations $\varphi:A^\omega \to A$:
\begin{enumerate}
\item[$(1)$] $\varphi\in \mathrm{Cl}_X(C)$;
\item[$(2)$] $\varphi\in Pol^\omega_{\mathcal I}(\mathcal R_C)$;
\item[$(3)$]  $\varphi[m]\in   \mathrm{Cl}_{\mathcal I_m}(R_{m,C})$, for every row-injective matrix $m$;
%\item[$(4)$]  $\varphi[m]\in  R_{m,C}$, for every row-injective matrix $m$ such that $R_{m,C}$ is $X_m$-closed;
\item[$(4)$]  $\varphi[m]\in   R_{m,C}$, for every row-injective matrix $m$ such that $\mathrm{row}(m)\in X$;
\item[$(5)$] $Inv^{\leq\omega}_{\mathcal I}\,\mathrm{Cl}_X(C)\subseteq   Inv^{\leq\omega}_{\mathcal I}\, \varphi$;
%\item[$(7)$] $\varphi\in Pol^\omega_X(\mathcal J_C)$;
\item[$(6)$] $\varphi\in Pol^\omega_{\mathcal I}(\mathcal W_{C})$.
\end{enumerate}

\end{theorem}
\begin{proof}

  \noindent $(1)\Rightarrow(2)$:  
 Let  $\varphi\in \mathrm{Cl}_X(C)$ and  consider the $\alpha$-relation $R_{m, C}$ for some row-injective matrix $m\in \mathsf M_\alpha^\omega$. We aim to show that for every matrix $r\in \mathsf M_\alpha^{\omega,R_{m, C}}$, we have $\varphi[r]\in \mathrm{Cl}_{\mathcal I_r}(R_{m, C})$.  Let $c\in \mathcal I(r)$. Since $\mathcal I(r)\subseteq \mathcal X(r)$, we have $\bar c=\{r_i: i\in c\}\in X$.
 By the given hypothesis there exists $\varphi' \in C$ such that $\varphi_{|\bar c}=\varphi'_{|\bar c}$. Additionally, there exist
 $\psi_n\in  C$ ($n\in\omega$) such that $r^n=  \psi_n[m]$. 
Let $\Psi=q^A(\varphi',\boldsymbol\psi)\in C$. By definition of $R_{m, C}$ we have $\Psi[m]\in R_{m, C}$. We claim that $\varphi(r_i)=\Psi(m_i)$ for every $i\in c$. We can see that:
$\Psi(m_i)=\varphi'(\psi_0(m_i),\ldots,\psi_n(m_i),\ldots  ) =
\varphi'(r_i^0\ldots,r^n_i,\ldots  ) =\varphi'(r_i)=\varphi(r_i)$.
Thus, we have shown that $\varphi(r_i)=\Psi(m_i)$ for every $i\in c$.
By the arbitrariness of $c$ we get the conclusion $\varphi[r]\in \mathrm{Cl}_{\mathcal I_r}(R_{m, C})$.

  \noindent $(2)\Rightarrow(3)$: Let $m$ be a row-injective matrix. Since $C$ contains all projections, then $m^j\in R_{m, C}$ for every $j$.  It follows that $m\in \mathsf M_\alpha^{\omega,R_{m, C}}$. Therefore, by the given hypothesis $\varphi[m]\in  \mathrm{Cl}_{\mathcal I_m}(R_{m,C})$.

  \noindent $(3)\Rightarrow(4)$: By (3) we have $\varphi[m]\in   \mathrm{Cl}_{\mathcal I_m}(R_{m,C})$.
   Since $\mathcal I(m)=\mathcal P(\alpha)$ whenever $\mathrm{row }(m)\in X$, we get the conclusion $\varphi[m]\in R_{m,C}$.

  \noindent $(4)\Rightarrow(1)$: 
  We show that, for every $c\in X$,  $\varphi_{|c}=\psi_{|c}$ for some $\psi\in C$.   By the hypothesis on $X$ we have that $|c|\leq \omega$. Choose   any row-injective matrix $m$ such that $\mathrm{row}(m)=c\in X$. 
  Since $\varphi[m]\in   R_{m,C}$,  by the definition of $R_{m,C}$, we have
  $\varphi[m]=\psi[m]$ for some $\psi\in C$. This implies that $\varphi_{|c}=\psi_{|c}$. Since $c$ was chosen arbitrarily from $X$,  it follows that  $\varphi\in  \mathrm{Cl}_X(C)$.

      \noindent $(1)\Rightarrow(5)$: by $\{ \varphi\} \subseteq \mathrm{Cl}_X(C)$ we obtain the conclusion.

 \noindent $(5)\Rightarrow(2)$:  By the equivalence of (1) and (2) we have that $\mathrm{Cl}_X(C)= Pol^\omega_{\mathcal I}(\mathcal R_C)$.  Therefore, $\mathcal R_C\subseteq Inv^{\leq\omega}_{\mathcal I}\,\mathrm{Cl}_X(C)\subseteq_{Hp}   Inv^{\leq\omega}_{\mathcal I}\,\varphi$. It follows that  $\varphi\in Pol^\omega_{\mathcal I}(\mathcal R_C)$.

   \noindent $(2)\Rightarrow(6)$: Trivial. 

   \noindent $(6)\Rightarrow(4)$: 
   Let $m\in\mathsf M_\alpha^{\omega}$ be a row-injective matrix with $\mathrm{row}(m)\in X$, so that  $R_{m, C}\in\mathcal W_{C}$. Since $C$ contains all projections, then $m^j\in R_{m, C}$ for every $j$, hence  $m\in \mathsf M_\alpha^{\omega,R_{m, C}}$.
   Moreover,  $R_{m,C}$ is
$\mathcal I_m$-closed since $\mathcal I(m)=\mathcal P(\alpha)$. Therefore 
from $\varphi\in Pol^\omega_{\mathcal I}(\mathcal W_{C})$ it follows that $\varphi[m]\in  R_{m,C}$. 
   \end{proof}

%It should be noted that in general, under the assumptions of the theorem, the set $\mathrm{Cl}_X(C) = Pol^\omega_\mathcal I(\mathcal R_C)$ is not necessarily an $\omega$-clone. However, if $X$ is (infinitely) substitutive, as stated in Proposition \ref{prop:sub}, then $\mathrm{Cl}_X(C)$ becomes an (infinitary) $\omega$-clone. Therefore, under this assumption, $Pol^\omega_\mathcal I(\mathcal R_C)$ is also an (infinitary) $\omega$-clone.

%\begin{corollary}\label{cor:imp22}
%  Let $C$ be an infinitary $\omega$-clone on $A$ and $X$ be an ideal on $A^\omega$. Then we have: $\mathrm{Cl}_X(C)=  Pol^\omega_\mathcal I(Inv^{\omega}_\mathcal I\,  \mathrm{Cl}_X(C))$ for every
%  $X$-adequate ideal map $\mathcal I$.
%\end{corollary}
%
%\begin{proof}
% Let $D= \mathrm{Cl}_X(C)$. Since by   Theorem \ref{thm:duedue2} we have $D= Pol^\omega_\mathcal I(\mathcal R_C)$, then $\mathcal R_C \subseteq Inv^{\leq\omega}_\mathcal I\,  D$. From the contravariance of $Pol^\omega_\mathcal I$, it follows that
%  $ D\subseteq Pol^\omega_\mathcal I(Inv^{\leq\omega}_\mathcal I\,  D) \subseteq  Pol^\omega_\mathcal I(\mathcal R_C) =D$. Moreover, by Proposition \ref{lem:2.5} we have $Pol^\omega_\mathcal I(Inv^{\omega}_\mathcal I\,  D)=Pol^\omega_\mathcal I(Inv^{\leq\omega}_\mathcal I\,   D)$.
%\end{proof}

\begin{corollary}\label{cor:imp2} Let $C$ be an infinitary $\omega$-clone on $A$, $X$ be an ideal on $A^\omega$ and $\mathcal I$ be an $X$-adequate ideal map.
The following conditions are equivalent:
 \begin{enumerate}
\item $ C$ is  $X$-closed.
\item $ C = Pol^\omega_\mathcal I(Inv^{\leq\omega}_\mathcal I\,  C)$.
\item $ C = Pol^\omega_\mathcal I(Inv^{\omega}_\mathcal I\,  C)$.

\end{enumerate}
\end{corollary}

\begin{proof} (1) $\Rightarrow$ (2) Since by   Theorem \ref{thm:duedue2} we have $C= Pol^\omega_\mathcal I(\mathcal R_C)$, then $\mathcal R_C \subseteq Inv^{\leq\omega}_\mathcal I\,  C$. From the contravariance of $Pol^\omega_\mathcal I$, it follows that
  $ C\subseteq Pol^\omega_\mathcal I(Inv^{\leq\omega}_\mathcal I\,  C) \subseteq  Pol^\omega_\mathcal I(\mathcal R_C) =C$. %Moreover, by Proposition \ref{lem:2.5} we have $Pol^\omega_\mathcal I(Inv^{\omega}_\mathcal I\,  C)=Pol^\omega_\mathcal I(Inv^{\leq\omega}_\mathcal I\,   C)$.

 (2) $\Rightarrow$ (1) By Lemma \ref{prop:ii}(i). 

  (2) $\Leftrightarrow$ (3) By Proposition \ref{lem:2.5}.
\end{proof}

%Let $Pol^\omega_{w,X}(R) = \bigcap \{ \bar m^{-1}(R): m\in   \mathsf M_\alpha^{\omega,R}\ \text{and $\mathrm{row}(m)\in X$} \}$. We have that $\varphi\in Pol^\omega_{w,X}(R)$ if for every matrix $m\in   \mathsf M_\alpha^{\omega,R}$ such that $\mathrm{row}(m)\in X$, $\varphi[m]\in R$.

The following theorem provides additional informations on locally closed infinitary $\omega$-clones. Recall from Notation \ref{not:0} the definition of $Pol^\omega$.

\begin{theorem}\label{thm:local}  Let $C$ be an infinitary $\omega$-clone on $A$, $L$ be the local ideal on $A^\omega$ and $\bot$ be the minimal ideal map.
The following conditions are equivalent, for $\mathcal I\in\{\bot,\mathcal L\}$:
   \begin{enumerate}
   \item[(1)] $C$ is  locally closed;
   \item[(2)] $C= Pol_{\mathcal I}^\omega(Inv_{\mathcal I}^{\leq\omega}\,  C)$;
   \item[(3)] $C= Pol_{\mathcal I}^\omega(Inv_{\mathcal I}^{\omega}\,  C)$;
   \item[(4)] $C= Pol^\omega(Inv^{<\omega}\,  C)$;
     \item[(5)] $ C = Pol_{\mathcal I}^\omega(Inv_{c,\mathcal I}^\omega\,  C)  $;
     \item[(6)]  $ C= Pol^\omega(Inv^\omega_c \,C)$.
  \end{enumerate}
The above equivalent conditions imply that $C= Pol_{\mathcal X}^\omega(Inv_{\mathcal X}^\omega\,  C)$ for every ideal $X$ on $A^\omega$ such that $L\subseteq X\subseteq G$.
\end{theorem}

\begin{proof}
  (1), (2) and (3) are equivalent by Corollary \ref{cor:imp2}, since $\bot$ is $L$-adequate.
 (4) $\Rightarrow$ (1) is trivial. 

  \noindent (1) $\Rightarrow$ (4) Observe that  $\mathcal W_{C}=\{R_{m,C} \ |\ m \in \mathsf M_\alpha^\omega\ \text{is row-injective and}\ \alpha<\omega\}$.
Indeed, when $m$ is row-injective and $row(m)\in L$, the number of rows of $m$ must be finite. Since  $\mathcal W_{C}$ is a set of finitary relations, then $\mathcal W_{C} \subseteq Inv^{<\omega}\,  C$ and  $ C\subseteq Pol^\omega(Inv^{<\omega}\,  C) \subseteq  Pol^\omega(\mathcal W_{C}) =_{\text{Thm }\ref{thm:duedue2}}C$.

\noindent (4) $\Rightarrow$ (6) %Let $T=\{\overline R: C\subseteq Pol^\omega(\overline R)   \}$.
Since the top extensions of finitary relations are locally closed, then we have   $(Inv^{<\omega}\,C)^\top \subseteq Inv^\omega_c \,C$. We get the following chain of inequalities:
$ C\subseteq  Pol^\omega( Inv^\omega_c \,C  ) \subseteq Pol^\omega((Inv^{<\omega}\,C)^\top)
=_{\text{Lem.} \ \ref{lem:4.9boh}} Pol^\omega(Inv^{<\omega}\,C) =C$. 

\noindent (6) $\Rightarrow$ (5)    By Proposition \ref{prop:ii}(iii) we have that $Inv^\omega_c \,C\subseteq Inv_{c,\mathcal I}^\omega\,  C$.   %By Lemma \ref{prop:ii}(i).

\noindent (5) $\Rightarrow$ (1) By Proposition \ref{prop:ii} $Pol_{\mathcal I}^\omega(Inv_{c,\mathcal I}^\omega\,  C)$ is $L$-closed, because $\bot\leq\mathcal L$.

\noindent We have concluded the proof that all six conditions of the theorem are equivalent. To establish the last property, we observe that if $C$ is locally closed, then by Lemma \ref{lem:X}(2), $C$ is also $X$-closed.
\end{proof}

%Notice that, if $R$ is a locally closed $\omega$-relation, then $\varphi$ is a $G$-polymorphism of $R$ iff $\varphi$ is a $L$-polymorphism of $R$.

%\begin{corollary} Let  $C$ be an infinitary $\omega$-clone on $A$,
%Then the following conditions are equivalent, for all $\omega$-operations $\varphi:A^\omega \to A$:
%\begin{enumerate}
%\item[$(1)$] $\varphi\in \overline C$;
%  \item[$(2)$] $\varphi\in Pol^\omega_{\mathcal L}(\mathcal R_C)$;
%\item[$(3)$] $\varphi\in Pol^\omega_{\mathcal F}(\mathcal R_C)$;
%  \item[$(4)$] $Inv^{<\omega}_{\mathcal L}\,\overline C\subseteq   Inv^{<\omega}_{\mathcal L}\, \varphi$.
%\end{enumerate}
%\end{corollary}
%\begin{proof}
%By Theorem \ref{thm:duedue2} and $\mathcal F\leq \mathcal L$.
%\end{proof}

\subsection{The finitary case}
A classical result in the theory of clones is that, for every set $F$ of finitary operations, $Pol(Inv(F))=F$ if and only if $F$ is a locally closed clone. By Proposition \ref{prop:4.9tris} the above two conditions are shown equivalent to $Pol^\top(Inv^\top\,F^\top)=F^\top$.

If $R$ is an $\omega$-relation, we define
  $Pol(R)=\{g\in O_A^{(n)} : n\in\omega,\ \forall m\in \mathsf M_\omega^{n,R}\ g[m]\in R  \}$.

\begin{proposition}\label{cor:33} Let $F\subseteq O_A$ be a clone on $A$.
  Then we have:
  \begin{enumerate}
   \item
$F= Pol(Inv^\omega_\mathcal G\, F^\top)$ iff  $F^\top = Pol^\omega_\mathcal G(Inv^\omega_\mathcal G\,   F^\top)$ iff $F^\top$ is globally closed;
  \item If $A$ is countable, then  $F= Pol(Inv_\mathcal G^\omega\, F^\top)$.
\end{enumerate}
\end{proposition}

\begin{proof} (1) Recall that if $F$ is a clone on $A$, then $F^\top$ is an infinitary $\omega$-clone. Therefore, we can apply Theorem \ref{thm:duedue2} and Corollary \ref{cor:imp2}.

  \noindent (2) By (1) and Proposition \ref{prop:ftop}(iii).
\end{proof}

%The following is a well-known result in the theory of clones: a clone $F$ is closed with respect to  the pointwise convergence topology on $O_A$ if and only if  $F = Pol(Inv\,  F)$  (see Section \ref{sec:cp} and \cite[Proposition 6.1.5]{B21}). Notice that $F$ is closed in $O_A$ if and only if $F^\top$ is locally closed. 
%
%\begin{proposition}
%$F= Pol(Inv\, F)$ iff $F^\top$ is locally closed.
%\end{proposition}
%\begin{proof}
%  We apply Theorem \ref{thm:local} and Lemma  \ref{lem:4.9}:  $F^\top$ is locally  closed iff $F^\top = Pol^\omega(Inv^{<\omega} F^\top)$ iff $F = Pol(Inv\, F)$.
% \end{proof}

\subsection{The topological aspects of global and local polymorphisms}\label{sec:pol-inv-inf2}
Let $Pol^\omega_\mathcal G\mathcal R$ be the set  of all $\mathcal G$-polymorphisms of a given set of $\omega$-relations $\mathcal R$.
In this section, we study $Pol^\omega_\mathcal G\mathcal R$ from a topological perspective.
%The main results are the following: 1)  $Pol^\omega\mathcal R$ is a globally closed infinitary $\omega$-clone; 2)  if $\mathcal R$ is a set of closed $\omega$-relations, then $Pol^\omega\mathcal R$ is a locally closed infinitary $\omega$-clone. 

%Let $A$ be a countable set and $\mathsf c$ be a compact trace on $A$. Since $\mathsf c$ is countable, we define $m(\mathsf c)$ to be a fixed $\omega\times\omega$ matrix such that $m(\mathsf c)_1,m(\mathsf c)_2,\ldots,m(\mathsf c)_n,\ldots$ is a (possibly with repetitions) enumeration of  $\mathsf c$. 

%Recall from Section \ref{sec:prel} that, if $R$ is an $\omega$-relation, then $\mathsf M^{R}$ is the set of all matrices $m$ of order $\omega\times\omega$ such that $m^j\in R$, for every $j\geq 0$.

\begin{proposition}\label{prop:fondamentale} Let $R\subseteq A^\omega$ and $C\subseteq O_A^{(\omega)}$. Then we have:
\begin{enumerate}
\item $ Pol^\omega\, \overline R=  Pol^\omega_\bot \,R$.
%\item $Inv^1 (C)= \{R : \overline R\in Inv^\omega(C) \}$.
\item If $R\in Inv^\omega_\bot\, C$, then $\overline R \in Inv_c^\omega\,C$.
\item $Pol^\omega_\mathcal I \, R\subseteq Pol^\omega\, \overline R$, for every ideal map $\mathcal I$.
\end{enumerate}
\end{proposition}

\begin{proof} (1) The inclusion  $ Pol^\omega\, \overline R \subseteq  Pol^\omega_\bot \,R$ is trivial, because
  $\mathsf M^{ \omega, R}_\omega\subseteq \mathsf M^{\omega,\overline R}_\omega$.
  The other inclusion follows from the fact that for all $m\in  \mathsf M_\omega^{\omega,\overline R}$ and for all $d\subseteq_{\mathrm{fin}} \omega$ there exists
  $m'\in  \mathsf M_\omega^{\omega,R}$ such that $m_i=m'_i$ for every $i\in d$. If
  $\varphi\in Pol^\omega_\bot \,R$ then $\varphi[m']\in \overline R$ and
  $\varphi[m]_{|d}=\varphi[m']_{|d}$.
  By the arbitrariness of $d$ $\varphi[m]\in\overline{\overline  R}=\overline R$. Therefore,
  $\varphi\in  Pol^\omega\, \overline R$.

  (2) %Since  $Inv^\omega_\mathcal I \, C \subseteq Inv^\omega_\bot \,  C$ for every $\mathcal I$, then $R\in Inv^\omega_\bot\,  C$.
  The conclusion follows from (1).

  (3) It follows from $Pol^\omega_\mathcal I \, R\subseteq Pol^\omega_\bot  \, R =_{(1)} Pol^\omega\, \overline R$.
% The other inclusion is shown as follows. Let $\varphi$ be such that $\forall  m\in \mathsf M^{R},  \varphi[m]\in \overline{R}$. We have to show that, for every $m\in\mathsf M^{\overline R}$, $\varphi[m] \in \overline R$.
%Let $d\subseteq_{\mathrm{fin}} \omega$. Then, there exists a matrix $r\in\mathsf M^R$ such that 
%$r^j_i= m^j_i$ for every $j\in \omega$ and $i\in d$; thus, $r_i = m_i$  for every $i\in d$.
%By the hypothesis on $\varphi$ we have that
%$\varphi[r]\in \overline R$ and $\varphi(r_i)= \varphi(m_i)$ for every $i\in d$. In other words, $\varphi[m]_i=\varphi[r]_i$ for every $i\in d$.
%From $\varphi[r]\in \overline R$ and from the arbitrariness of $d$ it follows that $\varphi[m]\in \overline{\overline R}=\overline R$.
\end{proof}

\begin{proposition}\label{prop:closedpol} Let $R\subseteq A^\omega$ be an $\omega$-relation. Then we have:
  \begin{itemize}
     \item[(i)]      $Pol^\omega\, \overline R$ is locally (hence globally, uniform, and trace) closed.
   \item[(ii)] $\mathrm{Cl}_L(Pol^\omega_\mathcal G\, R)\subseteq Pol^\omega\,\overline R$.
    \end{itemize}

\end{proposition}

\begin{proof}

  (i) By Proposition \ref{prop:fondamentale}(1)   we have $Pol^\omega \, \overline R= Pol^\omega_\bot \, R$. The conclusion follows from Proposition \ref{prop:ii}(i) because
  $\bot\leq \mathcal L$, where $\mathcal L$ is the ideal map associated to the local ideal $L$.

(ii) It follows from (i) and Proposition \ref{prop:fondamentale}(3).
%Let $\varphi \in  \alpha_F[Pol2_\mathsf aR]$ for every $F\subseteq_{\mathrm{fin}}\mathsf a$. We have to show that $\varphi\in Pol2_\mathsf a R$, i.e., for every $s\in \mathcal M_\mathsf a^R$,  $\langle \varphi(s_i): i\in \omega\rangle\in R$. Let $k >0$ and $F=\{s_1,\dots,s_k\}$. Then there exists $\psi_k\in Pol2_\mathsf a(R)$ such that $\varphi(s_j)=\psi_k(s_j)$ for every $j=1,\dots,k$. Then, for every $k>0$, the initial segment $(\varphi(s_1),\dots,\varphi(s_k))$ coincides with the corresponding initial segment of  $(\psi_k(s_i): i\in \omega)\in R$. Since $R$ is closed, we get that $(\varphi(s_i): i\in \omega)\in \overline R = R$ and $\varphi\in Pol2_\mathsf aR$.
%Let $\varphi$ be an element of the local closure of $Pol^\omega R$ and let $s\in\mathcal M^{R}$. We have to show that $\varphi^\omega(s)\in \overline R$.
% We fix $k\in\omega$. 
% Since $\varphi$ is an element of the local closure of $Pol^\omega R$, there exists $\psi_k\in Pol^\omega R$ such that
% $\psi_k^\omega(s)\in R$ and    $\varphi(s_i)=\psi_k(s_i)$ for all $1\leq i\leq k$. By the arbitrariness of $k$, $\varphi^\omega(s)\in \overline R$.
%(iii) Let $\varphi \in  \widehat{Pol^3 R}$. We have to show that $\varphi\in Pol^3  R$, i.e., for every $s\in \mathsf m^{R}_\mathsf c$,  $ \varphi^\omega(s)\in  R$. By hypothesis there exists $\psi\in Pol^3 R$ such that $\varphi_{|\mathsf c}=\psi_{|\mathsf c}$. Then $\varphi^\omega(s)= \psi^\omega(s)\in R$, because $\psi \in Pol^3 R$.
\end{proof}

%\begin{lemma}\label{lem:invpol}
%  Let $C$ be a set of $\omega$-operations on $A$.  Then  $Inv_c^\omega C= Inv_c^\omega \overline C$ and  $Inv^\omega C= Inv^\omega \widetilde C$.
%\end{lemma}
%\begin{proof} 
%  If $R$ is a closed relation, then by Proposition \ref{prop:closedpol}(ii)  $C\subseteq Pol^\omega R$ iff $\overline C\subseteq Pol^\omega R$, because $Pol^\omega R$ is locally closed.
%  Similarly,  Proposition \ref{prop:closedpol}(i)  allows us to prove the other statement.
% \end{proof}
 
 %  We identify $2^\omega$ and the powerset $\mathcal P(\omega)$ of $\omega$ through the bijection $f: 2^\omega\to \mathcal P(\omega)$ such that $f(s)=\{ i : s_i=1\}$. We have $s\in R$ iff there exists $n\geq 0$ such that $f(s)=\omega\setminus \{0,\dots,n-1\}$. For every $\varphi: 2^\omega \to 2$, we denote by $\varphi^\star: \mathcal P(\omega)\to 2$ the map $\varphi\circ f^{-1}$.

 In the following we assume $0< 1$ in set $2$ and we  extend  this partial order pointwise to $2^\omega$. We have that $0^\omega$ is the bottom element, and $1^\omega$ is the top element of $2^\omega$. Additionally, if $s\in 2^\omega$, then we define $\# s=\{ i\in\omega : s_i=1\}$.

Let us introduce an $\omega$-relation $R\subseteq 2^\omega$ that will be used in the proof of Proposition \ref{prop:example} below: 
  \begin{equation}\label{eqq}
    R=\{0^n1^\omega: n\geq 0\}.
  \end{equation}
 The elements of $R$ form a decreasing chain with top element $1^\omega$.
 The following two lemmas, with their proofs left to the reader as they are straightforward,
 %whose easy proofs are left to the reader, 
 characterise the sets $\mathsf M^{\omega,R}_\omega$ and 
  $Pol^\omega_\mathcal G\,  R$.
 
 %We leave to the reader the easy proofs of the following two lemmas.

\begin{lemma} Let $m\in \mathsf M^\omega_\omega$ be a matrix on $2$. Then
  $m\in\mathsf M^{\omega,R}_\omega$ if and only if  $m_i \leq  m_{i+1}$ for all $i$ and $\bigvee_{k\geq 0} m_k=1^\omega$.
\end{lemma}

%We assume $0< 1$ in set $2$.

\begin{lemma}\label{cl:2} Let $\varphi: 2^\omega \to 2$. Then
  $\varphi\in Pol^\omega_\mathcal G\,  R$ if and only if it satisfies the following conditions:
\begin{itemize}
\item[(i)] $\varphi$ is monotone.
%\item[(ii)] $\varphi(\omega)=0$.
\item[(ii)] For every directed family $X$ of elements of $2^\omega$ such that $\bigvee X=1^\omega$, there exists $s\in X$ such that $\varphi(s)=1$.

%\item[(iii)] If $\varphi(S)=1$ and $S_1=S  \subsetneq S_2 \subsetneq\dots \subsetneq S_n\subsetneq\dots$ is an infinite sequence of subsets of $\omega$, then there exists $k\in\omega$ such that $\varphi(S_n)=0$ for every $n\geq k$.
\end{itemize}
\end{lemma}

The monotonicity of $\varphi$ is equivalent to the following condition: for every $s\in 2^\omega$:  $\varphi(s)=1\Rightarrow \forall u\geq s.\ \varphi(u)=1$.
Moreover, if $\varphi$ satisfies item (ii) of Lemma \ref{cl:2}, then $\varphi(1^\omega)=1$.

\begin{proposition}\label{prop:example}
 \begin{enumerate}
   \item There exists  an $\omega$-relation $R$ such that $\mathrm{Cl}_L(Pol^\omega_\mathcal G\, R)\subsetneq Pol^\omega\,  \overline R$.  
   \item  There exists an $\omega$-relation $R$ such that $Pol^\omega_\mathcal G\, R$ is trace closed but not locally closed.
     \item  There exists an $\omega$-relation $R$ such that $Pol^\omega _\mathcal G\, R$ is not trace closed.
\end{enumerate}
\end{proposition}

\begin{proof} (1)  Let $R$ be the $\omega$-relation defined in (\ref{eqq}).
  $R$ is not locally closed, because $0^\omega\in \overline R\setminus R$.
  The constant function $c_0: 2^\omega\to 2$ of value $0$ belongs to $Pol^\omega\, \overline R$.
  We now show that  $c_0\not\in \mathrm{Cl}_L(Pol^\omega _\mathcal G\, R)$, by showing that $\varphi(1^\omega)=1$ for every $\varphi\in Pol^\omega _\mathcal G\, R$. Indeed, let $m\in {\mathsf M}^{\omega,R}_\omega$ be the constantly  $1$ matrix and $\varphi\in Pol^\omega _\mathcal G\, R$.   
  Then   $\varphi[m]=(\varphi(1^\omega))^\omega \in R$. Since $0^\omega\notin R$, then $\varphi[m] =1^\omega$; hence $\varphi(1^\omega)=1$. 
 %  It follows that $c_0(1^\omega) \neq \varphi(1^\omega)$ for all $\varphi\in Pol^\omega R$, so that $c_0\not\in \overline{Pol^\omega R}$.
%Let $\mathsf a$ be a trace on $2$ and $r\in \mathsf a$. Without loss of generality, we assume that $r_i = 0$ for an infinite number of $i\in\omega$.   
% Let  $R\subseteq 2^\omega$ be the $\omega$-relation defined as follows: 
%$$\text{$s\in R$ iff $\exists n\geq 0$ such that
%  $s_i=\begin{cases}1&\text{if $1\leq i\leq n$}\\ r_i&\text{if $i>n$}\end{cases}$.}$$
%$R$ is not closed, because $1^\omega\in \overline R\setminus R$.
% The constant function $c_1$ of value $1$ belongs to $Pol^1_\mathsf a R\setminus Pol^\omega_\mathsf a R$.
%  We now show that  $c_1\not\in \overline{Pol^\omega_\mathsf a R}$. Let $s\in {\mathcal M}^R$ be the constantly  $0$ matrix.   
%  Since   $\varphi^\omega(s)\in R$ for all $\varphi\in Pol^\omega R$, then $\varphi(0^\omega)$ must be equal to $0$. 
%   It follows that $(c_1)_{|\{0^\omega\}} \neq \varphi_{|\{0^\omega\}}$ for all $\varphi\in Pol^\omega R$, so that $c_1\not\in \overline{Pol^\omega R}$.

(2) Let $R$ again be the $\omega$-relation defined in (\ref{eqq}). 
%The $\omega$-relation $R$ corresponds to family of cofinite subsets of $\omega$.
%We define the following partial order on $A^\omega$: $s\leq r$ iff $\{ i : s_i=0\} \subseteq \{ i: r_i=0\}$. The bottom element is the sequence $1^\omega$ and the top element is the sequence $0^\omega$.
We first prove that  $Pol^\omega_\mathcal G\,  R$  is not locally closed.
Let $\psi: 2^\omega \to 2$ be the monotone function defined as follows, for every $r\in 2^\omega$:
\begin{equation}\label{eqq2}\psi(r) = \begin{cases} 0&\text{if $|\#r|< \omega$}\\
1&\text{otherwise.}\end{cases}\end{equation}
Let $X$ be the directed family of all $s\in 2^\omega$ such that $|\# s|<\omega$.  We have that $\bigvee X=1^\omega$, but $\psi(s)=0$ for every $s\in X$. From Lemma \ref{cl:2} it follows that $\psi\notin Pol^\omega_\mathcal G\, R$. 

 Let us show that $\psi \in \mathrm{Cl}_L(Pol^\omega_\mathcal G\, R)$: let $k\in\omega$ and $d\in \mathsf M^\omega_k$ be an arbitrary matrix. We need to prove the existence of $\varphi\in Pol_\mathcal G^\omega R$ such that $\psi[d]=\varphi[d]$. We define the $\omega$-operation $\varphi$ as follows, for every $c\in  2^\omega$:
$\varphi(c)=0$ iff $\exists 0\leq i\leq k-1$ such that $|\#d_i|< \omega$ and $c\leq d_i$. It is obvious that $\varphi$ is monotone and $\psi[d]=\varphi[d]$. It remains to show that $\varphi\in Pol_\mathcal G^\omega R$.
 Let $X$ be an arbitrary directed family of elements of $2^\omega$ such that $\bigvee X=1^\omega$. Since the set $\bigcup \{\#d_i :0\leq i\leq k-1,  |\#d_i|< \omega\} =\bigcup \{\#s: \varphi(s)=0\}$ is finite,  there exists $x\in X$ such that $\varphi(x)=1$. Then by Lemma \ref{cl:2},
 $\varphi\in Pol_\mathcal G^\omega R$ and we get the conclusion.

We prove  that $Pol^\omega_\mathcal G\,  R$ is trace closed. Let $\varphi: 2^\omega\to 2$ be an element of the trace closure of $Pol^\omega_\mathcal G\,  R$.
 We will show that $\varphi$ satisfies the conditions of Lemma \ref{cl:2}, which implies that $\varphi\in Pol^\omega_\mathcal G\,  R$.
 Regarding monotonicity, suppose $s\leq u$ and $\varphi(s)=1$. Since  $\varphi$ is simulated by a monotone function on the compact trace generated by $s$ and $u$, it follows that $\varphi(u)=1$. 
  Concerning  condition (ii) of Lemma \ref{cl:2}, consider the basic trace
 $X=[0^\omega]_\equiv$. This  trace is directed such that $\bigvee X = 1^\omega$.
 %containing $s\in 2^\omega$ iff either $|\#s|<\omega$ or $|\omega\setminus \#s|<\omega$.
 Since $\varphi$  is simulated on $X$  by a function in  $Pol^\omega_\mathcal G\,  R$, we have that $\varphi(s)=1$ for some $s$ such that $|\# s|<\omega$. Thus, condition (ii) of Lemma \ref{cl:2} is satisfied.
 Therefore, $\varphi$ satisfies the conditions of Lemma \ref{cl:2}, and hence $\varphi \in Pol^\omega_\mathcal{G}\, R$. This concludes the proof that $Pol^\omega_\mathcal{G}\, R$ is trace closed.

 (3) 
 Consider the set $A=\omega$ and the relation $Z=\{s\in A^\omega : | \text{set}(s)| < \omega\}$. We observe that all $\omega$-operations whose range is finite belong to $Pol^\omega_\mathcal{G}(Z)$.
Now, let us define an $\omega$-operation $\varphi: A^\omega \to A$ as follows, for every $s\in A^\omega$:
 $$\varphi(s)=\begin{cases} i&\text{if  $s\in [i^\omega]_\equiv$ for some $i$}\\ 0 &\mbox{ otherwise.} 
 \end{cases}$$
 The $\omega$-operation $\varphi$ is semiconstant and has an infinite range. 
 We will prove that $\varphi$ belongs to the trace closure of $Pol^\omega_\mathcal{G}(Z)$ but not to $Pol^\omega_\mathcal{G}(Z)$ itself.

 To show that $\varphi$ belongs to the trace closure of $Pol^\omega_\mathcal{G}(Z)$, consider a compact trace $\mathsf{d} = [s_0]_\equiv \cup \ldots \cup [s_k]_\equiv$. We define a semiconstant function $\psi$ as follows: for $s \in \mathsf{d}$, $\psi(s) = \varphi(s)$, and for $s$ not in $\mathsf{d}$, $\psi(s) = 0$. Note that $\psi$ has a finite range, and on the compact trace $\mathsf{d}$, $\psi$ coincides with $\varphi$. Therefore, $\varphi$ belongs to the trace closure of $Pol^\omega_\mathcal{G}(Z)$.
To demonstrate that $\varphi \notin Pol^\omega_\mathcal{G}(Z)$, we consider the matrix $m \in \mathsf{M}^{\omega,Z}_\omega$ defined as $m_i = 012 \ldots (i-1)i^\omega$. For this matrix, $\varphi[m] \notin Z$, since $\varphi(m_i)=i$, and its range is infinite.
Hence, we have shown that $\varphi$ belongs to the trace closure of $Pol^\omega_\mathcal{G}(Z)$ but not to $Pol^\omega_\mathcal{G}(Z)$ itself.
\end{proof}

\subsection{$X$-closed $\omega$-clones}\label{sec:matrical}

Infinitary $\omega$-clones have been mainly studied with respect to both local topology and global topology. However, to extend the previous results to $\omega$-clones that are not necessarily infinitary, we require a new concept of polymorphism.

In this section, we aim to introduce and explore this new notion of polymorphism that will allow us to generalise Theorem \ref{thm:duedue2} to a wider class of $\omega$-clones, and to 
characterise trace (resp. uniform) closed sets  of $\omega$-operations.

\begin{definition}\label{def:wpol} Let  $R\subseteq A^\alpha$, $m\in \mathsf M_\alpha^{\omega,R}$ be a matrix and   $\mathcal I$ be an ideal map. We say that an $\omega$-operation $\varphi$ is a \emph{matrical $\mathcal I$-polymorphism of $(R,m)$} if 
   for every $n$ and  $\mathbf r= r^0,\dots,r^{n-1}\in R$, $\varphi[m[\mathbf r]]\in \mathrm{Cl}_{\mathcal I_{m[\mathbf r]}}(R)$.
\end{definition}

  We denote by $M\!Pol^\omega_{\mathcal I}(R,m)$ the set of all  matrical $\mathcal I$-polymorphisms of $(R,m)$.
If $\mathcal R$ is a set of pairs $(R,m)$ of this kind, we define  $M\!Pol^\omega_{\mathcal I}(\mathcal R) = \bigcap_{(R,m)\in\mathcal R} M\!Pol^\omega_{\mathcal I}(R,m)$.

Let  $C$ be a set of $\omega$-operations. We define $M\!Inv^{\leq\omega}_{\mathcal I}\,C$ (resp. $M\!Inv^{\omega}_{\mathcal I}\,C$) to be the set of all pairs $(R,m)$ such that $R\subseteq A^\alpha$ ($\alpha\leq\omega$), $m\in 
\mathsf M_\alpha^{\omega,R}$ (resp. $R\subseteq A^\omega$, $m\in \mathsf M_\omega^{\omega,R}$) and
each $\varphi\in C$ is a matrical $\mathcal I$-polymorphism of $(R,m)$.

Let $\mathcal M_C=\{(R_{m,C},m) : m\in\mathsf M\}$, where $R_{m,C}$ is defined in Definition \ref{def:rmc}.

 \begin{theorem}\label{thm:duedue5} Let  $C$ be an $\omega$-clone on $A$,
   $X$ be an ideal on $A^\omega$ satisfying $L\subseteq X\subseteq G$, and
    $\mathcal I$ be a $X$-adequate ideal map.
Then the following conditions are equivalent, for all $\omega$-operations $\varphi:A^\omega \to A$:
\begin{enumerate}
\item[$(1)$] $\varphi\in \mathrm{Cl}_X(C)$;
\item[$(2)$] $\varphi\in M\!Pol^\omega_{\mathcal I}(\mathcal M_C)$;
\item[$(3)$]  $\varphi[m]\in   \mathrm{Cl}_{\mathcal I_m}(R_{m,C})$, for every row-injective matrix $m$;
\item[$(4)$]  $\varphi[m]\in   R_{m,C}$, for every row-injective matrix $m$ such that $\mathrm{row}(m)\in X$;
\item[$(5)$] $M\!Inv^{\leq\omega}_{\mathcal I}\,\mathrm{Cl}_X(C)\subseteq   M\!Inv^{\leq\omega}_{\mathcal I}\, \varphi$.
\end{enumerate}

\end{theorem}
\begin{proof}

  \noindent $(1)\Rightarrow(2)$:  
 Let  $\varphi\in \mathrm{Cl}_X(C)$ and  consider the pair $(R_{m, C},m)$ for some row-injective matrix $m\in \mathsf M_\alpha^\omega$. We aim to show that, for every  $\mathbf r=r^0,\dots,r^{n-1}\in R_{m, C}$,  $\varphi[m[\mathbf r]]\in \mathrm{Cl}_{\mathcal I_{m[\mathbf r]}}(R_{m, C})$.  
 
 Let $c\in \mathcal I(m[\mathbf r]) \subseteq \mathcal X(m[\mathbf r])$ and $\bar c=\{m[\mathbf r]_i: i\in c\}\in X$.
 By the given hypothesis there exists $\varphi' \in C$ such that $\varphi_{|\bar c}=\varphi'_{|\bar c}$. Additionally, there exist
 $\psi_i\in  C$ ($0\leq i\leq n-1$) such that $r^i=  \psi_i[m]$. 
Let $\Psi_n=q_n^A(\varphi',\psi_0,\dots,\psi_{n-1})\in C$. By definition of $R_{m, C}$ we have $\Psi_n[m]\in R_{m, C}$. We claim that $\varphi(m[\mathbf r]_i)=\Psi_n(m_i)$ for every $i\in c$. We have that:
$\Psi_n(m_i)=\varphi'(\psi_0(m_i),\ldots,\psi_{n-1}(m_i),m_i^{n},\ldots  ) =
\varphi'(r_i^0\ldots,r_i^{n-1},m_i^{n}\ldots  ) =\varphi'(m[\mathbf r]_i)=\varphi(m[\mathbf r]_i)$.
Thus, we have shown that $\varphi(m[\mathbf r]_i)=\Psi_n(m_i)$ for every $i\in c$.
By the arbitrariness of $c$ we get the conclusion $\varphi[m[\mathbf r]]\in \mathrm{Cl}_{\mathcal I_{m[\mathbf r]}}(R_{m, C})$.

  \noindent $(2)\Rightarrow(3)$: Since $m\in \mathsf M_\alpha^{\omega,R_{m, C}}$ and $\varphi$ is a matrical  $\mathcal I$-polymorphism of $(R_{m,C},m)$, by choosing $n=0$, we get $\varphi[m]\in  \mathrm{Cl}_{\mathcal I_m}(R_{m,C})$.

  \noindent $(3)\Rightarrow(4)$: 
   Since $\mathcal I(m)=\mathcal P(\alpha)$ whenever $\mathrm{row }(m)\in X$, by $\varphi[m]\in   \mathrm{Cl}_{\mathcal I_m}(R_{m,C})$ we get the conclusion $\varphi[m]\in R_{m,C}$.

  \noindent $(4)\Rightarrow(1)$ and   $(1)\Rightarrow(5)$: 
 As in Theorem \ref{thm:duedue2}.

 \noindent $(5)\Rightarrow(2)$:  By the equivalence of (1) and (2) we have that $\mathrm{Cl}_X(C)= M\!Pol^\omega_{\mathcal I}(\mathcal M_C)$.  Then $\mathcal M_C\subseteq M\!Inv^{\leq\omega}_{\mathcal I}\,\mathrm{Cl}_X(C)\subseteq_{Hp}   M\!Inv^{\leq\omega}_{\mathcal I}\,\varphi$. It follows that  $\varphi\in M\!Pol^\omega_{\mathcal I}(\mathcal M_C)$.
\end{proof}

\begin{corollary}\label{cor:mah!} Let $C$ be an $\omega$-clone, $X$ be an ideal on $A^\omega$ and $\mathcal I$ be an $X$-adequate ideal map.
Then the following conditions are equivalent: 
  \begin{enumerate}
\item $ C$ is  $X$-closed.
\item $ C = M\!Pol^\omega_\mathcal I(M\!Inv^{\leq\omega}_\mathcal I\,  C)$.
\item $ C = M\!Pol^\omega_\mathcal I(M\!Inv^{\omega}_\mathcal I\,  C)$.
\end{enumerate}
\end{corollary}

\begin{proof} The proof is similar to that of Corollary \ref{cor:imp2} and it is omitted.
%(1) $\Rightarrow$ (2)  Let $D= \mathrm{Cl}_X(C)$. Since by   Theorem \ref{thm:duedue5} we have $D= Mol^\omega_\mathcal I(\mathcal M_C)$, then $\mathcal M_C \subseteq Mnv^{\leq\omega}_\mathcal I\,  D$. From the contravariance of $Mol^\omega_\mathcal I$, it follows that
%  $ D\subseteq Mol^\omega_\mathcal I(Mnv^{\leq\omega}_\mathcal I\,  D) \subseteq  Mol^\omega_\mathcal I(\mathcal M_C) =D$. 
%
% (2) $\Rightarrow$ (1) By Lemma \ref{prop:ii}(i). 
%
%  (2) $\Leftrightarrow$ (3) By Proposition \ref{lem:2.5}.
\end{proof}

\begin{remark} Let $T$ be the trace ideal, $U$ be the uniform ideal on $A^\omega$, and $X\in\{T,U\}$. In the hypothesis of Definition \ref{def:wpol}, we have that $\mathcal X(m[\mathbf r])= \mathcal X(m)$ for every $\mathbf r$. Then 
an $\omega$-operation $\varphi$ is a matrical $\mathcal X$-polymorphism of $(R,m)$ if 
   for every $n$ and  $\mathbf r= r^0,\dots,r^{n-1}\in R$, $\varphi[m[\mathbf r]]\in \mathrm{Cl}_{\mathcal X_m}(R)$.
\end{remark}

\begin{corollary} Let $C$ be an $\omega$-clone. Then $C$ is trace (resp. uniform) closed iff $C= M\!Pol^\omega_\mathcal T(M\!Inv^{\omega}_\mathcal T\,  C)$ (resp. $C= M\!Pol^\omega_\mathcal U(M\!Inv^{\omega}_\mathcal U\,  C)$). 
\end{corollary}

%\subsection{Mnv$^\omega$-Mol$^\omega$}
%
%In this section we assume $X$ to be an ideal on $A^\omega$ such that $\mathcal X_{m[\mathbf r]}= \mathcal X_{m}$ for every matrix $m\in \mathsf M^\omega_\omega$ and sequence $\mathbf r=r^0\dots r^{n-1}$ of threads. For example, $X$ may be the trace ideal $T$ or the uniform ideal $U$.
%
%\begin{itemize}
%\item $Mol^\omega_\mathcal X(\Delta_A,m) = O_A^{(\omega)}$ for every matrix $m\in \mathsf M_\omega^{\omega,\Delta_A}$;
%\item Let $\sigma:\omega\to\omega$ be a permutation. We have $\varphi\in Mol^\omega_\mathcal X(R,\sigma^{-1}(m))$ if and only if  $\varphi\in Mol^\omega_\mathcal X(\sigma(R),m)$.\\
%Proof: Since $\sigma^{-1}(m)\in\mathsf M_\omega^{\omega,R}$, then $\varphi[\sigma^{-1}(m)[\mathbf r]]\in Cl_{\mathcal X_{\sigma^{-1}(m)}(R)}$. We conclude $\varphi[\sigma(m)]= \sigma(\varphi[m])\in R$; hence  $\varphi[m]\in R(\sigma)$.
%\end{itemize}
%

\section{Inv$^\omega$-Pol$^\omega$ and the $\omega$-relation clones}\label{sec:eight}

In this section we study the opposite side $Inv^\omega$-$Pol^\omega$ of the Galois connection. We define the notion of $\omega$-relation clone and show that the operators acting in such $\omega$-relation
clones preserve the local closedness of the $\omega$-relations. To compare relation clones and $\omega$-relation clones we introduce the notion of cut-closedness that will be used in Section \ref{sec:nine}.
In particular, we show that the following relation clones are cut-closed:
(a) $\mathcal R_{\mathrm{fin}}=\{S\in Rel_A:S^\top\in \mathcal R\}$, for every $\omega$-relation clone $\mathcal R$ on set $A$; (b) 
Relation clones on a finite set.

\subsection{Locally closed $\omega$-relations}
Let $R\subseteq A^\omega$ be an $\omega$-relation and $s\in A^\omega$.
For every function $\sigma:\omega\to\omega$ and $\Gamma\subseteq \omega$, we define:
\begin{itemize}
\item $\sigma(s)= (s_{\sigma 0},s_{\sigma 1},\dots)$;
\item $R(\sigma)=\{r\in A^\omega\ |\ \sigma(r)\in R\}$;
\item $\exists_\Gamma R=\{u\in A^\omega\ |\ \exists r\in R:\ u_i= r_i\ \text{for every $i\in \omega\setminus\Gamma$}\}$.
\end{itemize}

\begin{lemma}\label{lem:10.4}
  Let $R\subseteq A^\omega$ be an $\omega$-relation and  $S\subseteq A^n$ be a finitary relation. Then, we have:
  \begin{itemize}
\item[(i)]  The $\omega$-relation $\exists_\Gamma R$ is locally closed, for every cofinite subset $\Gamma$ of $\omega$;
\item[(ii)] $\exists_b( S^\top)$ is a locally closed $\omega$-relation, for every finite subset $b$ of $\omega$.
\item[(iii)] If $R$ is locally closed, then $R(\sigma)$ is a locally closed $\omega$-relation for every permutation $\sigma$ of $\omega$.
\end{itemize}
\end{lemma}

\begin{proof} 
(i) Let $s\in \overline{\exists_\Gamma R}$. Since $\Gamma$ is a cofinite set and $s\in \overline{\exists_\Gamma R}$, there exists $r\in \exists_\Gamma R$ such that $s_{|\omega\setminus \Gamma}=r_{|\omega\setminus \Gamma}$.
Since $r\in \exists_\Gamma R$, there exists $u\in R$ such that $r_{|\omega\setminus \Gamma}=u_{|\omega\setminus \Gamma}$. Therefore, we have $s_{|\omega\setminus \Gamma}=u_{|\omega\setminus \Gamma}$.
This implies that $s\in \exists_\Gamma R$. Thus, we have shown that $\overline{\exists_\Gamma R}= \exists_\Gamma R$.

(ii) Since $\exists_b (S^\top) =\exists_{b\cup (\omega\setminus n)} (S^\top)$, the conclusion follows from (i).

(iii) If $s\in \overline{R(\sigma)}$, then for every $c\subseteq_{\mathrm{fin}}\omega$, there exists $r\in R(\sigma)$ (i.e., $\sigma(r)\in R$) such that $s_i=r_i$ for every $i\in c$. This implies that $\sigma(s)_{j}= \sigma(r)_j$ for every $j\in \sigma^{-1} (c)$. Since $\sigma$ is a permutation, the arbitrariness of $c$ implies the arbitrariness of $\sigma^{-1}(c)$. As $R$ is locally closed, we can conclude that $\sigma(s)\in R$, and therefore $s\in R(\sigma)$. Thus, $R(\sigma)$ is locally closed.
\end{proof}

\begin{lemma}
   Let $A$ be a finite set and  $R\subseteq A^\omega$ be a locally closed $\omega$-relation. Then $\exists_b R$  is locally closed for every $b\subseteq_{\mathrm{fin}}\omega$.
\end{lemma}

\begin{proof} 
Since $A$ is a finite set, the space $A^\omega$ with the local topology is compact. 

Let $b\subseteq_{\mathrm{fin}}\omega$ and $co(b)=\{\Gamma : b\subseteq\Gamma\subseteq_\mathrm{cofin}\omega\}$.
According to Lemma \ref{lem:10.4}(i), the $\omega$-relation $X= \bigcap_{\Gamma \in co(b)}\exists_\Gamma R$ is locally closed. If we  demonstrate that $\exists_b R=X$, we obtain the desired conclusion.

%($\exists_b R\subseteq X$): If $s\in \exists_b R$ then there exists $r\in R$ such that $s_{|\omega\setminus b}= r_{|\omega\setminus b}$. If $b\subseteq\Gamma\subseteq_\mathrm{cofin}\omega$, then $\omega\setminus \Gamma\subseteq\omega\setminus b$. Then $s_{|\omega\setminus \Gamma}= r_{|\omega\setminus \Gamma}$ and $s\in \exists_\Gamma R$, for every $\Gamma \in co(b)$.

(1) To show $\exists_b R\subseteq X$, consider an element $s\in \exists_b R$. This means that there exists $r\in R$ such that $s_{|\omega\setminus b}= r_{|\omega\setminus b}$. Now, suppose $b\subseteq\Gamma\subseteq_\mathrm{cofin}\omega$. Since $\omega\setminus \Gamma\subseteq\omega\setminus b$, it follows that $s_{|\omega\setminus \Gamma}= r_{|\omega\setminus \Gamma}$. Therefore, $s\in \exists_\Gamma R$ for every $\Gamma \in co(b)$, which implies $s\in X$. Thus, we have shown $\exists_b R\subseteq X$.

(2) To show $X \subseteq \exists_b R$, let $s\in X$. We have:
\begin{equation}\label{eq:uno}
\text{For every $\Gamma\in co(b)$, there exists $r_\Gamma\in R$ such that $s_{|\omega\setminus \Gamma}= (r_\Gamma)_{|\omega\setminus \Gamma}$}
\end{equation}
Since the space $A^\omega$ is compact, the net $[r_\Gamma]_{\Gamma \in co(b)}$, ordered by $\supseteq$, has a cluster point $r$.

We have $r\in R$ since $R$ is locally closed and $r_\Gamma \in R$ for every $r_\Gamma$ in the net. It remains to show that $s_{\omega\setminus b}=r_{\omega\setminus b}$, which will imply that $s\in \exists_b R$.

Since $r$ is a cluster point of the net $[r_\Gamma]$, every open neighborhood of $r$ intersects every tail of the net. Suppose $\Gamma\supseteq b$ is cofinite. Then $O_{r, \omega\setminus\Gamma}=\{u\in A^\omega: r_{|\omega\setminus\Gamma}=u_{|\omega\setminus\Gamma}\}$ is an open neighborhood of $r$. Since the tail $\{r_\Delta: b\subseteq \Delta\subseteq \Gamma\}$ intersects $O_{r, \omega\setminus\Gamma}$, there exists a cofinite set $\Gamma'$ such that $b\subseteq\Gamma'\subseteq \Gamma$ and $r_{|\omega\setminus\Gamma}= (r_{\Gamma'})_{|\omega\setminus\Gamma}$.

By hypothesis (\ref{eq:uno})  we have $s_{|\omega\setminus \Gamma'}= (r_{\Gamma'})_{|\omega\setminus \Gamma'}$, and since $\omega\setminus\Gamma\subseteq \omega\setminus\Gamma'$, we obtain $r_{|\omega\setminus\Gamma} = s_{|\omega\setminus\Gamma}$. By considering the arbitrariness of the cofinite set $\Gamma$ containing $b$, we derive $r_{|\omega\setminus b} = s_{|\omega\setminus b}$. Hence, $s\in \exists_b R$, and we conclude that $X \subseteq \exists_b R$.
\end{proof}

\subsection{$\omega$-Relation clones}
In this section we introduce the $\omega$-relation clones and prove that $Inv^\omega_\mathcal G(C)$ is an $\omega$-relation clone for every set $C$ of $\omega$-operations.

\begin{definition}
Let $\mathcal R$ be a set of $\omega$-relations on $A$.
$\mathcal R$ is called an \emph{$\omega$-relation clone on $A$} if it satisfies the following conditions:
\begin{itemize}
\item[(i)] $\Delta_A\in  \mathcal R$.
\item[(ii)]  If $R\in \mathcal R$, then
  $R(\sigma)\in \mathcal R$ for every permutation $\sigma:\omega\to\omega$.
\item[(iii)] If $R_j\in\mathcal R$ for every $j\in J$, then $\bigcap_{j\in J} R_j\in \mathcal R$.
  \item[(iv)]  If $R\in \mathcal R$, then $\exists_\Gamma R\in \mathcal R$ for every cofinite $\Gamma\subseteq \omega$.
  \end{itemize}
  % A $\omega$-relation clone $\mathcal R$ is called \emph{strong} if it satisfies the following further property: 
%\begin{itemize}
%\item[(v)] For every directed family $\{R_i\}_{i\in I}$ of $\omega$-relations $R_i\in \mathcal R$, we have that $\bigcup_{i\in I}R_i\in \mathcal R$.
%\end{itemize}
\end{definition}

Notice that $\Delta_A$ is a locally closed $\omega$-relation.

$\mathcal R$ is called a \emph{$c\omega$-relation clone} if each $R\in\mathcal R$ is a locally closed $\omega$-relation.

%contains the diagonal relation , and is closed under the following operations: permutations, adding dummy variables, existential projection, and intersections. 
%We define by induction the set $\langle \mathcal R\rangle_{\mathrm{pp}}$ of closed pp-$\omega$-relations over $\mathcal R$ as follows:
%\begin{itemize}
%\item $\mathcal R \subseteq \langle \mathcal R\rangle_{\mathrm{pp}}$.
%\item $\Delta=\{ s\in A^\omega: \exists a\in A\ \forall i\in\omega.\ s_i=a\}\in \langle \mathcal R\rangle_{\mathrm{pp}}$.
%\item 
%  $R^\sigma=\{(s_{\sigma(i)}:i\in \omega)\ |\ s\in R\}\in \mathcal R$ for every permutation $\sigma:\omega\to\omega$ and every $R\in\mathcal R$.
%\item $\bigcap_{i\in I} R_i\in \mathcal R$ for every $R_i\in \mathcal R$ ($i\in I$).
%\item $\exists_i R=\{s\in A^\omega: \exists r\in R\ s_j=r_j\ \text{for all $j\neq i$}\}$
%  for every $R\in \langle \mathcal R\rangle_{\mathrm{pp}}$.
%  \item For every cofinite set $\Gamma\subseteq \omega$,
%  $\exists_\Gamma R=\{s\in A^\omega: \exists r\in R\ s_j=r_j\ \text{for all $j\notin \Gamma$}\}$
%  for every $R\in \langle \mathcal R\rangle_{\mathrm{pp}}$.
%\end{itemize}

The set of $\omega$-relation clones on $A$ is closed under arbitrary intersection. 
If $\mathcal R$ is a set of $\omega$-relations, then $\langle \mathcal R\rangle_{\omega}$ denotes the least $\omega$-relation clone including $\mathcal R$. 
It is easy to show that 
 $Pol^\omega_\mathcal G \mathcal R = Pol^\omega_\mathcal G \langle \mathcal R\rangle_{\omega}$ for every set  $\mathcal R$  of $\omega$-relations.

%\begin{lemma}\label{lem:8.5}
% Let $S_i\subseteq A^n$ ($i\in\omega$) be a chain of $n$-ary relations (i.e., $S_i\subseteq S_{i+1}$). 
% Then $(\bigcup_{i\in\omega} S_i)^\top = \bigcap_{k\in \omega}(\bigcup_{i\in \omega} S_i[k])^\top$.
%\end{lemma}
%
%\begin{proof}
% If $s\in (\bigcup_{i\in\omega} S_i)^\top$, then $s\in S_j^\top$ for some $j$; hence $(s_0,\dots,s_{n-1})\in S_j=S_j[n]$.
% It follows that  $(s_0,\dots,s_{k-1})\in S_j[k]$ for every $k\in\omega$. Then $s\in S_j[k]^\top$ for every $k$ and we get the conclusion $s\in  \bigcap_{k\in \omega}(\bigcup_{i\in \omega} S_i[k])^\top$.
%  The converse is trivial by considering that $S_i[k]^\top$ is closed.
%\end{proof}

Let $m\in \mathsf M_\omega^\omega$ a matrix  and $\sigma:\omega\to\omega$ be a permutation. Then we denote by $\sigma(m)$ the matrix, whose rows are defined as follows: $\sigma(m)_i=m_{\sigma(i)}$.

\begin{lemma}\label{lem:pp} Let $C\subseteq O_A^{(\omega)}$. Then
  $Inv_\mathcal G^\omega\, C$ and $Inv_{c}^\omega\, C$ are  $\omega$-relation clones.
\end{lemma}
\begin{proof}
\begin{itemize}
\item[(i)] It is easy to show that every $\omega$-operation is a $\mathcal G$-polymorphism of $\Delta_A$. 

\item[(ii)] Let  $\sigma$ be a permutation of $\omega$. We prove that $\varphi$ is a $\mathcal G$-polymorphism of an $\omega$-relation $R$ if and only if $\varphi$ is a $\mathcal G$-polymorphism of  $R(\sigma)$. Let $m\in\mathsf M_\omega^{\omega,R(\sigma)}$. Since $\sigma(m)\in \mathsf M_\omega^{\omega,R}$ then $\varphi[\sigma(m)]= \sigma(\varphi[m])\in R$; hence 
 $\varphi[m]\in R(\sigma)$.

\item[(iii)] If $\varphi$ is a $\mathcal G$-polymorphism of $R_j$ for every $j\in J$, then it is also a $\mathcal G$-polymorphism of $\bigcap_{j\in J} R_j$.

\item[(iv)] Let $\Gamma\subseteq_{\mathrm{cofin}} \omega$, $\varphi$ be a $\mathcal G$-polymorphism of $R\subseteq A^\omega$ and $m\in\mathsf M_\omega^{\omega,\exists_\Gamma R}$.  Then there exists a matrix $p\in\mathsf M_\omega^{\omega,R}$ such that $(p^j)_{|\omega\setminus\Gamma}=(m^j)_{|\omega\setminus\Gamma}$ for every $j\in\omega$. For every $i\in \omega\setminus\Gamma$, $p_i=m_i$, so that $\varphi(p_i)=\varphi(m_i)$. We conclude that $\varphi[m]\in \exists_\Gamma R$, because $\varphi[p]\in R$.
\end{itemize}
\noindent Concerning  $Inv_{c}^\omega\, C$, we use   Lemma \ref{lem:10.4} and the same arguments as above. Lemma \ref{lem:10.4} ensures  that the operations (i)-(iv) defining the $\omega$-relation clones preserve the locally closedness of the $\omega$-relations in $Inv_{c}^\omega\, C$.
\end{proof}

%\begin{lemma}\label{lem:pp} Let $C\subseteq O_A^{(\omega)}$. Then
%  $Inv^\omega C$ and $Inv_c^\omega C$ are  $\omega$-relation clones.
%\end{lemma}
%\begin{proof}
%\begin{itemize}
%\item[(a)] Every $\omega$-operation is an $\omega$-polymorphism of $\Delta_A$. 
%
%\item[(b)] Let $R$ be an $\omega$-relation, $\sigma$ be a permutation of $\omega$, $\varphi$ be an $\omega$-polymorphism of $R$ and $m\in\mathsf M^{R(\sigma)}$. Since $\sigma(m)\in \mathsf M^{R}$ then $\varphi[\sigma(m)]= \sigma(\varphi[m])\in R$; hence 
% $\varphi[m]\in R(\sigma)$.
%
%\item[(c)] If $\varphi$ is an $\omega$-polymorphism of $R_i$ for every $i\in I$, then it is also an $\omega$-polymorphism of $\bigcap_{i\in I} R_i$.
%
%\item[(d)] Let $\Gamma\subseteq_{\mathrm{cofin}} \omega$, $\varphi$ be an $\omega$-polymorphism of $R\subseteq A^\omega$ and $m\in\mathsf M^{\exists_\Gamma R}$.  Then there exists a matrix $p\in\mathsf M^R$ such that $(p^j)_{|\omega\setminus\Gamma}=(m^j)_{|\omega\setminus\Gamma}$ for every $j\in\omega$. For every $i\in \omega\setminus\Gamma$, $p_i=m_i$, so that $\varphi(p_i)=\varphi(m_i)$. We conclude that $\varphi[m]\in \exists_\Gamma R$, because $\varphi[p]\in R$.
%\end{itemize}
%\noindent By Lemma \ref{lem:10.4} the operations (i)-(iv) defining the $\omega$-relation clones preserve the closedness of the $\omega$-relations in $Inv_c^\omega C$.
%\end{proof}

In the following example we show that,  in general, the $\omega$-relation clone $Inv^\omega_\mathcal G\, C$ is not closed under union of directed sets of $\omega$-relations.

\begin{example} Let $R_n=\{s\in \omega^\omega : \forall j.\, 0\leq s_j\leq n-1\}$ ($n\geq 1$) be a chain of locally closed $\omega$-relations. The $\omega$-operation $\varphi$ such that $$\varphi(s)=\begin{cases}j &\text{if $\exists k$ s.t. $s_i=j$ for all $i\geq k$}\\s_0&\text{otherwise}
\end{cases}$$
 is a $\mathcal G$-polymorphism of each $R_n$, but it is not a $\mathcal G$-polymorphism of $R=\bigcup R_n$. Indeed, $\varphi[m]=(0,1,2,\dots,n,n+1,\dots)\notin R$, where
 $m\in \mathsf M_\omega^{\omega,R}$ is such that $m^j=m_j= (0,1,\dots,j-1,j,\dots,j,\dots)$.
\end{example}

\subsection{Cut-closed relation clones}

If $R$ is an $\omega$-relation on $A$ and $k\in\omega$, then $R[k]\subseteq A^k$ is defined as follows:
$(a_0,\ldots,a_{k-1})\in R[k]$ if there exists $s\in R$ such that
$s_i=a_i$ for all $0\leq i\leq k-1$. 
Note that $R\subseteq R[k]^\top$ for every $k\in\omega$.

\begin{definition}\label{def:cutc}
A relation clone $\mathcal S$ of finitary relations (see Definition \ref{def:src}) is called \emph{cut-closed} if, for every family of relations $S_i\in \mathcal S$ ($i\in I$), we have that 
$(\bigcap_{i\in I} S_i^\top)[n] \in\mathcal S$ for every $n\in\omega$. 
\end{definition}

In the hypothesis of Definition \ref{def:cutc}, we have that 
$S_i^\top[n]\in\mathcal S$. Indeed, 
if $k$ is the arity of $S_i$, then   %by Definition \ref{def:src}:
$$S_i^\top[n]=\begin{cases}S_i\times A^{n-k}&\text{if $k\leq n$}\\
\pi_{f_n^k}(S_i)&\text{if $k>n$}\end{cases}$$
where, for every $n\leq k$, $f_n^k$ is the canonical embedding of $ n$ into $ k$.
Moreover, $$(\bigcap_{i\in I} S_i^\top)[n] \subseteq \bigcap_{i\in I} (S_i^\top[n]),$$ but the opposite inclusion is false in general.
For example, if $S_1=\{(0,a)\}$ and $S_2=\{(0,b)\}$, then $(S_1^\top \cap S_2^\top)[1]=\emptyset \neq S_1^\top[1] \cap S_2^\top[1]= \{(0)\}$.

%\begin{lemma}\label{lem:cutt} Let $F\subseteq O_A$. Then $Inv\,F$ is cut-closed.
%\end{lemma}
%
%\begin{proof} Let $S_i\in Inv\,F$ ($i\in I$), $R=\bigcap_{i\in I} S_i^\top$, $f:A^r\to A\in F$ and $m$ be a matrix of order $n\times r$ such that $m^j\in R[n]$ for every $0\leq j\leq r-1$. We have to prove that $(f(m_0),\dots, f(m_{n-1}))\in R[n]$.
%Let  $z\in\mathsf M^R$ be a matrix such that $(z^j)_{|\hat n} = m^j$ for every $0\leq j\leq r-1$.
%We prove that $(f^\top)_\omega(z)\in R$. Indeed, if $S_i$ has  arity $k$, then
% $(z^j)_{|\hat k} \in S_i$ for every $j\in\omega$.  Since $S_i\in Inv\,F$, then  $[(f^\top)_\omega(z)]_{|\hat k} = (f((z_0)_{|\hat r}),\dots, f((z_{k-1})_{|\hat r}))\in S_i$ and $(f^\top)_\omega(z)\in S_i^\top$. The conclusion follows because 
% $f(m_j) = f((z_j)_{|\hat r})$ for every $0\leq j\leq n-1$.
%\end{proof}

%If $\mathcal R$ is a set of $\omega$-relations on $A$, then $\mathcal R_{\mathrm{fin}}= \{S\in Rel_A:  S^\top\in \mathcal R\}$.

\vspace{0.3cm}

The definition of $\mathcal R_{\mathrm{fin}}$ is given in Section \ref{sec:rel}.

\begin{proposition}\label{lem:finrel}  If $\mathcal R$ is an $\omega$-relation clone on $A$, then $\mathcal R_{\mathrm{fin}}$ is a cut-closed relation clone. 
%\item[(2)] If $\mathcal S$ is a relation clone, then $\mathcal S^\top$ is an $\omega$-relation clone.\end{itemize}
\end{proposition}

\begin{proof} We prove that $\mathcal R_{\mathrm{fin}}$ is a relation clone. 

(i)  $\Delta_A^{(n)}= \{ a^n: a\in A \}\in \mathcal R_{\mathrm{fin}} $, because $(\Delta_A^{(n)})^\top=\exists_{\omega\setminus n} \Delta_A\in  \mathcal R$.
 
(ii) Let $U\subseteq A^n$ and $S\subseteq A^k$ elements of $\mathcal R_{\mathrm{fin}}$. Then $U\times S\in \mathcal R_{\mathrm{fin}}$, because $(U\times S)^\top=U^\top \cap (S^\top(\sigma))$, where
$\sigma =(0,n)\cdots (k-1,n+k-1)$ is a product of transpositions.

(iii) Let $S\subseteq A^n$ be an element of $\mathcal R_{\mathrm{fin}}$ and $f: k\to n$ be a map. We recall that $\pi_f(S)=\{s\circ f: s\in S\}\subseteq A^k$.
Let $t > n+k$ and $\sigma$ be the permutation  $(0,t)\cdots(n-1,t+n-1)$. Then $S^\top(\sigma) = \{s\in A^\omega: (s_{t},\dots,s_{t+n-1})\in S\}$. Consider the following sequence of $\omega$-relations  $R_i\in \mathcal R,\ 0\leq i\leq k$:
$R_0= S^\top(\sigma)$,  $R_{i+1} = %R_i((i,t+f(i)))\cap
  R_i \cap \exists_{\omega\setminus\{i,t+f(i)\}} \Delta_A$.
Then $\pi_f(S)^\top = \exists_{\omega\setminus k}R_k$.

(iv) The closure under arbitrary intersection of finitary relations having the same arity is trivial.

Finally,  $\mathcal R_{\mathrm{fin}}$ is cut-closed, because $((\bigcap_{i\in I} S_i^\top)[n])^\top =  \exists_{\omega\setminus  n}\bigcap_{i\in I} S_i^\top$.
%Moreover, if $\mathcal R$ is strong, then it is easy to prove that  $\mathcal R_{\mathrm{fin}}$ is strong.
\end{proof}

%\begin{remark} \label{lem:finfin} Let $F\subseteq O_A$. Then by Lemma \ref{lem:finrel} we have that $Inv\, F= (Inv^\omega\, F)_{\mathrm{fin}}$ is a cut-closed relation clone, because 
% $f^\top$ is an $\omega$-polymorphism of $S^\top$ iff $f$ is a polymorphism of $S$.
%\end{remark}

In the following example, we show that $\mathcal R_{\mathrm{fin}}$ is not in general a strong relation clone (see Definition \ref{def:src}).

\begin{example}\label{exa:star}
 Let $A = \omega \cup \{\star\}$ and $\mathcal S$ be the chain of  $k$-ary relations  $S_n=\{s\in A^k : s_j\in n,\ \text{for every $0\leq j\leq k-1$}\}$ ($n\geq 1$).   We show that the cut-closed relation clone $(Inv_\mathcal G^\omega(Pol^\omega_\mathcal G\,\mathcal S^\top))_{\mathrm{fin}}$ is not strong.

 For all $s\in A^\omega$, we define $\mathrm{nat}(s) = \{s_i: s_i\in\omega\}$ and $\varphi: A^\omega\to A$ as follows:
 $$\varphi(s)=\begin{cases}\mathrm{max}(\mathrm{nat}(s))&\text{if the maximum of $\mathrm{nat}(s)$ exists}\\ \star &\text{otherwise.}\end{cases}$$
%  Let 
% $$\varphi(s)=\begin{cases}n &\text{if $\exists k$ s.t. $s_i=n$ for all $i\geq k$}\\ \star&\text{otherwise}
%\end{cases}$$
The $\omega$-operation $\varphi$ is a $\mathcal G$-polymorphism of each $S_n^\top$, but it is not a $\mathcal G$-polymorphism of $S^\top$, where $S=\bigcup S_n$. Indeed,
$\varphi[m]= \star^\omega$ for the  matrix $m\in \mathsf M^{\omega,S^\top}_\omega$ such that $m^j= j^\omega$. Therefore, $S_n\in (Inv_\mathcal G^\omega(Pol^\omega_\mathcal G\,\mathcal S^\top))_{\mathrm{fin}}$ for every $n\geq 1$, but  $\bigcup S_n\notin (Inv_\mathcal G^\omega(Pol^\omega_\mathcal G\,\mathcal S^\top))_{\mathrm{fin}}$.
\end{example}

%\begin{proposition}\label{prop:strongcut}
%Let $F$ be a set of finitary operations. Then $Inv(F)$ is cut-closed.
%\end{proposition}
%\begin{proof}
%  Let $S_i\in Inv(F) $ ($i\in I$), $T=\bigcap_{i\in I}S_i^\top$, $f\in F$ be a $k$-ary operation, and
%  $m\in \mathsf M_n^{k,T[n]}$. Then there exists a matrix $p\in \mathsf M_\omega^{k,T}$ such that 
%  $(p^j)_{|n}=m^j$ for all $0\leq j\leq k-1$.
%  We now prove that $f[p]\in T$. Indeed, given $i\in I$, let $r_i$ be the arity of $S_i$. Then $(p^j)_{|r_i}\in S_i$ for all $0\leq j\leq k-1$,  so that $f[p]_{|r_i}\in S_i$. Therefore, $f[p]\in S_i^\top$ for every $i\in I$ and $f[p]_{|n}\in T[n] $. Since $f[p]_{|n}=f[m]$, we get that $f[m]\in T[n]$, that is $T[n]\in Inv(F)$.
% \end{proof}
%As a consequence of Proposition \ref{prop:strongcut}, if $\mathcal S$ is a set of finitary relations, then 
%$Inv(Pol(\mathcal S))$ is a cut-closed, strong relation clone. Therefore, being cut-closed is a necessary condition for the equality $Inv(Pol(\mathcal S))=\mathcal S$ to hold.
%Example \ref{exa:star} illustrates a cut-closed relation clone
%$\mathcal R_{\mathrm{fin}}$, for a suitable $\omega$-relation clone $\mathcal R$,  such that 
%$\mathcal R_{\mathrm{fin}}\subsetneq Inv(Pol( \mathcal R_{\mathrm{fin}} )) $.
%\begin{corollary}\label{cor:corcor}
%  %\begin{itemize}
%    %\item  Every relation clone on a finite set is cut-closed.
%    Every strong relation clone on a countable set is cut-closed.
%  %  \end{itemize}
%\end{corollary}
%\begin{proof}
%  By Proposition \ref{prop:strongcut} and  Theorem % \ref{thm:invpol} and
%  \ref{thm:invpolromov}.
%  \end{proof}

The following proposition shows that the relation clones on finite sets are always cut-closed.
%is a direct consequence of Proposition \ref{prop:strongcut} and Theorem \ref{thm:invpol}.
%We provide a direct proof of this result, because of its independent interest.

%We provide an alternative, direct  proof that relation clones on finite sets are always cut-closed.
\begin{proposition} 
Every relation clone on a finite set is cut-closed.
\end{proposition}

\begin{proof} Let $\mathcal S$ be a relation clone on a finite set $A$, $S_i\in \mathcal S$ ($i\in I$) be a family of finitary relations, $R=\bigcap_{i\in I} S_i^\top$ and, for every $j\leq k$, $f_j^k$ be the map embedding $ j$ into $ k$.

Let $j\in\omega$ and $i\in I$. If $S_i$ has arity $k \leq j$, then the relation $S_{i,j}=S_i\times A^{j-k}\in \mathcal S$ has arity $j$. Similarly, if $S_i$ has arity $k > j$, then the relation $S_{i,j}=\pi_{f_j^k}(S_i)\in \mathcal S$ has arity $j$. We claim that for all $n\in \omega$
 $$R[n] =\bigcap_{j\geq n}[\pi_{f_n^j}(\bigcap_{i\in I} S_{i,j})].$$
 
Let $U=\bigcap_{j\geq n}[\pi_{f_n^j}(\bigcap_{i\in I} S_{i,j})]\in\mathcal S$ and  $s\in A^n$ such that $s\in U$. 
Consider the nonempty set $X_j$ of all $t\in A^j$ ($j\geq n$) such that $t_i=s_i$ for every $0\leq i\leq n-1$ and
$t\in \bigcap_{i\in I} S_{i,j}$. Let $X=\bigcup_{j\geq n} X_j$ be partially ordered by the prefix ordering. Then the infinite set $X$ is a finite branching tree, because $A$ is finite. Then by K\"onig Lemma there exists an infinite branch $B$. Since $B$ is totally ordered by prefix ordering, then we define $u\in A^\omega$ such that $u_i = t_i$ for some (and then all) $t\in B$ such that $i\in \mathrm{dom}(t)$.
We have that $u_{| n}=s$. We now prove that  $u\in R=\bigcap_{i\in I} S_i^\top$. Given $i\in I$, let
$k$ be the arity of $S_i$, and $j\geq \text{max}\{n,k\}$. Since $B\cap X_j\neq \emptyset$ then
$u_{|j}\in S_{i,j}=S_i\times A^{j-k}$. It follows that $u_{|i}\in S_i$ and hence $u\in S_i^\top$. By the arbitrariness of $i$, it follows that $u\in R$.
\end{proof}

\section{Limits of decreasing sequences}\label{sec:nine}
In order to characterise the $c\omega$-relation clones 
%$\mathcal R$ such that  $\mathcal R=Inv^\omega_c(Pol^\omega(\mathcal R))$ 
we define the notion of decreasing sequence of finitary relations. Each of these sequences has a locally closed $\omega$-relation as a limit.
The main results of the section are:

\begin{enumerate}
\item $\mathcal R$ is a $c\omega$-relation clone iff $\mathcal R =\mathrm{Lim}\,\mathcal S$, for some relation clone $\mathcal S$.
\item $\mathcal S$ is a cut-closed relation clone iff $\mathcal S=(\mathrm{Lim}\,\mathcal S)_{\mathrm{fin}}$.
\item
If $\mathcal R$ is a $c\omega$-relation clone, then 
  $Inv^\omega_{c}(Pol^\omega_\mathcal G\, \mathcal R)\subseteq \mathrm{Lim}\;Inv(Pol\;\mathcal R_{\mathrm{fin}}).$
\end{enumerate}

\begin {definition}\label{def:dec}
  An $\omega$-indexed sequence $\mathtt S$ of finitary relations on $A$ is \emph{decreasing} if for all $i\in\omega$
\begin{itemize}
\item $\mathtt S_i\subseteq A^i$ is a relation of arity $i$, and
\item $\mathtt S_i^\top \supseteq \mathtt S_{i+1}^\top$.
\end{itemize}
\end{definition}

Decreasing sequences of finitary relations will be denoted by $\mathtt R, \mathtt S, \mathtt T,\dots$. 
%If $\mathtt S$ is a decreasing sequence, then $S_i$ denotes the relation of arity $i$ in the sequence $\mathtt S$.

The \emph{limit of a decreasing sequence} $\mathtt S$  is the locally closed $\omega$-relation 
$$\mathrm{Lim} (\mathtt S)= \bigcap_{i\in \omega} \mathtt S_i^\top.$$

\begin{lemma}\label{lem:[k]}
 Let $R$ be an $\omega$-relation on $A$. Then we have:
 \begin{itemize}
\item[(i)]  $R[i]^\top=\exists_{\omega\setminus i}R$; 
\item[(ii)]  The sequence $\mathtt R=(R[i] :i\in\omega)$ is a decreasing sequence such that $\mathrm{Lim} (\mathtt R)=\overline R$.
%\item[(iv)] $\mathcal R = \{ \bigcap_{i\in \omega} S_i^\top :\ \text{$(S_i)_{i\in\omega}$ is a decreasing sequence of relations $S_i\in \mathcal R_{\mathrm{fin}}$}\}.$
\end{itemize}
\end{lemma}

\begin{proof}
  (i) Trivial.
  
  (ii) By Lemma \ref{lem:clofin},  $\bigcap_{i\in\omega} R[i]^\top$ is a locally closed $\omega$-relation containing $R$.
Let $s\in \bigcap_{i\in\omega} R[i]^\top$. Then, for every $i$ there exists $r\in R$ such that
 $s_{| i}= r_{| i}$. Therefore, $s\in \overline{R}$.
\end{proof}

The set of decreasing sequences of finitary relations on a set $A$ can be coordinate-wise partially ordered as follows:
$$\mathtt S \leq \mathtt T\ \text{iff}\ \forall i\in\omega.\ \mathtt S_i\subseteq \mathtt T_i.$$
If $\{\mathtt S^j\}_{j\in J}$ is a directed set of decreasing sequences $\mathtt S^j$, then 
$\bigvee_{j\in J} \mathtt S^j$ is the decreasing sequence such that $(\bigvee_{j\in J} \mathtt S^j)_i= \bigcup_{j\in J} (\mathtt S^j)_i$ for all $i\in\omega$.

\begin{lemma}
 Let  $\{\mathtt S^j\}_{j\in J}$ be a directed set of decreasing sequences. Then we have  $\mathrm{Lim}(\bigvee_{j\in J} \mathtt S^j) = \bigcup_{j\in J} \mathrm{Lim}(\mathtt S^j)$.
\end{lemma}

\begin{proof}
$\mathrm{Lim}(\bigvee_{j\in J} \mathtt S^j) =\bigcap_{i\in\omega} (\bigvee_{j\in J} \mathtt S^j)_i^\top = \bigcap_{i\in\omega}(\bigcup_{j\in J} (\mathtt S^j)_i)^\top=$\\
$\bigcap_{i\in\omega}\bigcup_{j\in J} (\mathtt S^j)_i^\top = \bigcup_{j\in J} \bigcap_{i\in\omega}(\mathtt S^j)_i^\top=
 \bigcup_{j\in J} \mathrm{Lim}(\mathtt S^j)
 $.
\end{proof}

 If $\mathcal S$ is a set of finitary relations, then we denote by $\mathrm{Dec}\,\mathcal S$ the set of all decreasing sequences $\mathtt S$ such that $\mathtt S_i\in \mathcal S$ for every $i$.
We define
$$\mathrm{Lim}\ \mathcal S =\{ \mathrm{Lim} (\mathtt S) : \mathtt S\in \mathrm{Dec}\,\mathcal S\}.$$

\begin{definition}\label{def:omegarc} An $\omega$-relation clone $\mathcal R$ is  \emph{strong} if, for every directed set $\{\mathtt S^j\}_{j\in J}$ of decreasing sequences $\mathtt S^j \in \mathrm{Dec}\,\mathcal R_{\mathrm{fin}}$, we have that $\mathrm{Lim}(\bigvee_{j\in J} \mathtt S^j)\in\mathcal R$.
\end{definition}

\begin{theorem}\label{thm:convergent} 
%Then $\mathrm{Lim}(\mathcal S)$ is an $\omega$-relation clone.
(1) $\mathcal R$ is a (strong) $c\omega$-relation clone iff $\mathcal R=\mathrm{Lim}\ \mathcal S$, for some (strong) relation clone $\mathcal S$.

(2) Let $\mathcal S$ be a relation clone. Then  $\mathcal S$ is  cut-closed iff $\mathcal S=(\mathrm{Lim}\ \mathcal S)_{\mathrm{fin}}$.
\end{theorem}

\begin{proof} (1) ($\Rightarrow$) Let  $\mathcal R$ be a  $c\omega$-relation clone. We prove that $\mathcal R=\mathrm{Lim}\, \mathcal R_{\mathrm{fin}}$, where by Proposition  \ref{lem:finrel} $\mathcal R_{\mathrm{fin}}$ is a cut-closed relation clone. 

  If $R\in\mathcal R$, then by Lemma \ref{lem:[k]} the sequence $\{R[n]\}_{n\in\omega}$ is decreasing,
each $R[n]^\top = \exists_{\omega\setminus  n}R \in \mathcal R$ and 
  $R=\overline R = \bigcap_{n\in\omega} R[n]^\top$, because $R$ is locally closed.
Therefore, $R\in \mathrm{Lim}\, \mathcal R_{\mathrm{fin}}$.

If $R\in \mathrm{Lim}\, \mathcal R_{\mathrm{fin}}$, then there exists a decreasing sequence $\mathtt S\in \mathrm{Dec}\, \mathcal R_{\mathrm{fin}}$ such that
$R=\bigcap_{n\in\omega} (\mathtt S_n)^\top$. Since $(\mathtt S_n)^\top\in \mathcal R$ and $\mathcal R$ is  closed under intersection, then $R\in\mathcal R$.

($\Leftarrow$) Let $\mathcal S$ be a relation clone on $A$. We prove that $\mathrm{Lim}\, \mathcal S$ is a  $c\omega$-relation clone.

(i) $\Delta_A = \bigcap_{n\in\omega}(\Delta_A^{(n)})^\top\in \mathrm{Lim}\, \mathcal S$,  because the decreasing sequence $(\Delta_A^{(n)}: n\in\omega)$ belongs to $\mathrm{Dec}\,\mathcal S$. 

  (ii) Let $R= \mathrm{Lim} (\mathtt S)$ with $\mathtt S\in \mathrm{Dec}\,\mathcal S$ and $\sigma$ be a permutation of $\omega$, whose inverse $\sigma^{-1}$ will be denoted by $\tau$.
    We define a decreasing sequence $\mathtt T\in \mathrm{Dec}\,\mathcal S$ such that $\mathrm{Lim} (\mathtt T)=R(\sigma)$.
  For $k\in\omega$, let $f(k)=\mathrm{max}\{ \tau(0),\ldots,\tau(k-1)\}$ and $\mathtt T_k=  \pi_{\tau_{| k}}(\mathtt S_{f(k)})=
  \{t\in A^k: \exists s\in \mathtt S_{f(k)}\mbox{ such that } t=s\circ \tau_{| k}\}$.
 %The fact that the $T_i$s are in $\mathcal S$ is clear, since $T_k= \pi_{\sigma^{-1}_{| k}}(S_{f(k)})$, see Definition \ref{def:src}(iii).  
 The fact that  $$(t_0,\ldots,t_k)\in \mathtt T_{k+1}\Rightarrow (t_0,\ldots,t_{k-1})\in \mathtt T_k,$$ 
 entails that $\mathtt T_k^\top\supseteq \mathtt T_{k+1}^\top$.
  As for  $\bigcap_{k\in \omega} (\mathtt T_k)^\top=R(\sigma)$:

  $$\begin{array}{l}
    R(\sigma)=\{t\in A^\omega: (t_{\sigma(0)},\ldots,t_{\sigma(n)},\ldots)\in   \bigcap_{k\in \omega} (\mathtt S_k)^\top\}=\\
    \{t\in A^\omega: \forall k\ (t_{\sigma(0)},\ldots,t_{\sigma(k-1)})\in \mathtt S_k\}=\\
    \{t\in A^\omega: \forall k\ (t_0,\ldots,t_{k-1})\in  \pi_{\tau_{| k}}(\mathtt  S_{f(k)})=\mathtt  T_k\}=\\
     \bigcap_{k\in \omega} (\mathtt T_k)^\top.
    \end{array}
  $$
  
  %, notice that $t\in \bigcap_{i\in \omega} T_i^\top$ iff there exists $s=(s_0,\ldots,s_n,\ldots)\in  \bigcap_{i\in \omega} S_i^\top$ such that
 % $t_i=s_{\sigma^{-1}(i)}$ for all $i$, hence $(t_{\sigma(0)},\ldots,t_{\sigma(n)},\ldots)=s\in \bigcap_{i\in \omega} S_i^\top$, and we are done.

(iii) Let $R_i\in \mathrm{Lim}(\mathtt S^i)$ ($i\in I$) with $\mathtt S^i\in \mathrm{Dec}\,\mathcal S$. Then we show that $\mathrm{Lim}(\bigwedge_{i\in I} \mathtt S^i) = \bigcap_{i\in I} R_i $.
  Indeed, we have:
  $$\bigcap_{i\in I} R_i=\bigcap_{i\in I}\bigcap_{n\in\omega}((\mathtt S^i)_n)^\top=\bigcap_{n\in\omega}\bigcap_{i\in I}((\mathtt S^i)_n)^\top = \bigcap_{n\in\omega}[\bigcap_{i\in I} (\mathtt S^i)_n]^\top$$

  (iv) Let $R\in \mathrm{Lim}(\mathtt S)$ with $\mathtt S\in \mathrm{Dec}\,\mathcal S$ and $\Gamma\subseteq_{\mathrm{cofin}}\omega$. We claim that $\exists_\Gamma R= \bigcap_{n\in\omega}(\exists_{\Gamma\cap n}\mathtt S_n)^\top$, where $(\exists_{\Gamma\cap n}\mathtt S_n)_{n\in\omega}$ is a decreasing sequence of relations in $\mathcal S$.

 Moreover, if $\mathcal S$ (resp. $\mathcal R$) is strong, then it is easy to prove that $ \mathrm{Lim}\, \mathcal S$ (resp. $\mathcal R_{\mathrm{fin}}$) is strong.

 (2) ($\Rightarrow$) Let $T \in (\mathrm{Lim}\ \mathcal S)_{\mathrm{fin}}$ of arity $k$. Then there exists a decreasing sequence $\mathtt S\in \mathrm{Dec}\,\mathcal S$ such that $T^\top=\bigcap_{n\in\omega}
 \mathtt S_n^\top$. Then $T= (\bigcap_{n\in\omega}
 \mathtt S_n^\top)[k] \in\mathcal S$, because $\mathcal S$ is cut-closed. The inclusion $\mathcal S\subseteq (\mathrm{Lim}\ \mathcal S)_{\mathrm{fin}}$ is trivial.

 ($\Leftarrow$) By (1) $\mathrm{Lim}(\mathcal S)$  is a $c\omega$-relation clone  and by Proposition \ref{lem:finrel}  $(\mathrm{Lim}\ \mathcal S)_{\mathrm{fin}}$ is cut-closed. 
 \end{proof}

%If $\mathcal S$ is a set of finitary relations, then $\mathcal S^\top=\{ S^\top: S\in \mathcal S\}$.

\begin{lemma}\label{cor:pol2fin}
  For every $c\omega$-relation clone $\mathcal R$,
  $Pol^\omega\mathcal R=Pol^\omega(\mathcal R_{\mathrm{fin}}^\top)$.

\end{lemma} 

\begin{proof} 
Since $\mathcal R_{\mathrm{fin}}^\top \subseteq \mathcal R$, then $Pol^\omega\mathcal R\subseteq Pol^\omega(\mathcal R_{\mathrm{fin}}^\top)$. For the opposite inclusion, let $\varphi\in Pol^\omega(\mathcal R_{\mathrm{fin}}^\top)$ and $R\in\mathcal R$. 
By Lemma \ref{lem:[k]}(ii) $R= \bigcap_{i\in \omega} R[i]^\top$. Since $\varphi$ is a $\mathcal G$-polymorphism of each $R[i]^\top$, then it is also a $\mathcal G$-polymorphism of the intersection $R=\bigcap_{i\in \omega} R[i]^\top$. %The conclusion follows because $Pol^\omega_\mathcal G(\mathcal R_{\mathrm{fin}}) = Pol^\omega_\mathcal G(\mathcal R_{\mathrm{fin}}^\top)$.
\end{proof}

The following lemma directly follows from that definition
of $Inv^\top$ in Section \ref{sec:clomega}. 

\begin{lemma}\label{lem:last}
Let $C$ be a set of $\omega$-operations. Then we have $ ([Inv^\omega_c(C)]_{\mathrm{fin}})^\top=Inv^\top(C)$.
\end{lemma}

Notice that, if $\mathcal R$ and $\mathcal T$ are two  $c\omega$-relation clones such that $\mathcal R_{\mathrm{fin}}=\mathcal T_{\mathrm{fin}}$, then $\mathcal R=\mathrm{Lim}\,\mathcal R_{\mathrm{fin}} =\mathrm{Lim}\,\mathcal T_{\mathrm{fin}} =\mathcal T$.

\begin{theorem}\label{thm:clone}

  Let    $\mathcal R$ be  a $c\omega$-relation clone. 
  Then  $$ Inv_{c} ^\omega(Pol^\omega\mathcal R)\subseteq\mathrm{Lim}\, Inv(Pol\,\mathcal R_{\mathrm{fin}}). $$
 The opposite inclusion is  in general false.

\end{theorem}

\begin{proof} By Lemma \ref{lem:pp} and Theorem \ref{thm:convergent} $Inv_{c} ^\omega(Pol^\omega\, \mathcal R)$ and $\mathrm{Lim}\, Inv(Pol\,\mathcal R_{\mathrm{fin}})$ are $c\omega$-relation clones. 
By Lemmas \ref{cor:pol2fin} and \ref{lem:last} we have
$$([Inv_{c}^\omega(Pol^\omega\, \mathcal R)]_{\mathrm{fin}})^\top= ([Inv^\omega_{c} (Pol^\omega\,\mathcal R_{\mathrm{fin}}^\top)]_{\mathrm{fin}})^\top=Inv^\top (Pol^\omega\,\mathcal R_{\mathrm{fin}}^\top).$$
Since by Proposition  \ref{lem:finrel}  $[\mathrm{Lim}\, Inv(Pol\,\mathcal R_{\mathrm{fin}})]_{\mathrm{fin}}$ is cut-closed, then by Theorem \ref{thm:convergent} we have that $[\mathrm{Lim}\, Inv(Pol\,\mathcal R_{\mathrm{fin}})]_{\mathrm{fin}}=Inv(Pol\,\mathcal R_{\mathrm{fin}})$. Therefore, the conclusion follows from $(Pol\,\mathcal R_{\mathrm{fin}})^\top\subseteq  Pol^\omega\,\mathcal R_{\mathrm{fin}}^\top$ and the contravariance of $Inv^\top$.

The following example, borrowed from \cite{BHM12}, is a counterexample to the opposite inclusion. 
%$$Inv(Pol(\mathcal S))\not\subseteq(Inv^\omega_c(Pol^\omega(\mathcal S^\top)))_{\mathrm{fin}}.$$
 %$Inv(Pol^\omega(\mathcal S)) \cap \langle \mathcal S\rangle_{\mathrm{fo}} = \langle \mathcal S\rangle_{\mathrm{rc}}$.
Let $\mathcal S$ be the set of unary relations $S_i=\omega\setminus\{0,i\}$, $i\geq 1$. Consider the relation $P=S_1\cup S_2=\omega\setminus\{0\}$ and the $\omega$-operation $\varphi:\omega^\omega\to\omega$ defined as follows:
  $\varphi(s)=0$ if $set(s) = \omega\setminus\{0\}$;  $\varphi(s)=s_0$ in all other cases. 
  %Let $\langle \mathcal S\rangle_{rc}$ be the least relation clone containing $\mathcal S$.
  Then we have   $P\in Inv(Pol( \mathcal S))$  but $P^\top \notin Inv^\top(Pol^\omega(\mathcal S))$,
  because  $\varphi\in Pol^\omega(\mathcal S)$ is not a $\mathcal G$-polymorphism of $P^\top$.
  Then the counterexample is obtained by letting $\mathcal R$ to be the least $c\omega$-relation clone including $\mathcal S^\top$.
% \begin{itemize}
%\item $\varphi\in Pol^\omega(\mathcal S^\top)$ and $P^\top\notin Inv^\omega(Pol^\omega(\mathcal S^\top))$, because $\varphi$ is not an $\omega$-polymorphism of $P^\top$.
%\item Since $P\in Inv(Pol( \mathcal S))$ then  $Inv(Pol( \mathcal S))\neq  \langle \mathcal S\rangle_{rc}$.
%\end{itemize} 
\end{proof}

Let $\mathcal R$ be a 
$c\omega$-relation clone.
Since $$\mathcal R\subseteq Inv_{c} ^\omega(Pol^\omega \mathcal R)\subseteq\mathrm{Lim}\, Inv(Pol\,\mathcal R_{\mathrm{fin}})$$
then the equality $Inv(Pol\,\mathcal R_{\mathrm{fin}})=\mathcal R_{\mathrm{fin}}$ implies $Inv_{c}^\omega(Pol^\omega\,  \mathcal R)=\mathcal R$. As a consequence of this reasoning  we get   the following corollary.

\begin{corollary}
  If   $\mathcal R$ is either a $c\omega$-relation clone on a finite set or a strong $c\omega$-relation clone on a countable set, then 
 $\mathcal R=Inv_{c} ^\omega(Pol^\omega\,\mathcal R)$. 
\end{corollary}
\begin{proof} By Theorems \ref{thm:clone}, \ref{thm:convergent}(1), \ref{thm:invpol} and \ref{thm:invpolromov}.
    \end{proof}

\section{Conclusion}
In this paper, we examine various  concepts of clones of operations and relations on a given base set $A$. Unlike classic clone theory, which limits the arities of functions and relations to be finite, our study allows for arity $\omega$  for both operations and relations. Additionally, there are no restrictions on the cardinality of the base set $A$.
Since  there are still restrictions in arity,  topology plays an important role in
this setting.
In Section \ref{sec:topo} we develop a method for equipping the  set $A^B$ with a topology using Boolean ideals on $B$. 
This framework is particularly applicable when $B = A^\omega$, hence $A^B = O_A^{(\omega)}$. However, not all ideals on $A^\omega$ are useful in the context of clone theory. It is essential for the topologies to behave well with respect to composition. To identify such suitable ideals, we introduce the concepts of substitutive and infinitely substitutive ideals (cf. Definition \ref{def:subst} and Proposition \ref{prop:sub}).
Of special interest are the local, global, trace, and uniform ideals on $A^\omega$, which lead to corresponding topologies on $O_A^{(\omega)}$. The local topology has been significant in clone theory since its inception, while the other three, though less known, are important when dealing with arity $\omega$.

On the relational side, topology is introduced through the concept of ideal maps.
An ideal map on a set $A$ associates an ideal on $\alpha$ to every $(\alpha \times \omega)$-matrix  over $A$, thereby equipping $A^\alpha$ with a topology.
Each  ideal map $\mathcal{I}$ on $A$ induces
a Galois connection between $\mathcal{P}(O_A^{(\omega)})$ and $\mathcal{P}(Rel_A^{(\leq \omega)})$. However, the general notion of an ideal map is somewhat broad. We have determined that the well-behaved ideal maps are the canonical ones,  derived from ideals on $A^\omega$.
Given an ideal $X$ on $A^\omega$, one of the main results of the paper is the characterisation of $X$-closed infinitary $\omega$-clones, as stated in Theorem \ref{thm:duedue2}.
To describe $\omega$-clones that are not necessarily infinitary through invariant relations, in Section \ref{sec:matrical} we introduce the notion of matrical polymorphisms. The main result characterising $X$-closed $\omega$-clones is presented in Theorem \ref{thm:duedue5}. As a corollary, we obtain a characterisation of trace-closed and uniform-closed $\omega$-clones.

Sections \ref{sec:eight} and \ref{sec:nine} are dedicated to studying clones of relations of arity $\leq \omega$. The notions of $\omega$-relational clone (cf. Definition \ref{def:omegarc}) and relational clone are compared and connected through the concept of cut-closedness (cf. Definition \ref{def:cutc}). We demonstrate that the set of finitary relations, whose top-extensions belong to a given $\omega$-relational clone, constitutes a cut-closed relational clone. Special attention is given to $\omega$-relational clones whose relations are locally closed (termed $c\omega$-relational clones). The main results regarding these clones are Theorem \ref{thm:convergent} and Theorem \ref{thm:clone}, where cut-closed relation clones and $c\omega$-relation clones are characterised. The key concept in this context is the limit of decreasing sequences of finitary relations (cf. Definition \ref{def:dec}).

Some open questions and material for further work are analyzed below.

\begin{enumerate}

   \item Consider the set $\mathrm{RC}$ (resp. $c\omega\mathrm{RC}$) of relation clones (resp. $c\omega$-relation clones) on a given set $A$, partially ordered by $\subseteq$. In Theorem \ref{thm:convergent}, we defined two monotonic functions
   \[
   \mathrm{Lim}: \mathrm{RC} \to c\omega\mathrm{RC} \quad \text{and} \quad (-)_{\mathrm{fin}}: c\omega\mathrm{RC} \to \mathrm{RC}
   \]
   such that
      $\mathrm{Lim} \circ (-)_{\mathrm{fin}} = Id_{c\omega\mathrm{RC}}$  and $(-)_{\mathrm{fin}} \circ \mathrm{Lim}$
   is a closure operator, whose closed elements are the cut-closed relation clones. %We conjecture that this closure operator is not trivial, i.e., that there exist (strong) relation clones which are not cut-closed. Moreover,
   We conjecture that a strong relation clone $\mathcal{S}$ is cut-closed iff $\mathcal{S} = \mathrm{Inv}(\mathrm{Pol}\; \mathcal{S})$. %(the “only if” part follows from Proposition \ref{prop:strongcut}).

   \item We have studied $\mathrm{Inv}^\omega_\mathcal{X}$-$\mathrm{Pol}^\omega_\mathcal{X}$ in the case $\mathcal{X}=\mathcal{G}$. This is the simplest case since for all matrices $m \in \mathrm{M}_\omega^\omega$, the topology $\mathcal{G}_m$ on $A^\omega$ is the discrete one. The general case is left for future work.

   \item In Proposition \ref{prop:can}, we have shown that an ideal map is canonical iff it is the object part of a presheaf from $\mathbb{M}^{\mathrm{op}}$ into $\mathbf{Sets}$. We conjecture that this presheaf is indeed a sheaf with respect to the Grothendieck topology defined in Section \ref{sec:int}.

   \item In the last theorem of this paper, we have shown that, given a $c\omega$-relation clone $\mathcal{R}$,
   \[
   \mathcal{R} \subseteq \mathrm{Inv}_c^\omega(\mathrm{Pol}^\omega \mathcal{R}) \subseteq \mathrm{Lim} \, \mathrm{Inv}(\mathrm{Pol} \, \mathcal{R}_{\mathrm{fin}}).
   \]
   These inclusions become equalities when $\mathcal{R}_{\mathrm{fin}} = \mathrm{Inv}(\mathrm{Pol} \, \mathcal{R}_{\mathrm{fin}})$, but this is not always the case. The challenge is to identify a more general condition on $\mathcal{R}_{\mathrm{fin}}$ that ensures these inclusions are indeed equalities. In particular, if $\mathcal{R} = \mathrm{Lim} \, \mathcal{S}$ for a cut-closed strong relation clone $\mathcal{S}$, we conjecture that the equality
   $\mathrm{Lim} \, \mathcal{S} = \mathrm{Inv}_c^\omega(\mathrm{Pol}^\omega \, \mathrm{Lim} \, \mathcal{S})$ 
   implies that $\mathrm{Inv}(\mathrm{Pol} \, \mathcal{S}) = \mathcal{S}$.
\end{enumerate}

We conclude by sketching some work in progress. The underlying idea is to replace the operations $q_n$ of clone algebras in the finitary case and $q$ in the infinitary case with a binary operation of composition.
For instance, given an infinitary clone algebra $(C, q, \epsilon_i)_{i \in \omega}$, we consider the algebra
$(C^\omega, \cdot, 1)$, defined by
$$
a \cdot b = (q(a_0, b), q(a_1, b), \ldots, q(a_n, b), \ldots); \quad 1 = (\epsilon_0, \epsilon_1, \ldots, \epsilon_n, \ldots),
$$

\noindent which turns out to be a monoid, whose universe is the set $C^\omega$ of $\omega$-sequences. In the case of functional clone algebras, this amounts to replacing functions from $A^\omega$ into $A$ with functions from $A^\omega$ into $A^\omega$.

Based on this idea, we are working on a new and more abstract theory of clones that encompasses both $\omega$-clones and infinitary $\omega$-clones into a variety of finitary algebras, called cl-monoids. In these algebras, the monoidal structure is enriched with the algebraic abstraction of the concept of $\omega$-sequence.

%through an appropriate notion of merge-algebra.

%The multipolymorphisms transform $\omega \times \omega$ matrices whose columns belong to a given relation $R$ into matrices of the same type.

%{\bf Da chiarire}:

%1) trovare un relation clone non cut-closed

%Scaletta:

%1) Def. di   $\omega$,$c\omega$-relation clone. Se $\mathcal R$ é un $\omega$-rel. clone
%allore $\mathcal R_{fin}$ é un cut-closed relation clone (in generale non strong)

%2) Def. di cut-closed. $Inv(F)$ é sempre un cut closed strong relation clone, per $F$ insieme di op. fin.

%3) Successioni decrescenti. Il Lim di una  successione decrescente é una
%$\omega$-relazione. Le successioni decrescenti formano un cpo.
%$Lim(\bigvee S^j)=\bigcup Lim(S^j)$.

%4) Def. di strong $c\omega$-relation clone. Thm. fond.:
 %$\mathcal R$ is a (strong) $c\omega$-relation clone iff
%$\mathcal R=Lim\ S$ for some strong relation clone $S$.

%5) A relation clone $S$ is cut-closed iff $S=(Lim \ S)_{fin}$.

%6) $R$ is a strong $c\omega$-relation clone iff $R=Inv^\omega_G(Pol^\omega_G(R))$.

%due vie: o come conseguenza di thm 9.8 e del thm di Szabo $Inv(Pol(S))=S$ per ogni ``strong relation clone'' alla Szabo.
%Opppure prova diretta, adattando quella di Szabo.

\end{document}